\documentclass[preprint,12pt,authoryear]{elsarticle}

\usepackage[T1]{fontenc}
\usepackage{lmodern}

\usepackage{amsmath,amsfonts,amssymb,amsthm,mathtools}
\usepackage{array}
\usepackage{xurl}
\usepackage{graphicx}
\usepackage{adjustbox}
\usepackage{enumitem}
\usepackage{tabularx, booktabs, ragged2e}
\newcolumntype{L}{>{\RaggedRight\arraybackslash}X}
\usepackage{tikz}
\usetikzlibrary{arrows.meta,backgrounds,calc,decorations.pathreplacing,
  fit,positioning,shapes.geometric}
\tikzset{
  ppi/.style={>={Stealth[length=2.2mm,width=1.6mm]},font=\footnotesize,
    line cap=round},
  leaf/.style={circle,draw,thick,fill=blue!8,inner sep=1pt,
    minimum size=17pt},
  inter/.style={rectangle,draw,thick,rounded corners=2pt,fill=orange!12,
    inner sep=2pt,minimum size=15pt},
  cached/.style={inter,fill=orange!55,very thick},
  query/.style={rectangle,draw,thick,rounded corners=2pt,double,
    inner sep=2.5pt,minimum size=15pt},
  exposed/.style={draw=red!70!black,very thick,fill=red!12},
  heavy/.style={double,double distance=0.8pt},
  arc/.style={->,thick,black!50},
  free/.style={->,very thick,red!70!black},
  region/.style={draw=black!50,dashed,rounded corners=4pt},
  note/.style={font=\scriptsize,text=black!75,align=center},
  andnode/.style={circle,fill=black,inner sep=1.4pt},
}
\usepackage[hidelinks]{hyperref}

\journal{Theoretical Computer Science}


\normalfont

\DeclareFontShape{T1}{lmr}{bx}{sc}{<->ssub*lmr/bx/n}{}
\DeclareFontShape{T1}{lmr}{m}{scit}{<->ssub*lmr/m/it}{}

\theoremstyle{plain}

\newtheorem{theorem}{Theorem}[section]
\newtheorem{lemma}{Lemma}[section]
\newtheorem{proposition}{Proposition}[section]
\newtheorem{corollary}{Corollary}[section]
\theoremstyle{definition}
\newtheorem{assumption}{Assumption}[section]
\newtheorem{definition}{Definition}[section]

\theoremstyle{remark}
\newtheorem{remark}{Remark}[section]

\numberwithin{equation}{section}

\DeclareMathOperator{\Cn}{Cn}
\newcommand{\E}{\mathbb{E}}

\DeclareMathOperator{\Enc}{Enc}
\DeclareMathOperator{\Dec}{Dec}

\begin{document}

\begin{frontmatter}

\title{Partial Path Interception in Derivation DAGs: Structural
Complexity, Parameterized Algorithms, and
Approximation}

\author[sjtu]{Jianfeng Xu\corref{cor1}}
\ead{xujf@sjtu.edu.cn}
\cortext[cor1]{Corresponding author.}

\affiliation[sjtu]{organization={Koguan School of Law, China Institute for Smart Justice, School of Computer Science, Shanghai Jiao Tong University},
  city={Shanghai}, postcode={200030}, country={China}}

\begin{abstract}
Partial path interception (PPI) asks for a minimum-cost vertex
set whose removal from a directed acyclic graph leaves at most
\(N^*\) designated terminals able to reach a protected sink.
We study how costs, depth, and fan-out affect its complexity.
Unit-cost PPI is NP-complete at depth two with exactly two paths
per terminal.
Unique-path and individual-coverage instances admit polynomial-time
algorithms, and bounded treewidth yields fixed-parameter tractability.
Parameterized by \(N^*\), weighted PPI is in XP and W[1]-hard
already at depth two, with no
\(f(N^*)|I|^{o(N^*)}\)-time algorithm under ETH.
At depth two, undeletable terminals with fan-out at most
\(p\) admit an FPT algorithm parameterized by \(N^*+p\).
For unit costs, a dual-gammoid rank identity characterizes the
optimum.
Unit-cost depth-two instances with intermediate in-degree at most two are
linear-time solvable for every \(N^*\), whereas depth-three
instances are W[1]-hard even with terminal out-degree one.
On general DAGs, minimum cuts and submodular optimization give a
randomized \(\min\{t,\widetilde O(\sqrt n)\}\)-approximation,
where \(t\) is the coverage demand and \(n\) the number of initially
exposed terminals.
Reductions from smallest \(T\)-edge subgraph, densest subgraph,
and minimum union transfer conditional approximation lower bounds.
PPI also gives an exact formulation of proof-valid caching under
conjunctive derivation semantics: independent premise erasures yield
recovery probability \((1-\varepsilon)^e\), where \(e\) is the
number of exposed premises.
We extend this formulation to shared workloads, derive a
sorting-based optimum under exact module routing, and give an
AND--OR reliability envelope and a separate coded
payload-recovery benchmark.
\end{abstract}

\begin{keyword}
Partial path interception \sep derivation DAGs \sep proof-valid
caching \sep approximation algorithms \sep parameterized complexity
\MSC[2020] 68Q25 \sep 68R10 \sep 90C27 \sep 05C85
\end{keyword}

\end{frontmatter}
\section{Introduction}
\label{sec:introduction}

Systems that reuse derived objects often retain their supporting
dependencies, as in deductive databases, provenance-aware query
processing, and auditable reasoning
\citep{abiteboul1995foundations,ullman1988database,
green2007provenance,cheney2009provenance,sikos2020provenance,
liu2026algebraic}.
When premises may be lost, a designer can store selected
consequences in advance and use them together with the surviving
premises.
The resulting optimization differs from ordinary result caching:
an internal consequence may replace several dependencies, while
shared consequences may support multiple queries.

Under conjunctive derivation semantics, this optimization has an
exact graph formulation.
Orient dependency edges from premises toward derived objects and
let \(q\) be the query.
A source premise remains required after caching \(S\) precisely
when it has an \(S\)-avoiding path to \(q\).
If premises are independently erased with probability
\(\varepsilon\), recovery therefore has probability
\[
  (1-\varepsilon)^{|\mathrm{Exp}(S)|},
\]
where \(\mathrm{Exp}(S)\) is the set of exposed source premises.
A reliability target \(1-\delta\) is equivalent to
\[
  |\mathrm{Exp}(S)|
  \le
  N^*(\varepsilon,\delta)
  :=
  \left\lfloor
    \frac{\ln(1-\delta)}{\ln(1-\varepsilon)}
  \right\rfloor.
\]
Thus a stochastic storage constraint becomes a deterministic
partial separation problem.
Figure~\ref{fig:pipeline} illustrates this reduction on a small example.

\begin{figure}[tbp]
\centering
\begin{adjustbox}{max width=\linewidth}
\begin{tikzpicture}[ppi]
\node[leaf]         (b1) at (0,0)   {$b_1$};
\node[leaf,exposed] (b2) at (1.2,0) {$b_2$};
\node[leaf,exposed] (b3) at (2.4,0) {$b_3$};
\node[leaf]         (b4) at (3.9,0) {$b_4$};
\node[leaf]         (b5) at (5.1,0) {$b_5$};
\node[cached] (u1) at (0.6,1.6) {$u_1$};
\node[inter]  (u2) at (1.8,1.6) {$u_2$};
\node[cached] (u3) at (4.5,1.6) {$u_3$};
\node[query]  (q)  at (2.55,3.2) {$q$};
\draw[arc] (b1)--(u1); \draw[arc] (b2)--(u1);
\draw[arc] (b4)--(u3); \draw[arc] (b5)--(u3);
\draw[arc] (u1)--(q);  \draw[arc] (u3)--(q);
\draw[free] (b2)--(u2); \draw[free] (b3)--(u2); \draw[free] (u2)--(q);
\node[note,anchor=north] at (2.55,-0.45)
  {premises erased independently with probability \(\varepsilon\)};
\node[draw=black!40,rounded corners=3pt,fill=black!3,anchor=west,
  align=left,font=\footnotesize,inner sep=5pt] at (6.0,1.6)
  {cache \(S=\{u_1,u_3\}\), unit costs\\[2pt]
   \(\mathrm{Exp}(S)=\{b_2,b_3\}\) (red paths)\\[2pt]
   \(\Pr[q\ \text{recovered}]=(1-\varepsilon)^{2}\)\\[2pt]
   \(\varepsilon=0.1,\ \delta=0.2\ \Rightarrow\ N^*=2\)\\[2pt]
   leaf-only optimum: \(\kappa_q-N^*=3\)};
\end{tikzpicture}
\end{adjustbox}
\caption{A canonical derivation DAG under unit costs; each derived
object requires all of its parents.  The cache \(\{u_1,u_3\}\) blocks
\(b_1,b_4,b_5\), but \(b_2\) stays exposed through the cache-free path
\(b_2\to u_2\to q\), although its child \(u_1\) is cached.  Hence \(q\)
is recovered exactly when \(b_2\) and \(b_3\) survive
(Theorem~\ref{thm:exact-residual-law}).  Every admissible singleton
cache leaves at least three leaves exposed, whereas the displayed cache
leaves only two.  It is therefore optimal for \(N^*=2\)
(Theorem~\ref{thm:exact-single-query-opt}).}
\label{fig:pipeline}
\end{figure}

\subsection{Problem and contributions}
\label{subsec:problem-statement}
\label{subsec:contributions}

We study \emph{partial path interception} (PPI).
Given a DAG, a protected sink \(q\), designated terminals,
nonnegative vertex costs, and a residual threshold \(N^*\),
the objective is to remove a minimum-cost vertex set so that
at most \(N^*\) terminals retain a path to \(q\).
Terminals may themselves be removed.

The central question is how the cost model and dependency
structure change the complexity of this problem.
Our results establish the following boundaries.

\begin{enumerate}[leftmargin=6mm]
\item
\emph{Structural and parameterized complexity.}
Canonical unique-path instances and source-terminal
instances satisfying individual coverage are
polynomial-time solvable.
PPI is also fixed-parameter tractable in treewidth.
Unit-cost instances are NP-complete already at depth two with
exactly two paths per leaf.
For the residual parameter \(N^*\), weighted PPI is in XP
and W[1]-hard, with no
\(f(N^*)|I|^{o(N^*)}\)-time algorithm under ETH.
At depth two, bounded leaf fan-out yields FPT in \(N^*\) and the
fan-out bound when leaves are forbidden.

\item
\emph{The effect of unit costs.}
A dual-gammoid rank identity characterizes the unit-cost optimum.
At depth two, intermediate in-degree at most two permits
linear-time optimization for every residual threshold.
This includes the graph class used for the weighted hardness
result.
At depth three, unit-cost PPI is W[1]-hard in \(N^*\)
even with leaf out-degree one.
The general depth-two unit-cost parameterization is not
resolved here.

\item
\emph{Approximation.}
The minimum cost of separating a terminal subset is submodular.
Combined with single-terminal cuts, this gives a randomized
\(\min\{t,\widetilde O(\sqrt n)\}\)-approximation, where \(t\)
is the coverage demand and \(n\) the number of initially exposed
terminals.
Approximation-preserving embeddings of smallest \(T\)-edge
subgraph and minimum union transfer conditional lower bounds.

\item
\emph{Reliability and storage extensions.}
For a shared workload, joint recovery is governed by the union
of its exposed-leaf sets.
Exact module routing admits a sorting-based optimum.
For alternative derivations, an AND--OR witness envelope replaces
the single exposed-set law.
A separate coded benchmark quantifies storage for full recovery
of the relevant premise payloads.
\end{enumerate}

\subsection{Related work}
\label{subsec:related-work}

\paragraph*{Partial cuts and interdiction}
PPI is a partial multi-terminal, single-sink vertex-cut
problem.
It is related to network interdiction and directed separation
\citep{israeli2002shortest,vazirani2001approximation,
garey1979computers}.
The closest partial-cut formulation is partial one-target cut,
which separates a prescribed number of sources from one target
in an undirected graph
\citep{berge2019separating}.
Its parameterized complexity depends on the deletion model,
cut budget, and separation demand.
Here we study directed acyclic dependency graphs, allow terminal
deletion, and distinguish the residual threshold from the
coverage demand.
Depth, path multiplicity, fan-out, and unit costs yield the
structural boundaries developed below.
Related parameterized cut problems constrain a terminal region
and its boundary
\citep{lokshtanov2013clustering,fomin2013cutting};
in PPI, the residual threshold bounds only the number
of designated terminals on the sink side.

Single-vertex interception is characterized by postdomination
\citep{lengauer1979fast}.
Several vertices may nevertheless intercept all paths from a
terminal without any selected vertex individually postdominating
it.
This distinction separates the exact postdominator algorithm
from the general problem.

\paragraph*{Covering and dense subgraphs}
Classical set cover and partial covering provide approximation
frameworks
\citep{johnson1974approximation,lovasz1975ratio,
chvatal1979greedy,slavik1997improved,gandhi2004approximation,
konemann2011unified,koufogiannakis2013greedy}.
Unlike ordinary coverage, PPI protects a terminal only
after every path to the sink is intercepted.
Our edge-incidence constructions connect this requirement to
induced-edge optimization
\citep{feige2001dense,khot2006ruling}.
The reductions from smallest \(T\)-edge subgraph and minimum
union make these approximation connections explicit.
For general DAGs, we instead use the submodular function given
by minimum separation cost.

\paragraph*{Caching and dependency representations}
Demand paging concerns replacement decisions
\citep{belady1966study,megiddo2003arc}, while coded caching
optimizes placement and delivery
\citep{maddah2014fundamental,wan2021fundamental,lu2026demand}.
Semantic database caching reuses results through descriptions
of query regions \citep{dar1996semantic}.
Our objective is to preserve derivability after premise loss
by storing explicit consequences.
Deductive databases and provenance supply the dependency
representations
\citep{abiteboul1995foundations,ullman1988database,
dantsin2001complexity,green2007provenance,cheney2009provenance}.
The canonical model uses one conjunctive parent set per derived
object; recording a preferred proof in a richer deductive
system is not sufficient to establish this semantics.
Alternative derivations are treated separately in
\ref{app:envelope}.

\paragraph*{Coded comparison}
The payload-recovery benchmark uses packet-erasure coding,
MDS recovery, and finite-blocklength counting
\citep{shannon1948mathematical,cover2006elements,
polyanskiy2010channel,singleton1964maximum,
macwilliams1977theory}.
It requires recovery of all relevant premise payloads and is
therefore distinct from the query-recovery optimization.

\paragraph*{Companion manuscript}
A companion manuscript \citep{xu2026proofvalid} develops the
empirical framework.
It shares the erasure model and applies selected theoretical
results as benchmark references.
The present paper contains the structural complexity results,
algorithms, reductions, and proofs; the empirical datasets,
systems, and evaluations belong to the companion manuscript.

\subsection{Organization}
\label{subsec:organization}

Section~\ref{sec:model} defines the model and PPI.
Section~\ref{sec:single-query-kernel} proves the exact
single-query reliability law.
Section~\ref{sec:complexity-classification} develops the
complexity and approximation classification.
Section~\ref{sec:semantic-scenario} treats shared workloads
and exact module routing.
Section~\ref{sec:discussion} concludes.
\ref{app:envelope},
\ref{sec:coded-outer}, and
\ref{app:treewidth-transitions} contain the AND--OR extension,
coded comparison, and treewidth transition equations,
respectively.

\section{Model and Preliminaries}
\label{sec:model}

We define the conjunctive dependency model, erasure process,
and transparent cache classes, then isolate PPI
as an explicit graph optimization problem.
The coded payload-recovery benchmark is defined separately
in \ref{sec:coded-outer}.

\subsection{Deductive Setting}
\label{subsec:deductive-setting}

Let \(\mathcal U\) be a universe of finitely encoded objects,
let \(\Cn:2^{\mathcal U}\to2^{\mathcal U}\) be a deductive closure
operator, and fix
\[
  B=\{b_1,\ldots,b_m\},
  \qquad m\ge1,
  \qquad
  \mathcal Q\subseteq\Cn(B)
\]
with \(\mathcal Q\) finite.
Premise identifiers, rules, and the designated-parent map are
fixed metadata; an identifier alone does not supply its object.

Storage is additive.
For literal bit accounting, an object's cost includes its
encoded record and required framing.
The structural results also allow abstract positive weights;
algorithmic results use binary-encoded rational costs.

The canonical assumptions below require \(\Cn\) itself to obey
the designated-parent recursion.
Selecting one proof from a deductive system with additional
available proofs does not establish this requirement.

\begin{assumption}[Effective deterministic local semantics]
\label{assump:det-local}
The operator \(\Cn\) is extensive, idempotent, monotone, and finitary:
for every \(A\subseteq\mathcal U\),
\[
  \Cn(A)
  =
  \bigcup_{\substack{\Gamma\subseteq A\\ \Gamma\text{ finite}}}
  \Cn(\Gamma).
\]
Membership in \(\Cn(\Gamma)\) is decidable for finite
\(\Gamma\subseteq\mathcal U\).

There is a finite relevant region
\[
  \mathcal V_{\mathrm{rel}}(B)\subseteq\Cn(B),
  \qquad
  B\cup\mathcal Q\subseteq\mathcal V_{\mathrm{rel}}(B),
\]
with a computable designated-parent map
\[
  \operatorname{par}:
  \mathcal V_{\mathrm{rel}}(B)\setminus B
  \longrightarrow
  \mathcal V_{\mathrm{rel}}(B)^{<\infty}.
\]
For each relevant derived object \(v\), write
\[
  \operatorname{par}(v)=(u_1,\ldots,u_{a(v)}),
  \qquad a(v)\ge1,
\]
with pairwise distinct entries.
Write
\[
  \operatorname{pa}(v):=
  \{u_1,\ldots,u_{a(v)}\}
\]
for the underlying parent set.
The tuple \(\operatorname{par}(v)\) records its canonical
ordering, whereas \(\operatorname{pa}(v)\) is used in
set-valued statements.
Backward expansion along this map is well-founded and terminates
at objects in \(B\).

For every finite \(\Gamma\subseteq\Cn(B)\), availability satisfies
\begin{align}
  b\in\Cn(\Gamma)
  &\Longleftrightarrow b\in\Gamma,
  \qquad b\in B,
  \label{eq:base-irreducibility}\\
  v\in\Cn(\Gamma)
  &\Longleftrightarrow
  \bigl[
    v\in\Gamma
    \ \text{or}\
    \operatorname{pa}(v)\subseteq\Cn(\Gamma)
  \bigr],
  \notag\\[-1mm]
  &\hspace{35mm}
  v\in\mathcal V_{\mathrm{rel}}(B)\setminus B.
  \label{eq:deterministic-local-law}
\end{align}
Thus a base premise must be explicitly present, whereas a relevant
derived object is available either directly or through all of its
designated parents.
These equivalences apply even when \(\Gamma\) contains consequences
outside the relevant region.
\end{assumption}

\begin{assumption}[Finite leaf-hereditary canonical derivation DAG]
\label{assump:hereditary-dag}
For each \(v\in\mathcal V_{\mathrm{rel}}(B)\), let \(G(v,B)\)
be the graph obtained by recursively expanding \(v\) and every
encountered non-base vertex along its designated parent tuple.
Each parent \(u\) of a vertex \(w\) contributes the edge
\(u\to w\).
For a query \(q\in\mathcal Q\), write
\[
  G(q,B)=\bigl(V(G(q,B)),E(G(q,B))\bigr).
\]

By the finiteness and well-foundedness clauses of
Assumption~\ref{assump:det-local}, each such graph is a finite DAG.
Every vertex of \(G(q,B)\) reaches \(q\), which is its unique sink,
and its source set is exactly
\[
  V_0(G(q,B)):=V(G(q,B))\cap B.
\]
We call these sources \emph{base leaves}, reflecting their role as
terminal vertices under backward expansion.
For every \(v\in V(G(q,B))\setminus B\), its in-neighborhood in
\(G(q,B)\) is exactly \(\operatorname{pa}(v)\).
Tuple order belongs to the canonical description but plays no role
in graph reachability.

For \(v\in V(G(q,B))\), let \(G_q[v]\) be the subgraph induced by
\(v\) and all of its predecessors.
Restarting the same global parent expansion from \(v\) gives
\begin{equation}
  G_q[v]=G(v,B),
  \label{eq:hereditary-subdag}
\end{equation}
and hence
\begin{equation}
  V_0(G_q[v])=V_0(G(v,B)).
  \label{eq:hereditary-leaves}
\end{equation}
For \(v\in B\), both graphs are the singleton graph on \(v\).
This assumption records the graph representation and hereditary
identities induced by the global parent map.
\end{assumption}

\begin{lemma}[Horn realization of designated-parent DAGs]
\label{lem:horn-realization}
Let \(H\) be a nonempty finite DAG, let \(B\) be its full source set,
and assign each \(v\in V(H)\setminus B\) a parent tuple whose
underlying set is \(N_H^-(v)\).
Let \(\Cn_H(\Gamma)\) be the least set containing
\(\Gamma\subseteq V(H)\) and closed under the definite Horn rules
\[
  v\leftarrow N_H^-(v),
  \qquad v\in V(H)\setminus B.
\]
Then \(\Cn_H\), with relevant region \(V(H)\), satisfies
Assumptions~\ref{assump:det-local}
and~\ref{assump:hereditary-dag}.
For every \(q\in V(H)\), the canonical graph \(G(q,B)\) is the
subgraph of \(H\) induced by \(q\) and all of its predecessors.
\end{lemma}

\begin{proof}
For \(\Gamma\subseteq V(H)\), process the vertices in a topological
order and define
\[
  A_\Gamma(v)=
  \begin{cases}
    \mathbf 1\{v\in\Gamma\}, & v\in B,\\[1mm]
    \displaystyle
    \mathbf 1\{v\in\Gamma\}
    \,\lor\!
    \bigwedge_{u\in N_H^-(v)} A_\Gamma(u),
      & v\notin B.
  \end{cases}
\]
The set
\[
  C_\Gamma:=\{v\in V(H):A_\Gamma(v)=1\}
\]
contains \(\Gamma\) and is closed under every Horn rule.
Conversely, induction in the same topological order shows that
every rule-closed set containing \(\Gamma\) contains \(C_\Gamma\).
Thus \(C_\Gamma=\Cn_H(\Gamma)\).

Extensivity and monotonicity follow directly from the recursion.
Since \(C_\Gamma\) is already rule-closed, applying the operator
again leaves it unchanged, proving idempotence.
Finitarity is automatic on the finite universe \(V(H)\), and the
recursion computes the closure in
\(O(|V(H)|+|E(H)|)\) time.
It also gives the base-irreducibility and designated-parent laws.

Every non-source has a nonempty parent set, and backward expansion
in a finite DAG terminates at sources.
Taking \(\Gamma=B\) therefore gives \(\Cn_H(B)=V(H)\).
Expansion from \(q\) visits exactly its predecessors and includes
all parent edges among them.
Restarting from any visited vertex \(v\) yields the same predecessor
subgraph at \(v\), establishing the hereditary identities.
\end{proof}

\begin{remark}[Scope of the canonical semantics]
\label{rem:exact-scope}
\label{rem:assumption-roles}
Finite acyclic Horn systems with one rule per derived object
satisfy the canonical assumptions.
A designated proof skeleton or query plan also qualifies if
availability is evaluated within that restricted representation.
Retaining alternative derivations in the ambient closure
requires the AND--OR model of \ref{app:envelope};
query-dependent parent choices require query-indexed dependency
structures.
\end{remark}

\subsection{Premise Erasures and Transparent Cache Classes}
\label{subsec:erasure-cache-model}

Premises survive independently with probability
\(1-\varepsilon\), where \(\varepsilon\in(0,1)\):
\[
  Z_i\sim\operatorname{Bernoulli}(1-\varepsilon),
  \qquad
  \widetilde B:=\{b_i:Z_i=1\}.
\]
A transparent cache is a finite set
\(S\subseteq\Cn(B)\), chosen before erasure and not itself erased.
Its objects are admitted as explicit inputs after placement;
they need not be re-derived from \(\widetilde B\).
Recovery succeeds exactly when
\[
  q\in\Cn(\widetilde B\cup S).
\]
The cache is \((\varepsilon,\delta)\)-reliable for \(q\) if
\[
  \Pr_\varepsilon
  \left[q\in\Cn(\widetilde B\cup S)\right]\ge1-\delta,
  \qquad \delta\in(0,1).
\]

\begin{remark}[Trusted records and proof storage]
\label{rem:trusted-cache-records}
A cached object is a trusted explicit input.
If an independently checkable certificate is required, its
stored length is included in the object cost.
The structural results are unchanged when this affects only
the additive charge.
Additional certificate dependencies or erasures must instead
be represented explicitly.
Here proof-valid caching means that recovery respects the
specified derivation semantics with cached records admitted
as explicit inputs; it does not, by itself, require
reconstruction of an original proof from surviving premises.
\end{remark}

Fix a positive explicit-object cost function
\[
  c:\Cn(B)\longrightarrow\mathbb R_{>0},
\]
and define the additive storage cost of a finite cache by
\[
  \ell_c(S):=\sum_{v\in S}c(v).
\]
For bit accounting, \(c(v)\) is the length of its stored encoded
record; for normalized or weighted models, \(c(v)\) is the assigned
storage charge in the chosen unit.
Fixed system metadata does not supply stored copies of the
corresponding objects.
General real-valued costs suffice for the structural identities;
the algorithmic results use finitely encoded rational costs.

Direct-answer caching is a separate baseline.
Unless stated otherwise, we impose
\[
  q\notin S
  \qquad\text{or, for a workload }W,\qquad
  S_W\cap W=\varnothing.
\]
This isolates reuse of dependencies from storing the requested
answers themselves.
Workload queries are non-base consequences:
\(W\subseteq\mathcal Q\setminus B\).
For an answer-excluding base query, only its own survival matters.

For a workload
\[
  W=\{q_1,\ldots,q_L\}\subseteq\mathcal Q\setminus B,
\]
a common transparent cache \(S_W\subseteq\Cn(B)\) is \emph{jointly}
\((\varepsilon,\delta)\)-reliable if
\begin{equation}
  \Pr_{\varepsilon}\!\left[
    \bigcap_{\ell=1}^{L}
    \left\{
      q_\ell\in\Cn(\widetilde B\cup S_W)
    \right\}
  \right]
  \ge 1-\delta,
\label{eq:model-joint-criterion}
\end{equation}
and it is \emph{maximal-error} \((\varepsilon,\delta)\)-reliable if
\begin{equation}
  \Pr_{\varepsilon}\!\left[
    q_\ell\in\Cn(\widetilde B\cup S_W)
  \right]
  \ge 1-\delta,
  \qquad \ell=1,\ldots,L.
\label{eq:model-max-criterion}
\end{equation}
Joint reliability constrains one intersection event and
therefore implies maximal-error reliability.
Here maximal error is the maximum failure probability over
the fixed workload, with probability taken over premise
erasures.
It is distinct from worst-case reliability over payload
values, which is imposed separately in
\ref{sec:coded-outer}.

A \emph{leaf-only} cache satisfies \(S_W\subseteq B\).
A \emph{semantic-aware} cache may contain any finite subset of
\(\Cn(B)\), subject to answer exclusion.
An \emph{unrestricted coded} cache may store computable binary
functions of premise payloads, under a separately specified
operational recovery task.

For \(\bullet\in\{\mathrm{leaf},\mathrm{sem}\}\), write
\[
  \sigma_\bullet^{*,\mathrm{joint}}
  (W,\varepsilon,\delta),
  \qquad
  \sigma_\bullet^{*,\max}
  (W,\varepsilon,\delta)
\]
for the corresponding minimum storage costs.
We omit the criterion superscript when it is fixed by context.
Class inclusion gives
\[
  \sigma_{\mathrm{sem}}^*
  \le
  \sigma_{\mathrm{leaf}}^*.
\]

For the coded class, write \(\sigma_{\mathrm{code}}^*\) for the
minimum charged storage under a specified common query-recovery
task.  The additional inequality
\[
  \sigma_{\mathrm{code}}^*
  \le
  \sigma_{\mathrm{sem}}^*
\]
requires such a common task and a representation of
every transparent cache as a coded cache with no greater charge
and the same recovery behavior.
Under bit accounting, this requires that the full stored records
be computable from the payloads and charged consistently.
It is a comparison convention, not a consequence of closure
decidability.
\ref{sec:coded-outer} instead studies the stronger task
of recovering all relevant payloads.

\subsection{Standing Notation and the Survival Threshold}
\label{subsec:standing-notation}

For \(x\in\mathbb R\), write \(x^+:=\max\{x,0\}\).

For a fixed non-base query \(q\in\mathcal Q\setminus B\), we will frequently abbreviate
\(G_q:=G(q,B)\), \(V_0(G_q):=V(G_q)\cap B\), and
\(\kappa_q:=|V_0(G_q)|\); for a workload
\(W=\{q_1,\ldots,q_L\}\subseteq\mathcal Q\setminus B\), we also write
\(G_\ell:=G(q_\ell,B)\), \(D_\ell:=V_0(G_\ell)\),
\(\kappa_\ell:=|D_\ell|\), and
\[
  D_\cup:=\bigcup_{\ell=1}^{L}D_\ell.
\]
The fixed dependency union \(D_\cup\) is distinct from the
cache-dependent exposed union \(D_\cup(W,S_W)\) introduced in
Section~\ref{sec:semantic-scenario}.

Define the survival threshold
\begin{equation}
  N^*(\varepsilon,\delta)
  :=
  \max\{n\in\mathbb Z_{\ge0}:(1-\varepsilon)^n\ge1-\delta\}
  =
  \left\lfloor
    \frac{\ln(1-\delta)}{\ln(1-\varepsilon)}
  \right\rfloor.
  \label{eq:model-N-star}
\end{equation}
It is the largest number of distinct premises whose simultaneous
survival meets the reliability target.
The threshold is finite and nonnegative; in particular,
\(N^*(\varepsilon,\delta)=0\) when \(\delta<\varepsilon\).
For an explicitly represented dependency graph with terminal
set \(V_0\), only the clipped threshold
\[
  K:=\min\{|V_0|,N^*(\varepsilon,\delta)\}
\]
is needed.
For a canonical query graph, \(V_0=V_0(G_q)\).
When \(\varepsilon\) and \(\delta\) are rationally encoded, it can
be computed exactly by binary search over
\(j\in\{0,\ldots,|V_0|\}\), testing
\[
  (1-\varepsilon)^j\ge1-\delta
\]
with rational arithmetic.
This avoids rounding errors from evaluating the logarithmic formula.
The intermediate integers have bit length polynomial in the
explicit graph size and the parameter encoding.

\subsection{The Partial Path Interception Problem}
\label{subsec:ppi-definition}

We now define the graph optimization problem independently
of its caching interpretation.

\begin{definition}[Partial path interception, PPI]
\label{def:ppi}
An instance consists of a finite DAG \(G\), a distinguished sink
\(q\), a designated terminal set
\[
  V_0\subseteq V(G)\setminus\{q\},
\]
nonnegative vertex costs
\[
  c:V(G)\setminus\{q\}\longrightarrow\mathbb R_{\ge0},
\]
and an integer residual threshold \(N^*\ge0\).
We call the designated terminals \emph{leaves}; in the general
combinatorial problem they need not be sources.
In canonical derivation instances, they are exactly the source
vertices.

For \(S\subseteq V(G)\setminus\{q\}\), define
\[
  c(S):=\sum_{v\in S}c(v)
\]
and
\[
  \mathrm{Exp}(S)
  :=
  \{b\in V_0\setminus S:
    \text{\(b\) reaches \(q\) in \(G-S\)}\}.
\]
Thus selecting a leaf removes it and makes it non-exposed; selecting
the target \(q\) is forbidden.

The \emph{optimization problem} is to minimize \(c(S)\) subject to
\[
  |\mathrm{Exp}(S)|\le N^*.
\]
Its value is denoted by \(\mathrm{OPT}\), with
\(\mathrm{OPT}:=+\infty\) if no feasible set exists.
The \emph{decision problem} additionally receives a budget
\(\theta\ge0\) and asks whether a feasible set with
\(c(S)\le\theta\) exists.
For complexity statements, costs and budgets are rationals encoded
in binary, and the DAG is supplied explicitly.
For graph-operation bounds, we use the standard word-RAM model
with adjacency-list input; the residual threshold may first be
capped at \(|V_0|\), since larger values impose no restriction.
The \emph{coverage demand} is
\(t:=(|\mathrm{Exp}(\varnothing)|-N^*)^+\); it equals
\((|V_0|-N^*)^+\) whenever every terminal initially reaches \(q\), as
on canonical instances, and the empty set is optimal when \(t=0\).
Where explicitly stated, an instance may also specify an
admissible set \(\mathcal A\subseteq V(G)\setminus\{q\}\), in
which case feasible caches satisfy \(S\subseteq\mathcal A\);
by default \(\mathcal A=V(G)\setminus\{q\}\).
\end{definition}

\begin{definition}[Depth and layer terminology]
\label{def:depth}
The \emph{depth} of a PPI instance is the maximum number of arcs on a
directed path ending at \(q\).
In a canonical instance of depth two, every vertex other than \(q\) is
either a leaf (a source) or an \emph{intermediate vertex}; every in-neighbor
of an intermediate vertex is a leaf, and its only out-neighbor is \(q\).
Leaves may additionally have arcs to \(q\).
We write \(V_1\) for the set of intermediate vertices.
The \emph{fan-out} of a leaf is its number of intermediate out-neighbors,
and the \emph{in-degree} of an intermediate vertex is its number of leaf
parents.
\end{definition}

For a canonical query DAG, the terminals are its source leaves,
and \(\mathrm{Exp}(S)=D_{\mathrm{exp}}(q,S)\).
Section~\ref{sec:single-query-kernel} proves that reliability is
equivalent to
\[
  |\mathrm{Exp}(S)|\le N^*(\varepsilon,\delta).
\]
Conversely, if every vertex of a finite DAG reaches \(q\) and
the terminals are exactly its sources,
Lemma~\ref{lem:horn-realization} gives a canonical realization.
Positive vertex costs then preserve the transparent-cache
objective.

Definition~\ref{def:ppi} additionally permits zero costs and
explicitly forbidden vertices.
In the decision problem, a vertex is forbidden by pricing it
above the budget.

\begin{remark}[Explicit representation]
\label{rem:explicit-representation}
All algorithms take the dependency DAG, terminals, and rational
costs as explicit input.
For workloads, the input is the explicit union of the query DAGs.
The semantic assumptions guarantee computability, not
polynomial-time compilation from a compact deductive description;
such compilation costs are outside the stated running times.
\end{remark}

\section{The Exact Single-Query Law}
\label{sec:single-query-kernel}

Fix \(q\in\mathcal Q\setminus B\), and write
\(G_q:=G(q,B)\).
We characterize recovery for every
\(\widetilde B\subseteq B\), then derive its probability under
independent erasures.
Event identities allow arbitrary finite transparent caches;
answer exclusion is imposed only in the storage optimization.

\subsection{Query-Local Projection and Failure Paths}
\label{subsec:query-local-projection}

\begin{definition}[Query-local projection and exposed leaves]
\label{def:query-local-residual}
For a finite transparent cache \(S\subseteq\Cn(B)\), define its
query-local projection by
\begin{equation}
  U(q,S):=S\cap V(G_q).
  \label{eq:query-local-projection}
\end{equation}
For \(b\in V_0(G_q)\), let \(\mathcal P_q(b)\) be the family of
directed paths from \(b\) to \(q\) in \(G_q\).
For a path \(P\), write \(V(P)\) for its vertex set, including both
endpoints.

A path is \emph{cache-free} if \(V(P)\cap U(q,S)=\varnothing\).
Define the exposed-leaf set by
\begin{equation}
  D_{\mathrm{exp}}(q,S)
  :=
  \bigl\{
    b\in V_0(G_q):
    \exists P\in\mathcal P_q(b),\
    V(P)\cap U(q,S)=\varnothing
  \bigr\},
  \label{eq:def-exposed-leaves}
\end{equation}
and the blocked-leaf set by
\begin{equation}
  D_{\mathrm{blk}}(q,S)
  :=
  V_0(G_q)\setminus D_{\mathrm{exp}}(q,S).
  \label{eq:def-blocked-leaves}
\end{equation}
Thus a leaf is blocked exactly when every path from that leaf to
\(q\) intersects the cache.
\end{definition}

Exposure requires one cache-free path; blocking requires all
paths to be intercepted.
These are dependency paths in one conjunctive derivation,
not alternative proofs.

\begin{lemma}[Failure-path characterization]
\label{lem:failure-path}
Under Assumptions~\ref{assump:det-local}
and~\ref{assump:hereditary-dag}, let
\(S\subseteq\Cn(B)\) be finite and set \(U:=U(q,S)\).
For every \(\widetilde B\subseteq B\) and \(v\in V(G_q)\),
\begin{equation}
  \begin{aligned}
    v\notin\Cn(\widetilde B\cup U)
    \quad\Longleftrightarrow\quad&
    \text{there exist }
    b\in V_0(G_q[v])\setminus\widetilde B\\
    &\text{and a directed path }P:b\leadsto v
      \text{ in }G_q[v]\\
    &\text{such that }V(P)\cap U=\varnothing.
  \end{aligned}
  \label{eq:failure-path}
\end{equation}
Paths of length zero are allowed.
\end{lemma}

\begin{proof}
We induct over the vertices of \(G_q\) in a topological order.

If \(v\in B\), base irreducibility gives
\[
  v\notin\Cn(\widetilde B\cup U)
  \quad\Longleftrightarrow\quad
  v\notin\widetilde B\cup U.
\]
Since \(G_q[v]\) is the singleton graph on \(v\), this is exactly
the condition that the length-zero path at the missing leaf \(v\)
avoids \(U\).

Now suppose \(v\notin B\).
If \(v\in U\), then \(v\) is available and every path ending at \(v\)
meets \(U\), so both sides of \eqref{eq:failure-path} are false.
If \(v\notin U\), the local law gives
\begin{equation}
  v\notin\Cn(\widetilde B\cup U)
  \quad\Longleftrightarrow\quad
  \exists u\in\operatorname{pa}(v):
  u\notin\Cn(\widetilde B\cup U).
  \label{eq:parent-failure-recursion}
\end{equation}
By induction, an unavailable parent \(u\) has a cache-free path
from a missing base leaf to \(u\).
Appending \(u\to v\) gives the required path to \(v\).
Conversely, a cache-free path from a base leaf to the non-base
vertex \(v\) has a final edge \(u\to v\) from a designated parent.
Its prefix makes \(u\) unavailable by induction, and
\eqref{eq:parent-failure-recursion} makes \(v\) unavailable.
\end{proof}

\begin{theorem}[Query-local rigidity and exact residual-leaf law]
\label{thm:exact-residual-law}
Under Assumptions~\ref{assump:det-local}
and~\ref{assump:hereditary-dag}, for every finite
\(S\subseteq\Cn(B)\) and every \(\widetilde B\subseteq B\),
\begin{align}
  q\in\Cn(\widetilde B\cup S)
  &\Longleftrightarrow
  q\in\Cn\bigl(\widetilde B\cup U(q,S)\bigr),
  \label{eq:query-local-rigidity}\\
  q\in\Cn\bigl(\widetilde B\cup U(q,S)\bigr)
  &\Longleftrightarrow
  D_{\mathrm{exp}}(q,S)\subseteq\widetilde B.
  \label{eq:exact-residual-event}
\end{align}
Consequently, under independent premise erasures,
\begin{equation}
  \Pr_{\varepsilon}\!\left[
    q\in\Cn(\widetilde B\cup S)
  \right]
  =
  (1-\varepsilon)^{|D_{\mathrm{exp}}(q,S)|}.
  \label{eq:exact-residual-probability}
\end{equation}
In particular, with \(N^*(\varepsilon,\delta)\) defined by
\eqref{eq:model-N-star}, the cache \(S\) is
\((\varepsilon,\delta)\)-reliable for \(q\) if and only if
\(|D_{\mathrm{exp}}(q,S)|\le N^*(\varepsilon,\delta)\).
\end{theorem}

\begin{proof}
Set \(U:=U(q,S)\).
For every vertex \(v\in V(G_q)\), direct membership in
\(\widetilde B\cup S\) agrees with direct membership in
\(\widetilde B\cup U\).
For a base vertex, base irreducibility therefore gives the same
availability under both inputs.
For a non-base vertex, all designated parents belong to \(G_q\);
the local law and induction in a topological order give the same
availability under both inputs.
Applying this at \(q\) proves
\eqref{eq:query-local-rigidity}.

By Lemma~\ref{lem:failure-path}, \(q\) is unavailable under
\(\widetilde B\cup U\) exactly when some missing base leaf has a
cache-free path to \(q\).
This is equivalent to
\[
  D_{\mathrm{exp}}(q,S)\setminus\widetilde B\neq\varnothing.
\]
Taking complements proves \eqref{eq:exact-residual-event}.

The exposed set is fixed before erasure.
Recovery requires precisely the survival of its distinct members,
so independence gives \eqref{eq:exact-residual-probability}.
\end{proof}

Independence is needed only for the probability formula.
Shared leaves are counted once.
In particular, \(q\in S\) gives certain recovery, while
\(S=\varnothing\) gives
\[
  \Pr_\varepsilon[q\in\Cn(\widetilde B)]
  =(1-\varepsilon)^{\kappa_q}.
\]
Thus no cache is needed exactly when
\(\kappa_q\le N^*(\varepsilon,\delta)\).

\subsection{Path Interception and Query Postdominators}
\label{subsec:postdominator-interpretation}

\begin{definition}[Query-postdominated leaves]
\label{def:query-postdominator}
For \(v\in V(G_q)\), define
\[
  \operatorname{Pdom}_q(v)
  :=
  \{b\in V_0(G_q):
    \text{every directed \(b\)-to-\(q\) path contains \(v\)}\}.
\]
In particular,
\[
  \operatorname{Pdom}_q(b)=\{b\}
  \quad (b\in V_0(G_q)),
  \qquad
  \operatorname{Pdom}_q(q)=V_0(G_q).
\]
\end{definition}

Caching a single vertex \(v\) blocks exactly
\(\operatorname{Pdom}_q(v)\).
This is the set of designated leaves in the subtree rooted at
\(v\) in the postdominator tree, equivalently the dominator tree
of the reversed graph rooted at \(q\)
\citep{lengauer1979fast}.

On an in-tree,
\(\operatorname{Pdom}_q(v)=V_0(G_q[v])\).
In a general DAG, a predecessor leaf may bypass \(v\), and
\[
  \bigcup_{v\in U(q,S)}\operatorname{Pdom}_q(v)
  \subseteq D_{\mathrm{blk}}(q,S)
\]
may be strict: different selected vertices can intercept
different paths from the same leaf.
Figure~\ref{fig:interception-modes} illustrates this distinction.

\subsection{Exact Transparent-Storage Optimization}
\label{subsec:exact-storage-optimization}

For a single query the joint and maximal-error criteria coincide, and
we write \(\sigma^*_{\mathrm{sem}}(q,\varepsilon,\delta)\) and
\(\sigma^*_{\mathrm{leaf}}(q,\varepsilon,\delta)\) for the optima of
Subsection~\ref{subsec:erasure-cache-model} with \(W=\{q\}\).
Caches exclude \(q\); admitting \(q\) would only add the
direct-answer baseline of cost \(c(q)\).

\begin{theorem}[Exact query-local storage formulation]
\label{thm:exact-single-query-opt}
Under Assumptions~\ref{assump:det-local} and
\ref{assump:hereditary-dag},
\begin{align}
  \sigma^*_{\mathrm{sem}}(q,\varepsilon,\delta)
  =
  \min\Bigl\{
    \ell_c(S):
    S\subseteq V(G_q)\setminus\{q\},
    |D_{\mathrm{exp}}(q,S)|
      \le N^*(\varepsilon,\delta)
  \Bigr\}.
\label{eq:exact-single-query-opt}
\end{align}
Thus the answer-excluding single-query problem is exactly a weighted partial
path-interception problem: the cache must intercept every path from all but at
most \(N^*(\varepsilon,\delta)\) base leaves to \(q\).
\end{theorem}

\begin{proof}
Theorem~\ref{thm:exact-residual-law} converts reliability
exactly into the exposed-leaf constraint.
For every feasible finite cache \(S\), its projection
\[
  U(q,S)=S\cap V(G_q)
\]
has the same recovery event and no greater cost.
Conversely, every subset of
\(V(G_q)\setminus\{q\}\) is admissible.
The reduced family is finite and feasible, since caching
all source leaves guarantees recovery.
Hence its attained minimum equals the original optimum.
\end{proof}

\begin{corollary}[Exact leaf-only transparent baseline]
\label{cor:leaf-only-opt}
Consider the leaf-only class \(S\subseteq B\).
By query-local projection, an optimum may be restricted to
\(S\subseteq V_0(G_q)\).
Arrange the costs of these \(\kappa_q\) leaves in nondecreasing order:
\[
  c_{(1)}\le c_{(2)}\le\cdots\le c_{(\kappa_q)}.
\]
Then, writing \(N^*=N^*(\varepsilon,\delta)\),
\begin{equation}
  \sigma^*_{\mathrm{leaf}}(q,\varepsilon,\delta)
  =
  \sum_{i=1}^{(\kappa_q-N^*)^+}c_{(i)},
\label{eq:heterogeneous-leaf-opt}
\end{equation}
where an empty sum equals zero; under a homogeneous leaf cost \(c_B\)
this equals \((\kappa_q-N^*)^+c_B\).
\end{corollary}

\begin{proof}
Leaves outside \(V_0(G_q)\) can be discarded by query-local
projection.  For \(S\subseteq V_0(G_q)\), no path from one source
passes through a distinct source, so a leaf is blocked exactly when it
belongs to \(S\), and
\(D_{\mathrm{exp}}(q,S)=V_0(G_q)\setminus S\).
Reliability is therefore equivalent to \(|S|\ge(\kappa_q-N^*)^+\), and
positive costs give the minimum at the cheapest such leaves.
\end{proof}

\subsection{Leaf Completion}
\label{subsec:single-query-module-benchmark}

Caching base leaves interacts with the rest of a cache in the simplest
possible way: this reduces the single-query optimum to the choice of
its non-leaf part.

\begin{proposition}[Leaf completion]
\label{prop:module-root-achievability}
For all \(C\subseteq V(G_q)\setminus\{q\}\) and
\(Y\subseteq V_0(G_q)\),
\begin{equation}
  D_{\mathrm{exp}}(q,C\cup Y)=D_{\mathrm{exp}}(q,C)\setminus Y.
  \label{eq:residual-after-leaf-completion}
\end{equation}
Consequently, if every base leaf costs \(c_B\), then
\begin{equation}
  \sigma^*_{\mathrm{sem}}(q,\varepsilon,\delta)
  =\min_{C\subseteq V(G_q)\setminus(B\cup\{q\})}
   \Bigl\{\ell_c(C)+\bigl(|D_{\mathrm{exp}}(q,C)|
   -N^*(\varepsilon,\delta)\bigr)^+c_B\Bigr\}.
  \label{eq:module-achievable-cost}
\end{equation}
\end{proposition}

\begin{proof}
A base leaf is a source, so it lies on no path from a distinct leaf, and
caching it blocks exactly itself; this proves
\eqref{eq:residual-after-leaf-completion}.  By
Theorem~\ref{thm:exact-single-query-opt}, every candidate cache has the
form \(C\cup Y\) with \(C\subseteq V(G_q)\setminus(B\cup\{q\})\)
and \(Y\subseteq V_0(G_q)\).  It is feasible exactly when \(Y\)
contains at least
\((|D_{\mathrm{exp}}(q,C)|-N^*(\varepsilon,\delta))^+\) leaves of
\(D_{\mathrm{exp}}(q,C)\), and the cheapest such \(Y\) gives
\eqref{eq:module-achievable-cost}.
\end{proof}

\begin{figure}[tbp]
\centering
\begin{adjustbox}{max width=\linewidth}
\begin{tikzpicture}[ppi,x=1.35cm,y=1.35cm]
\node[leaf]   (a-b) at (0,0.5) {$b$};
\node[cached] (a-r) at (1,0.5) {$r$};
\node[query]  (a-q) at (2,0.5) {$q$};
\draw[arc] (a-b)--(a-r); \draw[arc] (a-r)--(a-q);
\node[note] at (1,-0.55) {(a) \(r\) individually blocks \(b\)};
\begin{scope}[xshift=4.3cm]
\node[leaf]   (b-b) at (0,0.5) {$b$};
\node[cached] (b-u) at (1,1)   {$u$};
\node[cached] (b-v) at (1,0)   {$v$};
\node[query]  (b-q) at (2,0.5) {$q$};
\draw[arc] (b-b)--(b-u); \draw[arc] (b-b)--(b-v);
\draw[arc] (b-u)--(b-q); \draw[arc] (b-v)--(b-q);
\node[note] at (1,-0.55) {(b) \(\{u,v\}\) blocks \(b\) jointly};
\end{scope}
\begin{scope}[xshift=9.1cm]
\node[leaf,exposed] (c-b)  at (0,0.5) {$b$};
\node[cached]       (c-r)  at (1,1)   {$r$};
\node[query]        (c-q1) at (2,1)   {$q_1$};
\node[query]        (c-q2) at (2,0)   {$q_2$};
\draw[arc] (c-b)--(c-r); \draw[arc] (c-r)--(c-q1);
\draw[free] (c-b)--(c-q2);
\node[note] at (1,-0.55) {(c) the bypass keeps \(b\) exposed};
\end{scope}
\end{tikzpicture}
\end{adjustbox}
\caption{Individual, set-wise, and workload interception; cached
vertices are shaded.  (a)~The vertex \(r\) postdominates \(b\).
(b)~Neither \(u\) nor \(v\) postdominates \(b\), but together they
intercept both paths; the parent set of \(q\) is conjunctive, so these
paths are not alternative proofs.  (c)~Caching \(r\) protects \(q_1\),
but the direct dependency of \(q_2\) keeps \(b\) workload-exposed.}
\label{fig:interception-modes}
\end{figure}

\section{Complexity and Algorithms for PPI}
\label{sec:complexity-classification}

We study classical hardness, exact algorithms,
residual-threshold parameterization, and approximation.
The input graph is explicit, and costs are nonnegative
binary-encoded rationals unless stated otherwise.
Running times count arithmetic operations on values of
polynomial bit length.
The notation \(\mathrm{Exp}(S)\) agrees with
\(D_{\mathrm{exp}}(q,S)\) on canonical query DAGs.
Unless stated otherwise, ETH and Gap-ETH lower bounds concern
deterministic algorithms.
The corresponding exclusions for bounded-error randomized
algorithms require the randomized versions of these hypotheses.
SPCH is used in its probabilistic formulation.
Figure~\ref{fig:landscape} summarizes how the residual-threshold
complexity depends on the cost model and the dependency structure.

\begin{figure}[ht]
\centering
\begin{tikzpicture}[ppi,x=1.95cm,y=1.4cm,
  cell/.style={draw=white,line width=1.2pt,minimum width=1.95cm,
    minimum height=1.4cm,text width=1.85cm,align=center,
    font=\scriptsize,inner sep=1pt},
  hd/.style={cell,fill=black!8,font=\scriptsize\bfseries},
  cP/.style={cell,fill=green!22},
  cF/.style={cell,fill=cyan!20},
  cH/.style={cell,fill=red!18},
  cO/.style={cell,fill=yellow!35}]
\node[hd] at (0,0) {parameter \(N^*\)};
\node[hd] at (1,0) {in-tree};
\node[hd] at (2,0) {depth 2,\\ in-degree \(\le2\)};
\node[hd] at (3,0) {depth 2,\\ fan-out \(\le p\)};
\node[hd] at (4,0) {depth 2,\\ general};
\node[hd] at (5,0) {depth 3,\\ leaf out-deg.\ 1};
\node[hd, font=\linespread{0.85}\scriptsize\bfseries] at (0,-1) {weighted,\\ leaves above\\ budget};
\node[hd, font=\linespread{0.85}\scriptsize\bfseries] at (0,-2) {unit costs,\\ leaves\\ admissible};
\node[cP] at (1,-1) {polynomial\\ Thm.~\ref{thm:unique-path-dp}};
\node[cH] at (2,-1) {W[1]-hard,\\ ETH-tight\\ Thm.~\ref{thm:ppi-wh}};
\node[cF] at (3,-1) {FPT in \(N^*{+}p\)\\ Thm.~\ref{thm:ppi-fpt}};
\node[cH] at (4,-1) {W[1]-hard,\\ ETH-tight\\ (contains col.~2)};
\node[cH] at (5,-1) {W[1]-hard\\ Prop.~\ref{prop:fanout-depth}};
\node[cP] at (1,-2) {polynomial\\ Thm.~\ref{thm:unique-path-dp}};
\node[cP] at (2,-2) {linear time\\ for every \(N^*\)\\
  Cor.~\ref{cor:indegree-two-poly}};
\node[cO] at (3,-2) {NP-hard (\(p\ge2\));\\ FPT open\\
  Prop.~\ref{prop:module-selection-np-hard}};
\node[cO] at (4,-2) {NP-hard;\\ FPT open\\ Rem.~\ref{rem:nstar-open}};
\node[cH] at (5,-2) {W[1]-hard\\ Thm.~\ref{thm:unit-wh}};
\node[cell,fill=green!10,minimum width=11.7cm,text width=11.5cm]
  at (2.5,-3) {every class: XP in \(N^*\)
  (Prop.~\ref{prop:xp-threshold}); FPT in treewidth
  (Thm.~\ref{thm:bounded-treewidth-fpt})};
\node[fill=green!22,minimum size=8pt] (k1) at (0.3,-3.75) {};
\node[right=2pt of k1,font=\scriptsize] (t1) {polynomial};
\node[fill=cyan!20,minimum size=8pt,right=10pt of t1] (k2) {};
\node[right=2pt of k2,font=\scriptsize] (t2) {FPT};
\node[fill=red!18,minimum size=8pt,right=10pt of t2] (k3) {};
\node[right=2pt of k3,font=\scriptsize] (t3) {W[1]-hard};
\node[fill=yellow!35,minimum size=8pt,right=10pt of t3] (k4) {};
\node[right=2pt of k4,font=\scriptsize] {open};
\end{tikzpicture}
\caption{The residual-threshold landscape.  Rows fix the cost model and
columns restrict the dependency structure; the column classes overlap,
and each cell concerns the class in its header.  The in-degree-two
column isolates the effect of costs: the weighted hardness construction
of Theorem~\ref{thm:ppi-wh} becomes polynomial once its leaves cost one,
because deleting a leaf is then affordable.  Depth three restores
hardness for unit costs.  The weighted hardness also persists, with the
same ETH exponent, when every leaf has out-degree one and every non-sink
vertex has out-degree at most two
(Proposition~\ref{prop:fanout-depth}).}
\label{fig:landscape}
\end{figure}

\subsection{NP-Completeness and Budget Parameterization}
\label{subsec:cache-selection-hardness}

\begin{proposition}[NP-completeness on depth-two derivation DAGs]
\label{prop:module-selection-np-hard}
The decision version of PPI is NP-complete, even on
canonical derivation DAGs of depth two with unit costs and exactly
two directed paths from every base leaf to the query.
In particular, every base leaf has out-degree two.
The residual threshold and storage budget are part of the input.

The same hardness applies to the joint workload problem, which
contains the single-query problem as the case \(L=1\).
\end{proposition}

\begin{proof}
Given a candidate cache \(S\), one reverse reachability search from
\(q\) in \(G-S\) identifies all exposed terminals.
The cost and threshold constraints are therefore checkable in
polynomial time, proving membership in NP.

For hardness, reduce from \textsc{Clique}
\citep{garey1979computers}.
Let \(H=(V_H,E_H)\) be a simple undirected graph, write
\(m:=|E_H|\), and let \(\kappa\ge4\).
We may assume \(m\ge\binom{\kappa}{2}\).
Inputs with \(\kappa<4\) or \(m<\binom{\kappa}{2}\) can be decided
directly and mapped to fixed YES or NO instances.

Let
\[
  V_+:=\{v\in V_H:\deg_H(v)>0\}.
\]
Construct a DAG with base leaves \(B=\{p_e:e\in E_H\}\),
intermediate vertices \(X=\{x_v:v\in V_+\}\), and target \(q\).
For each edge \(e=\{u,v\}\), add
\[
  p_e\to x_u,\qquad p_e\to x_v,
\]
and add \(x_v\to q\) for every \(v\in V_+\).
Equivalently, the designated parent tuples are
\[
  \operatorname{par}(x_v)=(p_e:e\ni v),
  \qquad
  \operatorname{par}(q)=(x_v:v\in V_+),
\]
with any fixed ordering within each tuple.
All tuples are nonempty.
Lemma~\ref{lem:horn-realization} realizes this graph as a canonical
conjunctive derivation DAG.

Assign unit costs and set
\[
  N^*:=m-\binom{\kappa}{2},
  \qquad
  \theta:=\kappa.
\]
The construction has depth two, and each leaf has exactly two
paths to \(q\).
Figure~\ref{fig:clique-reduction} illustrates the construction.

For \(C\subseteq X\), define
\[
  e(C):=
  \bigl|
    \{\{u,v\}\in E_H:x_u,x_v\in C\}
  \bigr|.
\]
A leaf \(p_{\{u,v\}}\) is blocked by \(C\) exactly when both
\(x_u\) and \(x_v\) belong to \(C\).
Thus \(m-e(C)\) leaves remain exposed.
Proposition~\ref{prop:module-root-achievability}, with
\(c_B=1\), therefore gives
\[
  \mathrm{OPT}=\min_{C\subseteq X}
  \Bigl[|C|+\Bigl(\tbinom{\kappa}{2}-e(C)\Bigr)^+\Bigr].
\]

If \(H\) has a \(\kappa\)-clique, its corresponding \(\kappa\)
intermediate vertices form a feasible cache of cost \(\kappa\).

Conversely, suppose the displayed minimum is at most \(\kappa\),
and let \(a:=|C|\) for a minimizing set.
Clearly \(a\le\kappa\).
If \(a=0\), its value is \(\binom{\kappa}{2}>\kappa\).
If \(1\le a<\kappa\), simplicity of \(H\) gives
\(e(C)\le\binom{a}{2}\), so its value is at least
\[
  a+\binom{\kappa}{2}-\binom{a}{2}
  \ge
  (\kappa-1)+\binom{\kappa}{2}-\binom{\kappa-1}{2}
  =
  2(\kappa-1)>\kappa.
\]
Here \(a-\binom{a}{2}\) is nonincreasing for integers \(a\ge1\).
Therefore \(a=\kappa\).
A value at most \(\kappa\) then forces
\(e(C)=\binom{\kappa}{2}\), so the corresponding vertices form a
\(\kappa\)-clique.

The construction is polynomial, completing the reduction.
\end{proof}

\begin{figure}[tbp]
\centering
\begin{adjustbox}{max width=\linewidth}
\begin{tikzpicture}[
  x=0.84cm,
  font=\scriptsize,
  leaf/.style={circle,draw,thick,fill=blue!10,inner sep=1.2pt,minimum size=15pt},
  int/.style={rectangle,draw,thick,rounded corners=2pt,fill=orange!20,inner sep=2.2pt},
  cached/.style={int,fill=orange!45,very thick},
  query/.style={rectangle,draw,thick,rounded corners=2pt,double,inner sep=2.2pt},
  kv/.style={circle,draw,very thick,fill=orange!45,inner sep=1pt,minimum size=14pt},
  wv/.style={circle,draw,thick,fill=orange!20,inner sep=1pt,minimum size=14pt},
  reg/.style={draw=black!55,dashed,rounded corners=3pt},
  lab/.style={font=\scriptsize,text=black!70},
  >=stealth]
\node[kv] (v1) at (-3.75,2.6) {\(v_{1}\)};
\node[kv] (v2) at (-2.50,1.4) {\(v_{2}\)};
\node[kv] (v3) at (-3.75,0.2) {\(v_{3}\)};
\node[kv] (v4) at (-5.00,1.4) {\(v_{4}\)};
\node[wv] (v5) at (-5.95,-0.5) {\(v_{5}\)};
\draw[black!70,thick] (v1)--(v2) (v1)--(v3) (v1)--(v4) (v2)--(v3) (v3)--(v4);
\draw[black!70,thick] (v2) to[bend left=15] (v4);
\draw[red!70!black,thick] (v3)--(v5) (v4)--(v5);
\node[lab] at (-3.95,-1.9) {(a) source instance \((H,\kappa)\), \(\kappa=4\), \(m=8\)};
\node[leaf] (p12) at (0.95,0) {\(p_{12}\)};
\node[leaf] (p13) at (1.90,0) {\(p_{13}\)};
\node[leaf] (p14) at (2.85,0) {\(p_{14}\)};
\node[leaf] (p23) at (3.80,0) {\(p_{23}\)};
\node[leaf] (p24) at (4.75,0) {\(p_{24}\)};
\node[leaf] (p34) at (5.70,0) {\(p_{34}\)};
\node[leaf] (p35) at (7.15,0) {\(p_{35}\)};
\node[leaf] (p45) at (8.10,0) {\(p_{45}\)};
\node[cached] (x1) at (1.70,1.8) {\(x_{1}\)};
\node[cached] (x2) at (3.05,1.8) {\(x_{2}\)};
\node[cached] (x3) at (4.40,1.8) {\(x_{3}\)};
\node[cached] (x4) at (5.75,1.8) {\(x_{4}\)};
\node[int]    (x5) at (7.60,1.8) {\(x_{5}\)};
\node[query] (q) at (4.65,3.3) {\(q\)};
\draw[->,black!45,thin] (p12)--(x1) (p12)--(x2) (p13)--(x1) (p13)--(x3)
  (p14)--(x1) (p14)--(x4) (p23)--(x2) (p23)--(x3) (p24)--(x2) (p24)--(x4)
  (p34)--(x3) (p34)--(x4) (p35)--(x3) (p35)--(x5) (p45)--(x4) (p45)--(x5);
\draw[->,black!45,thin] (x1)--(q) (x2)--(q) (x3)--(q) (x4)--(q) (x5)--(q);
\draw[->,red!70!black,very thick] (p35)--(x5) (p45)--(x5) (x5)--(q);
\node[lab, anchor=west, align=left] at (0.45,3.30)
  {\(|C|=\theta=4\)\\ \(|D_{\mathrm{exp}}(q,C)|=N^{*}=2\)};
\node[reg, inner sep=2.5pt, fit=(p12)(p34)] (cbox) {};
\node[lab, align=center, anchor=north] at (cbox.south) {\(\binom{\kappa}{2}=6\)\\clique-edge leaves};
\node[reg, inner sep=2.5pt, fit=(p35)(p45)] (rbox) {};
\node[lab, align=center, anchor=north] at (rbox.south) {\(m-\binom{\kappa}{2}=2\)\\residual leaves};
\node[lab] at (4.65,-1.9) {(b) constructed depth-two canonical DAG \(G_q\)};
\end{tikzpicture}
\end{adjustbox}
\caption{The depth-two construction of
Proposition~\ref{prop:module-selection-np-hard} for a graph with
a four-vertex clique and eight edges.
Each edge becomes a base leaf, each non-isolated graph vertex
becomes an intermediate vertex, and all intermediate vertices point to
\(q\).
Caching the four clique vertices blocks the six clique-edge
leaves.
The two remaining leaves retain cache-free paths through
\(x_5\), highlighted in red, meeting the residual threshold
\(N^*=2\).}
\label{fig:clique-reduction}
\end{figure}

\begin{corollary}[Budget and demand parameters]
\label{cor:demand-budget-wh}
Unit-cost PPI on depth-two DAGs with leaf out-degree two is W[1]-hard
parameterized by the budget \(\theta\), by the coverage demand
\(t\), and by \(t+\theta\); it is in XP for each of these
parameters, and it admits no EPTAS unless \(\mathrm{FPT}=\mathrm{W}[1]\).
\end{corollary}

\begin{proof}
The reduction of
Proposition~\ref{prop:module-selection-np-hard}
has
\[
  \theta=\kappa,
  \qquad
  t=\binom{\kappa}{2}.
\]
Hence W[1]-hardness transfers to each of the three
parameters \citep{downey2013fundamentals}.

Under unit costs, a rational budget can be replaced
by its integer part without changing the answer.
Thus assume that \(\theta\) is a nonnegative integer.
Enumerating all vertex sets of size at most \(\theta\)
gives XP membership for the budget parameter.

Caching any \(t\) initially exposed terminals is feasible,
so instances with \(\theta\ge t\) are immediately YES.
On the remaining instances,
\[
  t\le t+\theta<2t,
\]
and every solution within budget has fewer than \(t\)
vertices.
Enumeration therefore also gives XP membership for
\(t\) and \(t+\theta\).

For the EPTAS exclusion, handle \(\theta=0\) directly.
For \(\theta\ge1\), run the EPTAS with accuracy
\(\eta:=1/(2\theta)\).
If \(\mathrm{OPT}\le\theta\), its feasible output has
integer cost at most
\[
  (1+\eta)\theta=\theta+\tfrac12,
\]
and hence at most \(\theta\).
If \(\mathrm{OPT}>\theta\), every feasible output costs
at least \(\theta+1\).
This would yield an FPT decision algorithm parameterized
by \(\theta\), contradicting W[1]-hardness unless
\(\mathrm{FPT}=\mathrm{W[1]}\)
\citep{marx2008parameterized}.
\end{proof}

\subsection{Exact Algorithms}
\label{subsec:exact-algorithms}

\subsubsection{Unique-Path Instances}
\label{subsec:unique-path-dp}

Canonical DAGs with one path from each source to \(q\)
are precisely in-trees.
Their optimum follows from a truncated subtree dynamic program.

\begin{lemma}[Unique-path canonical DAGs are in-trees]
\label{lem:unique-path-intree}
Let \(G_q\) be a canonical derivation DAG under
Assumptions~\ref{assump:det-local} and~\ref{assump:hereditary-dag} with
\(q\notin B\).  The following are equivalent:
\begin{enumerate}[label=\textup{(\roman*)}]
\item every base leaf \(b\in V_0(G_q)\) has a unique directed path to
\(q\) in \(G_q\);
\item every vertex \(v\in V(G_q)\setminus\{q\}\) has out-degree exactly
one in \(G_q\);
\item \(G_q\) is an in-tree rooted at \(q\): the underlying undirected
graph is a tree and every edge is oriented toward \(q\).
\end{enumerate}
The condition can be checked in
\(O(|V(G_q)|+|E(G_q)|)\) time.
When it holds,
\begin{equation}
  \operatorname{Pdom}_q(v)=V_0(G_q[v])
  \qquad(v\in V(G_q)),
\label{eq:intree-pdom}
\end{equation}
and these sets form a laminar family.
\end{lemma}

\begin{proof}
(i)\(\Rightarrow\)(ii).  Every vertex of \(G_q\) reaches \(q\) by the
construction of the canonical DAG.  Suppose some
\(v\in V(G_q)\setminus\{q\}\) had two distinct out-neighbors
\(w_1\neq w_2\).  Backward expansion from \(v\) along in-edges terminates
at a base leaf \(b\), giving a directed path \(b\leadsto v\); forward
expansion from \(w_1\) and from \(w_2\) gives directed paths
\(w_1\leadsto q\) and \(w_2\leadsto q\).  Concatenation produces two
distinct directed paths from \(b\) to \(q\), contradicting (i).  Hence
every \(v\neq q\) has out-degree at most one, and exactly one because
\(v\) reaches \(q\).

(ii)\(\Rightarrow\)(iii).  Each of the \(|V(G_q)|-1\) vertices other than
\(q\) contributes exactly one out-edge, so
\(|E(G_q)|=|V(G_q)|-1\).  Following out-edges from any vertex yields a
directed path, finite by acyclicity, which ends at a sink; by (ii) the
only sink is \(q\), so the underlying graph is weakly connected.  A weakly
connected undirected graph on \(|V(G_q)|\) vertices with \(|V(G_q)|-1\)
edges is a tree, and every edge is oriented toward \(q\).

(iii)\(\Rightarrow\)(i).  In a tree, any two vertices are joined by a
unique undirected path; since every edge is oriented toward \(q\), the
path from a source to \(q\) is directed, and it is the only directed
path.

The decidability statement is immediate from (ii).  Under the in-tree
structure, the sources are precisely the base leaves, and a base leaf's
unique path to \(q\) passes through \(v\) if and only if the leaf belongs
to the predecessor-induced sub-DAG \(G_q[v]\); hence
\(\operatorname{Pdom}_q(v)=V_0(G_q[v])\) by
Definition~\ref{def:query-postdominator}.  Subtree leaf sets of a tree are
laminar.
\end{proof}

Fix an in-tree instance \(G_q\).  For each vertex \(v\), write
\[
  T_v:=V(G_q[v]),
  \qquad
  L_v:=V_0(G_q[v]),
\]
for the vertex set and the base-leaf set of the subtree rooted at \(v\).
For a non-base vertex \(v\), its children are its in-neighbors
\(\operatorname{par}(v)=(u_1,\ldots,u_{a(v)})\).
Their subtrees partition \(T_v\setminus\{v\}\), and their
base-leaf sets partition \(L_v\).
A base leaf has no children.
For a cache
\(S\subseteq T_v\setminus\{q\}\), call a leaf \(b\in L_v\)
\emph{\(v\)-exposed} when its unique path from \(b\) to \(v\) avoids
\(S\).  For integers \(e\ge0\), define
\begin{equation}
  F(v,e)
  :=
  \min\Bigl\{
    \ell_c(S):
    S\subseteq T_v\setminus\{q\},\
    \bigl|\{b\in L_v:\ b\text{ is \(v\)-exposed}\}\bigr|=e
  \Bigr\},
\label{eq:dp-state}
\end{equation}
with \(F(v,e)=+\infty\) when no such cache exists.  
We retain only states with
\[
  0\le e\le K_v:=\min\{|L_v|,N^*\}.
\]
This truncation is exact.
The branch that does not cache \(v\) adds nonnegative child counts,
so a parent state with count at most \(N^*\) never needs a child
count above \(N^*\).
The branch that caches \(v\) blocks the entire subtree and does
not consult child states.

\begin{theorem}[Exact dynamic program for unique-path instances]
\label{thm:unique-path-dp}
Let \(G_q\) satisfy the unique-path condition of
Lemma~\ref{lem:unique-path-intree}, let \(c\) be nonnegative vertex
costs, and let \(N^*\) be a residual threshold; write
\(K_v:=\min\{|L_v|,N^*\}\) for each vertex \(v\).  Then the values
\(F(v,e)\) of \eqref{eq:dp-state} with \(0\le e\le K_v\) satisfy the
recurrence
\begin{enumerate}[label=\textup{(\roman*)}]
\item if \(v\in V_0(G_q)\) is a base leaf, then
\[
  F(v,0)=c(v),
  \qquad
  F(v,1)=0\ \ (\text{when }K_v\ge1);
\]
\item if \(v\notin V_0(G_q)\), including the target \(q\), with
\(\operatorname{par}(v)=(u_1,\ldots,u_{a(v)})\), then
\begin{equation}
  F(v,e)
  =
  \min\left\{
    A_v(e),\
    \min_{\substack{
      0\le e_i\le K_{u_i}\ (1\le i\le a(v))\\
      \sum_{i=1}^{a(v)}e_i=e
    }}
    \sum_{i=1}^{a(v)}F(u_i,e_i)
  \right\},
  \label{eq:dp-recurrence}
\end{equation}
where the option of caching \(v\) itself is
\[
  A_v(e):=
  \begin{cases}
    c(v), & e=0\ \text{and}\ v\neq q,\\
    +\infty, & \text{otherwise};
  \end{cases}
\]
\end{enumerate}
and the minimum cache cost leaving at most \(N^*\) leaves
exposed is
\begin{equation}
  \mathrm{OPT}
  =
  \min_{0\le e\le K_q}F(q,e).
  \label{eq:dp-answer}
\end{equation}
For the choice \(N^*=N^*(\varepsilon,\delta)\), this is the exact
answer-excluding single-query optimum
\(\sigma^*_{\mathrm{sem}}(q,\varepsilon,\delta)\) of
Theorem~\ref{thm:exact-single-query-opt}.
The truncated table is computable in time
\[
  O\bigl(|V(G_q)|\,\bigl(\min\{\kappa_q,N^*\}+1\bigr)^{2}\bigr)
\]
and space \(O\bigl(|V(G_q)|\,\bigl(\min\{\kappa_q,N^*\}+1\bigr)\bigr)\), and an
optimal cache is recovered by standard backpointers.  In particular, the
decision version of PPI is polynomial-time solvable on canonical
unique-path instances.
\end{theorem}

\begin{proof}
We prove by induction on the in-tree, in a topological order
(sources first), that the right-hand side of the recurrence equals
\eqref{eq:dp-state}.

\emph{Base case.}  If \(v\in V_0(G_q)\), then \(T_v=\{v\}\) and
\(L_v=\{v\}\).  The cache \(S=\{v\}\) blocks \(v\) at cost \(c(v)\),
giving zero \(v\)-exposed leaves; the empty cache costs \(0\) and leaves
\(v\) exposed.  No other count is attainable, which gives (i).

\emph{Induction step.}  Let \(v\) be internal with children
\(u_1,\ldots,u_{a(v)}\), and let \(S\subseteq T_v\setminus\{q\}\).

If \(v\in S\), then the path from every \(b\in L_v\) to \(v\) ends at
\(v\in S\), so no leaf is \(v\)-exposed.  Moreover
\(\ell_c(S)\ge c(v)\), with equality at \(S=\{v\}\) since all costs are
nonnegative.  Hence the minimum cost over caches containing \(v\) is exactly
\(c(v)\), attained at count \(0\) only; this option is unavailable when
\(v=q\), since the answer-excluding class excludes \(q\).

If \(v\notin S\), then \(S\) is the disjoint union of
\(S_i:=S\cap T_{u_i}\), \(i=1,\ldots,a(v)\), because the subtrees
partition \(T_v\setminus\{v\}\).  The unique path from a leaf
\(b\in L_{u_i}\) to \(v\) traverses the unique path from \(b\) to \(u_i\)
followed by the edge \((u_i,v)\); since \(v\notin S\), the leaf \(b\) is
\(v\)-exposed if and only if it is \(u_i\)-exposed under \(S_i\).
Consequently the \(v\)-exposed count of \(S\) is \(\sum_i e_i\), where
\(e_i\) is the \(u_i\)-exposed count of \(S_i\), and
\(\ell_c(S)=\sum_i\ell_c(S_i)\).  Minimizing independently over the
subcaches subject to \(\sum_i e_i=e\) yields the convolution term of
\eqref{eq:dp-recurrence}, and taking the better of the two cases
(\(v\in S\) versus \(v\notin S\)) yields \eqref{eq:dp-recurrence}.

\emph{Root value.}  By Lemma~\ref{lem:unique-path-intree}, each base leaf
has a single path to \(q\), so the \(q\)-exposed leaves of a cache
\(S\subseteq V(G_q)\setminus\{q\}\) are exactly
\(D_{\mathrm{exp}}(q,S)\).  Minimizing \(F(q,e)\) over \(e\le N^*\)
therefore computes the optimum of the instance, which proves
\eqref{eq:dp-answer}; the choice \(N^*=N^*(\varepsilon,\delta)\)
recovers the exact formulation \eqref{eq:exact-single-query-opt}.

\emph{Complexity.}  Write \(L:=\min\{\kappa_q,N^*\}+1\).  Each vertex
stores at most
\(K_v+1\le\min\{\kappa_q+1,N^*+1\}=L\) values.  The convolution at an
internal
vertex \(v\) is computed by a sequential knapsack over its \(a(v)\)
children: processing one child merges the current vector with the child
vector, each of length at most \(L\), in time
\(O\bigl(L^{2}\bigr)\).  Summing over all children and using
\(\sum_v a(v)=|E(G_q)|=|V(G_q)|-1\) gives the stated time
bound.
The value table uses \(O(|V(G_q)|L)\) space.
Leaf counts are computed in linear time.
For backtracking, store the minimizing split at each
child-processing stage and the selected branch at each vertex.
There are \(O(|E(G_q)|L)\) such records, so they fit within the
stated space bound.
\end{proof}

\subsubsection{Bounded Treewidth}
\label{subsec:bounded-treewidth}

Treewidth provides a tractable restriction independent of
path multiplicity.
The following algorithm permits arbitrarily many dependency
paths and an input residual threshold.

\begin{theorem}[PPI parameterized by treewidth]
\label{thm:bounded-treewidth-fpt}
The optimization version of
PPI (Definition~\ref{def:ppi}) with nonnegative rational costs
of polynomial bit length, on instances
\((G,V_{0},q,c,N^{*})\) whose underlying undirected graph has treewidth
at most \(w\), can be solved by an explicit dynamic program in time
\[
  O\bigl(3^{\,w+1}\,(K+1)^{2}\,(w+1)\,|V(G)|\bigr),
  \qquad
  K:=\min\{N^{*},|V_{0}|\},
\]
counting arithmetic operations on the cost encoding, once a tree
decomposition of width at most \(w\) is known.  
Here the supplied decomposition is assumed to have \(O(|V(G)|)\)
bags; otherwise, its encoding size must also be included in the
preprocessing time. Since
\(K\le|V_{0}|\le|V(G)|\), PPI is fixed-parameter tractable
parameterized by the treewidth alone, with elementary parameter
dependence, and is polynomial-time solvable on every class of instances
of bounded treewidth, even with the threshold part of the input.
Without a given decomposition, one of width at most \(5w+4\) is
computable in time \(2^{O(w)}\cdot|V(G)|\) \citep{bodlaender2016c},
which preserves the same FPT conclusion; an exact width-\(w\)
decomposition is computable in linear time for every fixed \(w\)
\citep{bodlaender1996linear}.
\end{theorem}

\begin{proof}
Let \(\mathcal A\) be the admissible set, with
\(\mathcal A=V(G)\setminus\{q\}\) by default.
We express feasibility through a local three-state
certificate and optimize over these certificates
on a nice tree decomposition.

\smallskip\noindent
\emph{Cut-pair formulation.}
A pair \((S,R)\), where
\[
  S\subseteq\mathcal A,
  \qquad
  R\subseteq V(G)\setminus S,
\]
is a \emph{cut pair} if
\begin{equation}
  q\in R,
  \qquad
  (u,v)\in E(G),\ v\in R
  \ \Longrightarrow\ u\in S\cup R.
  \label{eq:cut-pair}
\end{equation}
For a cache \(S\), let \(\mathcal R(S)\) be the vertices that reach
\(q\) in \(G-S\).
Then \((S,\mathcal R(S))\) is a cut pair and
\[
  \mathcal R(S)\cap V_0=\mathrm{Exp}(S).
\]
Conversely, for any cut pair \((S,R)\), backward induction along
an \(S\)-avoiding path ending at \(q\) shows that every vertex of
the path belongs to \(R\).
Thus
\[
  \mathrm{Exp}(S)\subseteq R\cap V_0.
\]
It follows that
\begin{equation}
  \mathrm{OPT}
  =
  \min\left\{
    c(S):
    (S,R)\text{ is a cut pair},\
    |R\cap V_0|\le N^*
  \right\}.
  \label{eq:cut-pair-opt}
\end{equation}
The set \(R\) is a certificate containing the actual reachable
set; its vertices need not themselves reach \(q\).

\smallskip\noindent
\emph{States.}
Use a rooted nice tree decomposition with empty root and leaf
bags, introduce-vertex, introduce-edge, forget, and join nodes.
Its width is at most \(w\), and it has
\(O((w+1)|V(G)|)\) nodes
\citep{cygan2015parameterized}.
Each arc is introduced exactly once at a bag containing its
endpoints.  The number of introduce-edge nodes is
\(O(|E(G)|)=O(w\,|V(G)|)\), since the underlying simple graph
has \(O(w\,|V(G)|)\) edges, so the total node count
remains \(O((w+1)\,|V(G)|)\).

For a decomposition node \(t\), let \(X_t\) be its bag,
\(V_t\) the union of bags below \(t\), including \(t\), and \(E_t\)
the arcs introduced below \(t\), including at \(t\).
A bag labeling
\[
  \sigma:X_t\to\{\mathrm S,\mathrm R,\mathrm O\}
\]
records membership in \(S\), in \(R\), or in neither set.
Whenever \(q\in X_t\), require \(\sigma(q)=\mathrm R\), and
require \(\sigma(v)\neq\mathrm S\) for every
\(v\in X_t\setminus\mathcal A\).

Write \(K:=\min\{N^*,|V_0|\}\).
For \(0\le k\le K\), let \(\Phi_t(\sigma,k)\) be the
minimum \(c(S_t)\) over disjoint sets \(S_t,R_t\subseteq V_t\)
such that
\[
  S_t\subseteq\mathcal A,
\]
the sets induce \(\sigma\), satisfy
\[
  q\in R_t\quad\text{whenever }q\in V_t,
\]
satisfy the arc condition of \eqref{eq:cut-pair} for every
arc in \(E_t\), and have
\[
  |R_t\cap V_0|=k.
\]
Infeasible entries have value \(+\infty\).
The counter and cost include bag vertices.
Counts above \(K\) can be discarded because the counter
\(|R_t\cap V_0|\) is nondecreasing along the decomposition (join
nodes add the children's counts minus the shared bag contribution,
which is itself included in each child's count), so a partial
solution with count above \(K\) never extends to a feasible one.

\smallskip\noindent
\emph{Transitions.}
At an empty leaf node, \(\Phi_t(\varnothing,0)=0\), and all other
entries are infeasible.
An introduce-vertex node adds the new vertex's cost and counter
contribution; an introduce-edge node for \(u\to v\) discards
precisely the entries with \(\sigma(u)=\mathrm O\) and
\(\sigma(v)=\mathrm R\); a forget node minimizes over the three
labels of the removed vertex.
At a join node with children \(t_1,t_2\), define
\[
  a(\sigma)
  :=
  |\{v\in X_t\cap V_0:\sigma(v)=\mathrm R\}|,
  \qquad
  b(\sigma)
  :=
  \sum_{\substack{v\in X_t\\\sigma(v)=\mathrm S}}c(v),
\]
and, since the child vertex sets intersect exactly in \(X_t\),
\[
  \Phi_t(\sigma,k)
  =
  \min_{\substack{
    0\le k_1,k_2\le K\\
    k_1+k_2-a(\sigma)=k
  }}
  \left[
    \Phi_{t_1}(\sigma,k_1)
    +\Phi_{t_2}(\sigma,k_2)
    -b(\sigma)
  \right];
\]
the subtractions remove the duplicated bag contributions.
The complete transition equations are collected in
\ref{app:treewidth-transitions}.
They preserve the stated invariant by induction on the
decomposition, and at the empty root the answer is
\[
  \min_{0\le k\le K}\Phi_{\mathrm{root}}(\varnothing,k),
\]
which equals \eqref{eq:cut-pair-opt}.

Each bag has at most \(3^{w+1}\) labelings.
The bag contributions \(a(\sigma)\) and \(b(\sigma)\) can be
precomputed over all labelings by incremental enumeration:
assigning one additional vertex updates each contribution in
constant time, taking \(O(3^{|X_t|})\) arithmetic operations per
bag.
Join nodes require \(O((K+1)^2)\) work per labeling, and the
other transitions require no more.
The resulting running time is
\[
  O\bigl(
    3^{w+1}(K+1)^2(w+1)|V(G)|
  \bigr).
\]
Backpointers recover an optimal cache.
\end{proof}

\begin{remark}[Scope of the treewidth algorithm]
\label{rem:tw-hardness-compatible}
The cut-pair formulation also applies to directed graphs
with cycles; acyclicity is needed only for the derivation
interpretation.

The depth-two hardness instances have unbounded treewidth:
after deleting \(q\), contracting the edge-subdivision
vertices recovers the source graph as a minor.
Thus bounded treewidth and bounded path multiplicity give
different restrictions.
\end{remark}

\subsubsection{Individual Coverage: An Exact Postdominator
Dynamic Program}

\begin{theorem}[Postdominator dynamic programming under
individual coverage]
\label{thm:dominator-dp}
Let \((G,V_0,q,c,N^*)\) be a PPI instance in which
every vertex reaches \(q\) and every designated terminal
is a source.
Let
\[
  \mathcal A\subseteq V(G)\setminus\{q\}
\]
be the admissible set, with
\(\mathcal A=V(G)\setminus\{q\}\) by default.
For \(v\in V(G)\), define
\[
  U(v):=
  \{b\in V_0:
    \text{every \(b\)-to-\(q\) path contains \(v\)}\}.
\]
Suppose that, for every \(S\subseteq\mathcal A\),
\[
  V_0\setminus\mathrm{Exp}(S)
  =
  \bigcup_{v\in S}U(v).
\]

Then \(U(v)\) is the set of designated terminals in the
subtree rooted at \(v\) in the postdominator tree.
The optimum is computable in time
\[
  T_{\mathrm{dom}}(G)
  +O\bigl(|V(G)|(t+1)^2\bigr),
  \qquad
  t:=(|V_0|-N^*)^+,
\]
where \(T_{\mathrm{dom}}(G)\) is the time needed to
compute that tree.
\end{theorem}

\begin{proof}
Compute the dominator tree of the reversed graph rooted at \(q\),
equivalently the postdominator tree \(T\) of \(G\)
\citep{lengauer1979fast}.
A vertex \(v\) postdominates \(b\) exactly when \(v\) is an
ancestor of \(b\) in \(T\).
Thus \(U(v)\) consists of the designated leaves in the subtree
\(T_v\), proving laminarity.

If \(t=0\), the empty cache is optimal.
Assume \(t\ge1\).
For
\[
  0\le k\le\min\{t,|U(v)|\},
\]
let \(Z(v,k)\) be the minimum cost of a selection from
\(T_v\cap\mathcal A\) whose union of individual blocking sets
contains at least \(k\) members of \(U(v)\).
Set \(Z(v,0)=0\).

For \(k\ge1\), define the branch selecting \(v\) by
\[
  A_v(k):=
  \begin{cases}
    c(v), & v\in\mathcal A,\\
    +\infty, & v\notin\mathcal A.
  \end{cases}
\]
Let \(\operatorname{ch}_T(v)\) be the children of \(v\) in \(T\).
The recurrence is
\[
  Z(v,k)
  =
  \min\left\{
    A_v(k),\
    \min_{\substack{
      0\le k_u\le\min\{t,|U(u)|\}\\
      \sum_{u\in\operatorname{ch}_T(v)}k_u=k
    }}
    \sum_{u\in\operatorname{ch}_T(v)}Z(u,k_u)
  \right\},
\]
where a minimum over an empty feasible index set is \(+\infty\).

Selecting an admissible vertex \(v\) covers all of \(U(v)\), and additional
selections cannot lower its cost because costs are nonnegative.
If \(v\) is not selected, coverage is supplied by its child
subtrees, whose marked sets are pairwise disjoint.
When \(v\in V_0\), it is a source of \(G\), so it postdominates
no distinct member of \(V_0\); hence \(U(v)=\{v\}\), and no child
subtree contains a designated leaf.
The element \(v\) itself cannot be covered by selecting a proper
descendant.

To justify the exact quota sum, suppose child selections cover
\(h_u\) leaves with \(\sum_u h_u\ge k\).
There exist integers \(0\le k_u\le h_u\), each at most \(t\),
whose sum is \(k\).
Each child selection is feasible for its assigned quota.
Conversely, selections meeting quotas summing to \(k\) cover at
least \(k\) distinct leaves.
This proves the recurrence by induction.

Under individual coverage, graph feasibility is equivalent to
covering at least \(t\) leaves by these sets.
Therefore the optimum is \(Z(q,t)\), with the root selection
branch disabled.

A sequential convolution over the children takes
\(O((t+1)^2)\) time per tree edge.
Together with the postdominator computation, this gives
the stated running time.
Backpointers recover an optimal admissible cache.
\end{proof}

\begin{remark}[Checking individual coverage]
\label{rem:ic-criterion}
For each \(b\in V_0\), define
\[
  S_b:=
  \{v\in V(G)\setminus\{q\}:b\notin U(v)\}.
\]
The individual-coverage condition holds if and only if
\(b\) reaches \(q\) in \(G-S_b\) for every \(b\in V_0\).
Indeed, failure of this test gives a cache that blocks \(b\)
without an individual blocker.
Conversely, any such cache is contained in \(S_b\), so it
cannot block \(b\) if the path in \(G-S_b\) survives.
The sets \(S_b\) and the reachability tests are computable
in polynomial time.
With forbidden vertices, restrict \(S_b\) to admissible
candidates.

The condition is strictly more general than unique paths.
For example, the DAG with arcs
\(b\to u\), \(b\to v\), \(v\to u\), and \(u\to q\)
has two \(b\)-to-\(q\) paths.
Its individual blockers are \(b\) and \(u\); the remaining
candidate \(v\) cannot block \(b\).
Figure~\ref{fig:postdom} shows this example.
\end{remark}

\begin{figure}[tbp]
\centering
\begin{tikzpicture}[ppi]
\node[leaf]  (b) at (0,0)      {$b$};
\node[inter] (v) at (1.3,-0.8) {$v$};
\node[inter] (u) at (2.6,0)    {$u$};
\node[query] (q) at (3.9,0)    {$q$};
\draw[arc] (b)--(u); \draw[arc] (b)--(v);
\draw[arc] (v)--(u); \draw[arc] (u)--(q);
\node[note] at (1.95,-1.6) {(a) two \(b\)-to-\(q\) paths};
\begin{scope}[xshift=6.2cm]
\node[query] (tq) at (0.8,0.9)  {$q$};
\node[inter] (tu) at (0.8,0)    {$u$};
\node[leaf]  (tb) at (0.2,-0.9) {$b$};
\node[inter] (tv) at (1.4,-0.9) {$v$};
\draw[thick,black!60] (tq)--(tu) (tu)--(tb) (tu)--(tv);
\node[note] at (0.8,-1.6) {(b) postdominator tree};
\end{scope}
\end{tikzpicture}
\caption{Individual coverage beyond unique paths
(Remark~\ref{rem:ic-criterion}).
The postdominators of \(b\) are \(b,u,q\), its ancestors
in the tree, including itself.
Because \(q\) is protected, the admissible individual
blockers are \(b\) and \(u\); the vertex \(v\) blocks
no designated terminal.
The dynamic program of
Theorem~\ref{thm:dominator-dp} runs on the tree in~(b).}
\label{fig:postdom}
\end{figure}

\begin{remark}[Joint workloads]
\label{rem:dominator-workload}
Add a super-sink \(q_0\) and arcs \(q_\ell\to q_0\)
to the workload union.
A leaf is workload-exposed exactly when it reaches \(q_0\)
after deletion.
If individual coverage holds for admissible caches in this
graph, Theorem~\ref{thm:dominator-dp} applies with selection
of \(W\cup\{q_0\}\) disabled and
\(t=(|D_\cup|-N^*)^+\).
\end{remark}

\subsection{Parameterization by the Residual Threshold}
\label{subsec:threshold-param}

Write \(k:=N^*\).
We first give an XP algorithm and matching weighted ETH
exponent, then identify a bounded-fan-out FPT regime.
For unit costs, we derive a rank formulation, a depth-two
polynomial-time subclass, and depth-three W[1]-hardness.
The remaining threshold question concerns depth two.

\begin{proposition}[Residual-threshold enumeration]
\label{prop:xp-threshold}
For a PPI instance,
\[
  \mathrm{OPT}
  =
  \min_{\substack{A\subseteq V_0\\|A|\le N^*}}
  \mathrm{MinVertexCut}(V_0\setminus A,q),
\]
where \(\mathrm{MinVertexCut}(Y,q)\) denotes the minimum
cost of an admissible vertex set separating \(Y\) from
\(q\).
The separator may contain admissible vertices of \(Y\);
its value is \(+\infty\) if no admissible separator exists.
Each inner problem is polynomial-time solvable by node
splitting and maximum flow.
With \(K:=\min\{N^*,|V_0|\}\), the enumeration uses
\[
  \sum_{i=0}^{K}\binom{|V_0|}{i}
\]
cuts.
Thus PPI is in XP parameterized by \(N^*\),
and is polynomial-time solvable for fixed reliability
parameters \((\varepsilon,\delta)\).
\end{proposition}

\begin{proof}
If \(N^*\ge|V_0|\), the empty cache is optimal and no
enumeration is needed; assume henceforth \(N^*<|V_0|\).
Let \(S\) be feasible and set \(A:=\mathrm{Exp}(S)\).
Then \(|A|\le N^*\), and every vertex of
\(V_0\setminus(A\cup S)\) is separated from \(q\) in \(G-S\).
Therefore
\[
  c(S)\ge
  \mathrm{MinVertexCut}(V_0\setminus A,q).
\]
Conversely, a cut for \(V_0\setminus A\), with \(|A|\le N^*\),
has exposed set contained in \(A\), so it is feasible.
Taking minima proves the identity.

For completeness, each inner cut problem has the following
node-splitting representation.
Let \(\mathcal A\) be the admissible set and put
\[
  C_\infty:=1+\sum_{v\in\mathcal A}c(v).
\]
Replace every vertex \(v\) by \(v^{\mathrm{in}}\) and
\(v^{\mathrm{out}}\).
Give \(v^{\mathrm{in}}\to v^{\mathrm{out}}\) capacity \(c(v)\) if
\(v\in\mathcal A\), and capacity \(C_\infty\) otherwise (in
particular for \(q\)).
Replace every original arc \(u\to v\) by
\(u^{\mathrm{out}}\to v^{\mathrm{in}}\) of capacity \(C_\infty\).
Connect a new super-source to \(x^{\mathrm{in}}\) with capacity
\(C_\infty\) for each required terminal \(x\), and use
\(q^{\mathrm{out}}\) as the sink.

Every cut of capacity less than \(C_\infty\) consists only of
splitting arcs of admissible vertices, and conversely every
admissible separating set yields such a cut of equal capacity.
Hence the inner problem is feasible if and only if the minimum
cut value is less than \(C_\infty\), and in that case a minimum
cut is a minimum-cost admissible separator.
Connecting the super-source to \(x^{\mathrm{in}}\) allows an
admissible terminal \(x\) itself to be deleted.

Enumerating all allowed sets \(A\) and solving these polynomial-time
cut problems gives the stated XP bound.
\end{proof}

\begin{theorem}[W[1]-hardness parameterized by the residual threshold]
\label{thm:ppi-wh}
The decision version of PPI is W[1]-hard parameterized by
the residual threshold \(k:=N^*\), even on canonical instances of depth
two with polynomially bounded integer costs in which every intermediate
vertex has in-degree two.
Moreover, assuming the Exponential Time Hypothesis (ETH),
PPI admits no \(f(k)\,|I|^{o(k)}\)-time algorithm for any
computable function \(f\), so the XP bound of
Proposition~\ref{prop:xp-threshold} is essentially tight.
\end{theorem}

\begin{proof}
Reduce from \textsc{Clique}.
Let \(H=(V_H,E_H)\) be a simple graph without isolated vertices
(isolated vertices can be deleted without changing the answer for
\(k\ge2\); the cases \(k\le1\) are trivial), write \(m:=|E_H|\), and
let \(k\ge2\).
If \(m<\binom{k}{2}\), the source instance is a NO instance and is
mapped to a fixed NO instance of PPI; henceforth assume
\(m\ge\binom{k}{2}\), so the budget defined below is nonnegative.
Create a leaf \(\ell_v\) of cost \(m+1\) for each \(v\in V_H\) and an
intermediate vertex \(w_e\) of cost \(1\) for each \(e=\{u,v\}\in E_H\),
with arcs \(\ell_u\to w_e\), \(\ell_v\to w_e\), and \(w_e\to q\).
Every vertex reaches \(q\) and the leaves are exactly the sources, so
by Lemma~\ref{lem:horn-realization} this is a canonical instance with
positive integer costs.
Set the budget \(\theta:=m-\binom{k}{2}\) and the residual threshold
\(N^*:=k\).
The construction has depth two, polynomial size, polynomially bounded
integer costs, and parameter exactly \(k\).

If \(Q\subseteq V_H\) is a \(k\)-clique, take
\(S:=\{w_e:e\not\subseteq Q\}\).
Its cost is \(m-\binom{k}{2}=\theta\).
A leaf \(\ell_v\) retains a cache-free path to \(q\) only through an
uncut vertex \(w_e\) with \(v\in e\subseteq Q\), so every exposed
leaf lies in \(\{\ell_v:v\in Q\}\) and at most \(k\) leaves are
exposed.

Conversely, let \(S\) be feasible with \(c(S)\le\theta\).
Since \(\theta<m+1\), the cache \(S\) contains no leaf.
At least \(\binom{k}{2}\) edge vertices survive, and both parents of
every surviving edge vertex are exposed.
Feasibility bounds the exposed leaves by \(k\), so at most \(k\)
vertices of \(H\) span at least \(\binom{k}{2}\) edges and therefore
form a \(k\)-clique.

The reduction is parameter-preserving and of polynomial size, so
W[1]-hardness follows from that of \textsc{Clique}
\citep{downey2013fundamentals}, and the ETH lower bound
\citep[Theorem~14.21]{cygan2015parameterized}
\citep{chen2006strong} transfers.
Finally, the residual thresholds used in the reduction can also be
realized by rational reliability parameters of polynomial encoding
length.  For \(k\ge1\), take
\[
  \varepsilon=\frac12,
  \qquad
  \delta=1-2^{-k}.
\]
Then \(N^*(\varepsilon,\delta)=k\) by \eqref{eq:model-N-star}, and
the binary encoding length of the parameters is \(O(k)\).
For \(k=0\), the choice
\(\varepsilon=\frac12\), \(\delta=\frac14\)
gives \(N^*(\varepsilon,\delta)=0\).
\end{proof}

\begin{remark}[Why the leaf costs matter]
\label{rem:wh-cost-model}
Theorem~\ref{thm:ppi-wh} prices leaves above the budget.
With unit costs, deleting a leaf provides an additional
affordable option, so the same reduction does not apply.
Indeed, its intermediate in-degree-two graph class becomes
linear-time solvable under unit costs
(Corollary~\ref{cor:indegree-two-poly}).
\end{remark}

\begin{theorem}[Bounded leaf fan-out restores tractability in \(N^*\)]
\label{thm:ppi-fpt}
Consider a canonical PPI instance of depth at most two
(Definition~\ref{def:depth}) with source set \(V_0\), intermediate
set \(V_1\), and protected sink \(q\).
Let \(\mathcal A\subseteq V_1\) be the admissible set, so no leaf is
deletable.
Suppose each leaf has at most \(p\) intermediate out-neighbors.
With \(k:=N^*\), a minimum-cost feasible cache, or a certificate of
infeasibility, can be computed in time \(2^{O(pk)}\,|I|^{O(1)}\):
PPI is fixed-parameter tractable parameterized by \(N^*+p\) on this
class.
The corresponding decision problem is solvable
within the same time bound when every deletable
leaf has cost greater than the budget.
\end{theorem}

\begin{proof}
Let \(\mathcal A\subseteq V_1\) be the admissible set,
and put \(k:=N^*\).
If \(k\ge|V_0|\), the empty cache is optimal.

A terminal is forced to remain exposed if it has a
direct arc to \(q\) or is a parent of an inadmissible
intermediate vertex.
Let
\[
  F:=
  \{b\in V_0:(b,q)\in E(G)\}
  \cup
  \bigcup_{w\in V_1\setminus\mathcal A}N^-(w).
\]
Every member of \(F\) is exposed under every admissible
cache.
If \(|F|>k\), the instance is infeasible.
Otherwise remove these terminals and replace \(k\) by \(k-|F|\).

Keep, and remove from the residual optimization, every
intermediate vertex with no remaining terminal parent.
Such a vertex can expose only terminals already counted
in \(F\), so selecting it is unnecessary.
This removes every inadmissible intermediate vertex.

Write \(V_0'\) and \(V_1'\) for the remaining terminal and
intermediate sets.
Every vertex of \(V_1'\) is admissible and has a nonempty
remaining parent set \(P(w)\subseteq V_0'\).
Each terminal has at most \(p\) neighbors in \(V_1'\).
If \(V_1'=\varnothing\), the empty cache is optimal.
If \(k=0\), every vertex of \(V_1'\) must be cached, which
is optimal.
Henceforth \(p\ge1\) and \(k\ge1\).

A cache \(S\subseteq V_1'\) exposes exactly the terminals incident
to \(V_1'\setminus S\).
Choose an optimum \(S^*\), and write
\[
  V_1^{\mathrm{keep}}:=V_1'\setminus S^*,
  \qquad
  X:=\bigcup_{w\in V_1^{\mathrm{keep}}}P(w).
\]
Thus \(|X|\le k\).
By nonnegativity of costs, we may normalize the optimum so that
\[
  V_1^{\mathrm{keep}}=\{w\in V_1':P(w)\subseteq X\}.
\]
Indeed, restoring any intermediate vertex whose parents all
belong to \(X\) creates no new exposed terminal.

All neighborhoods and components below are taken in the
underlying undirected bipartite graph on \(V_0'\cup V_1'\).
Since every member of \(V_1^{\mathrm{keep}}\) has a parent in \(X\),
\[
  |V_1^{\mathrm{keep}}|\le |N(X)|\le p|X|\le pk.
\]
Put
\[
  \partial:=N(X)\setminus V_1^{\mathrm{keep}},
  \qquad
  Z:=X\cup V_1^{\mathrm{keep}}.
\]
Then \(|\partial|\le pk\).
Figure~\ref{fig:fanout-coloring} illustrates the sets \(X\),
\(V_1^{\mathrm{keep}}\), and \(\partial\).

Color vertices independently green or red.
Since \(V_1^{\mathrm{keep}}\subseteq N(X)\), we have
\(|Z\cup\partial|=|X\cup N(X)|\le k+pk\); hence, with probability at
least \(2^{-(k+pk)}\), every vertex of
\(Z\) is green and every vertex of \(\partial\) is red.
Under this event, \(Z\) is a union of green connected
components.
Every neighbor of a vertex of \(X\) lies in \(V_1^{\mathrm{keep}}\cup\partial\),
and every parent of a vertex of \(V_1^{\mathrm{keep}}\) lies in \(X\).

For each green component \(C\), define
\[
  a(C):=|C\cap V_0'|,
  \qquad
  b(C):=
  \sum_{\substack{w\in C\cap V_1'\\P(w)\subseteq C}}c(w).
\]
Use a knapsack dynamic program to choose components of total
weight at most \(k\) and maximum total value.
For the selected component family \(\mathcal C\), keep exactly
\[
  V_{1,\mathcal C}^{\mathrm{keep}}
  :=
  \bigcup_{C\in\mathcal C}
  \{w\in C\cap V_1':P(w)\subseteq C\},
\]
and cache every vertex of \(V_1'\setminus V_{1,\mathcal C}^{\mathrm{keep}}\).

Every exposed terminal is a parent of a kept intermediate
vertex and therefore lies in a selected component.
The output is consequently feasible.
Its cost is exactly
\[
  c(V_1')-\sum_{C\in\mathcal C}b(C).
\]
For a coloring satisfying the stated event, the components
forming \(Z\) have total weight \(|X|\le k\) and total value
\(c(V_1^{\mathrm{keep}})\).
Thus the knapsack solution has cost at most
\(c(V_1')-c(V_1^{\mathrm{keep}})=c(S^*)\).
Since every output is feasible, it is optimal.

To derandomize, use random separation \citep{cai2006random}:
enumerate an \((N,s)\)-universal family with
\[
  N:=|V_0'|+|V_1'|,
  \qquad
  s:=\min\{N,k+pk\}.
\]
A deterministic construction produces such a family of size
\(2^s s^{O(\log s)}\log N\) in \(2^{O(s)}N^{O(1)}\) time
\citep{naor1995splitters}, and the family realizes the required
coloring on \(Z\cup\partial\).
Run the component computation and the knapsack step for every
coloring, and return an output of minimum cost.
Every output is feasible, and at least one coloring realizes the
required assignment on \(Z\cup\partial\), so the minimum-cost
output is optimal.
Intermediate vertices set aside during preprocessing remain
uncached.
For \(p\ge1\), the total running time is
\(2^{O(pk)}|I|^{O(1)}\).
For the stated decision variant, leaves whose costs exceed the
budget cannot belong to a feasible cache.
Treat them as inadmissible, compute the optimum over
intermediate vertices, and compare its cost with the budget.
\end{proof}

\begin{figure}[tbp]
\centering
\begin{tikzpicture}[ppi,x=1.15cm,y=1cm,
  gr/.style={fill=green!30},rd/.style={fill=red!30},
  un/.style={fill=black!6}]
\node[leaf,un] (b1) at (1,0) {$b_1$};
\node[leaf,gr] (b2) at (2,0) {$b_2$};
\node[leaf,gr] (b3) at (3,0) {$b_3$};
\node[leaf,gr] (b4) at (4,0) {$b_4$};
\node[leaf,un] (b5) at (5,0) {$b_5$};
\node[leaf,un] (b6) at (6,0) {$b_6$};
\node[leaf,un] (b7) at (7,0) {$b_7$};
\node[inter,rd,label={[note]left:$\partial$}]  (w1) at (0.45,1.5) {$w_1$};
\node[inter,gr] (w2) at (2.5,1.5) {$w_2$};
\node[inter,gr] (w3) at (3.5,1.5) {$w_3$};
\node[inter,rd,label={[note]left:$\partial$}] (w4) at (5.0,1.5) {$w_4$};
\node[inter,un] (w5) at (5.8,1.5) {$w_5$};
\node[inter,un] (w6) at (6.8,1.5) {$w_6$};
\node[query] (q) at (4,2.9) {$q$};
\foreach \i/\j in {1/1,2/1,2/2,3/2,3/3,4/3,4/4,5/4,5/5,6/5,6/6,7/6}
  \draw[arc] (b\i)--(w\j);
\foreach \j in {1,...,6} \draw[arc,black!25] (w\j)--(q);
\begin{scope}[on background layer]
\node[region, fill=green!6, fit=(b2)(b4)(w2)(w3),
      inner ysep=3pt, inner xsep=3pt+0.5mm] (Z) {};
\end{scope}
\node[note,anchor=east] at ($(Z.west)+(0.77,0.70)$)
  {\(Z=\) \\ \(X\cup K^*\)};
\draw[decorate,decoration={brace,mirror,amplitude=4pt}]
  ($(b2.south west)+(0,-0.25)$)--($(b4.south east)+(0,-0.25)$)
  node[midway,below=5pt,note]{\(X\): terminals exposed by \(S^*\),
  \(|X|\le k\)};
\end{tikzpicture}
\caption{Random separation in Theorem~\ref{thm:ppi-fpt}, with \(k=3\)
and unit intermediate costs.  The cache \(S^*=\{w_1,w_4,w_5,w_6\}\) is
optimal: any three kept intermediate vertices expose at least four
leaves.  Its kept vertices \(V_1^{\mathrm{keep}}=\{w_2,w_3\}\) have all
parents in the exposed set \(X\), and \(\partial=\{w_1,w_4\}\) are the
remaining neighbors of \(X\).  A coloring that makes \(Z\) green and
\(\partial\) red, with gray vertices arbitrary, turns \(Z\) into a union
of green components; since \(|Z\cup\partial|\le(p+1)k\), a knapsack
over green components recovers an optimal cache.}
\label{fig:fanout-coloring}
\end{figure}

\begin{proposition}[Bounded fan-out alone does not suffice]
\label{prop:fanout-depth}
PPI parameterized by \(N^*\) is W[1]-hard, and admits no
\(f(N^*)\,|I|^{o(N^*)}\)-time algorithm under ETH, even on canonical
instances with polynomially bounded positive integer costs in which
every leaf has out-degree one and every vertex other than \(q\) has
out-degree at most two.  The same holds on canonical instances of
depth three in which every leaf has out-degree one.
\end{proposition}

\begin{proof}
Start from the construction of Theorem~\ref{thm:ppi-wh} and, for each
\(v\in V_H\), replace the arcs \(\{\ell_v\to w_e:e\ni v\}\) as
follows.  Let \(A_v\) be a rooted out-arborescence of maximum
out-degree two with root \(y_v\) and exactly \(\deg_H(v)\) leaves
\(z_{v,e}\) (\(e\ni v\)); add \(\ell_v\to y_v\) and
\(z_{v,e}\to w_e\).  Every vertex of \(A_v\) costs \(m+1\).  For the
depth-three variant, take instead \(A_v\) to be the single vertex
\(y_v\) with arcs \(y_v\to w_e\) for all \(e\ni v\).

The budget \(\theta=m-\binom{k}{2}<m+1\) still forbids every vertex
other than the \(w_e\), and \(\ell_v\) retains a cache-free path to
\(q\) if and only if some \(w_e\) with \(e\ni v\) is uncached, exactly
as before.  Hence the YES/NO argument of Theorem~\ref{thm:ppi-wh}
applies verbatim.  The instance has polynomial size, every vertex
reaches \(q\), the leaves are the sources, and the stated degree
bounds hold.
\end{proof}

\begin{theorem}[Residual-threshold hardness with unit
costs at depth three]
\label{thm:unit-wh}
Unit-cost PPI is W[1]-hard parameterized by the residual
threshold \(N^*\), even on canonical instances of depth three in
which every leaf has out-degree one.  Assuming ETH, it admits no
\(f(N^*)\,|I|^{o(\sqrt{N^*})}\)-time algorithm.
\end{theorem}

\begin{proof}
Reduce from \textsc{Clique}.
Instances with \(\kappa<4\) can be decided directly and
mapped to fixed YES or NO instances.
For \(\kappa\ge4\), delete isolated vertices.
If \(m:=|E_H|<\binom{\kappa}{2}\), map the instance to a
fixed NO instance.
Thus assume that \(H=(V_H,E_H)\) has no isolated vertices
and \(m\ge\binom{\kappa}{2}\).  Put \(k:=\binom{\kappa}{2}\) and
\(\beta:=m+\kappa+1\).  For each \(v\in V_H\) create a hub \(h_v\)
and \(\beta\) copy leaves \(a_v^1,\dots,a_v^\beta\) with arcs
\(a_v^i\to h_v\).  For each \(e=\{u,v\}\in E_H\) create a leaf
\(p_e\) and an intermediate vertex \(s_e\) with arcs \(p_e\to s_e\),
\(h_u\to s_e\), \(h_v\to s_e\), and \(s_e\to q\).  The terminals
are all \(p_e\) and all \(a_v^i\); they are exactly the sources,
every vertex reaches \(q\), all costs are one, and every leaf has
out-degree one.  The depth is three and the size is polynomial.
Set \(N^*:=k\) and \(\theta:=m+\kappa-k\).

If \(Q\) is a \(\kappa\)-clique, cache
\(\{h_v:v\in Q\}\cup\{s_e:e\not\subseteq Q\}\), of cost
\(\kappa+m-k=\theta\).  A copy group of \(v\in Q\) is blocked at
\(h_v\); for \(v\notin Q\) every \(s_e\) with \(e\ni v\) is cached.
The only exposed leaves are the \(k\) leaves \(p_e\) with
\(e\subseteq Q\).

Conversely, let \(S\) be feasible with \(c(S)\le\theta\).
Put \(Y:=\{v:h_v\in S\}\) and \(F:=\{e:s_e\notin S\}\).
If \(e\in F\) has an endpoint \(v\notin Y\), each path
\(a_v^i\to h_v\to s_e\to q\) is cache-free unless \(a_v^i\in S\).
Feasibility then forces at least \(\beta-k=\theta+1\) cached copies,
already exceeding the budget, a
contradiction.  Hence \(F\subseteq E_H(Y)\).  The leaves \(p_e\),
\(e\in F\), are exposed unless cached, so at least \((|F|-k)^+\) of
them are cached, and
\[
  c(S)\ \ge\ |Y|+(m-|F|)+(|F|-k)^+
  \ \ge\ m+|Y|-\min\{k,e_H(Y)\}.
\]
Thus \(|Y|-\min\{k,e_H(Y)\}\le\kappa-k\).  If \(|Y|>\kappa\), the
left side is at least \(|Y|-k>\kappa-k\).  If \(Y=\varnothing\), it
is \(0>\kappa-k\).  If \(1\le|Y|<\kappa\), it is at least
\(h(|Y|)\ge h(\kappa-1)>h(\kappa)=\kappa-k\), where
\(h(x):=x-\binom x2\) is nonincreasing on the positive
integers.  Hence
\(|Y|=\kappa\) and \(e_H(Y)\ge k\), so \(Y\) is a \(\kappa\)-clique.

The parameter \(N^*=\binom{\kappa}{2}\) depends only on \(\kappa\),
which gives W[1]-hardness, and the instance size is polynomial in
\(|V_H|\).  The ETH lower bound for \textsc{Clique}
\citep[Theorem~14.21]{cygan2015parameterized}
\citep{chen2006strong} transfers with \(\kappa=\Theta(\sqrt{N^*})\).
\end{proof}

\begin{remark}[Role of the copy gadget]
\label{rem:hub-contraction}
The construction uses
\[
  \beta-N^*=\theta+1.
\]
If a hub is not selected while one of its incident edge
vertices remains, at least \(\beta-N^*\) copy leaves must
be selected.
This exceeds the budget.
Thus the copy gadget enforces the exclusion previously
implemented by large costs, while keeping every vertex
cost equal to one.
\end{remark}

Figure~\ref{fig:wh-gadgets} contrasts the two constructions.

\begin{figure}[ht]
\centering
\begin{adjustbox}{max width=\linewidth}
\begin{tikzpicture}[ppi]
\node[leaf,exposed,heavy] (l1) at (0,0)   {$\ell_1$};
\node[leaf,exposed,heavy] (l2) at (1.3,0) {$\ell_2$};
\node[leaf,exposed,heavy] (l3) at (2.6,0) {$\ell_3$};
\node[leaf,heavy]         (l4) at (3.9,0) {$\ell_4$};
\node[inter]  (w12) at (0.3,1.6) {$w_{12}$};
\node[inter]  (w13) at (1.3,1.6) {$w_{13}$};
\node[inter]  (w23) at (2.3,1.6) {$w_{23}$};
\node[cached] (w34) at (3.6,1.6) {$w_{34}$};
\node[query]  (q)   at (1.95,3.1) {$q$};
\draw[free] (l1)--(w12); \draw[free] (l2)--(w12);
\draw[free] (l1)--(w13); \draw[free] (l3)--(w13);
\draw[free] (l2)--(w23); \draw[free] (l3)--(w23);
\draw[arc]  (l3)--(w34); \draw[arc]  (l4)--(w34);
\draw[free] (w12)--(q); \draw[free] (w13)--(q); \draw[free] (w23)--(q);
\draw[arc]  (w34)--(q);
\node[note] at (1.95,-1.05) {(a) weighted, depth two:\\
  triangle \(\{1,2,3\}\) plus edge \(\{3,4\}\)\\
  \(N^*=k=3\), \(\theta=m-\binom k2=1\)};
\begin{scope}[xshift=6.3cm]
\node[leaf] (au1) at (0,0)   {$a_u^1$};
\node       at (0.75,0)      {$\cdots$};
\node[leaf] (aub) at (1.5,0) {$a_u^\beta$};
\node[leaf] (av1) at (3.3,0) {$a_v^1$};
\node       at (4.05,0)      {$\cdots$};
\node[leaf] (avb) at (4.8,0) {$a_v^\beta$};
\node[cached]       (hu) at (0.75,1.5) {$h_u$};
\node[cached]       (hv) at (4.05,1.5) {$h_v$};
\node[leaf,exposed] (pe) at (2.4,1.5)  {$p_e$};
\node[inter]        (se) at (2.4,2.8)  {$s_e$};
\node[query]        (q2) at (2.4,4.0)  {$q$};
\draw[arc] (au1)--(hu); \draw[arc] (aub)--(hu);
\draw[arc] (av1)--(hv); \draw[arc] (avb)--(hv);
\draw[arc] (hu)--(se);  \draw[arc] (hv)--(se);
\draw[free] (pe)--(se); \draw[free] (se)--(q2);
\draw[decorate,decoration={brace,mirror,amplitude=4pt}]
  (-0.3,-0.4)--(1.8,-0.4)
  node[midway,below=4pt,note]{\(\beta=m+\kappa+1\) copies};
\node[note] at (2.4,-1.35) {(b) unit costs, depth three:
  an edge \(e=\{u,v\}\) inside the clique};
\end{scope}
\end{tikzpicture}
\end{adjustbox}
\caption{The two residual-threshold hardness constructions.
(a)~Theorem~\ref{thm:ppi-wh}: leaves represent vertices of \(H\) and
are priced above the budget (double border), and intermediate vertices
represent edges.  Caching the edges outside the clique leaves exactly
its \(k\) vertices exposed.  With unit costs this gadget fails, since
caching a leaf becomes cheap
(Corollary~\ref{cor:indegree-two-poly}).
(b)~Theorem~\ref{thm:unit-wh}: each vertex receives a hub \(h_v\) fed by
\(\beta=N^*+\theta+1\) copy leaves.  If \(s_e\) and \(h_u\) are both
uncached, at most \(N^*\) copies of \(u\) may stay exposed, so at least
\(\theta+1\) must be cached, exceeding the budget.  The copies thus play
the role of the large leaf costs in~(a).}
\label{fig:wh-gadgets}
\end{figure}

\begin{corollary}[The role of leaf fan-out at depth two]
\label{cor:fanout-dichotomy}
On depth-two instances with forbidden leaves, a fixed
bound on leaf fan-out gives fixed-parameter tractability
in \(N^*\)
(Theorem~\ref{thm:ppi-fpt}).
Without such a bound, the problem is W[1]-hard in \(N^*\)
(Theorem~\ref{thm:ppi-wh}).
The depth restriction
cannot be dropped (Proposition~\ref{prop:fanout-depth}).  In
particular, the two-path class of
Proposition~\ref{prop:module-selection-np-hard} with forbidden
leaves is solvable in time \(2^{O(N^*)}|I|^{O(1)}\).  When leaves are
admissible and all costs are one, the residual-threshold
parameterization of this two-path class remains unresolved here
(Remark~\ref{rem:nstar-open}); the settled unit-cost boundary is
linear-time solvability at intermediate in-degree two
(Corollary~\ref{cor:indegree-two-poly}) versus W[1]-hardness at
depth three (Theorem~\ref{thm:unit-wh}).
\end{corollary}

\begin{proposition}[Matroid form of unit-cost PPI]
\label{prop:ppi-matroid}
Consider unit-cost PPI with every vertex other than \(q\)
admissible.
For \(Y\subseteq V_0\), let \(g(Y)\) be the minimum cardinality of a
vertex cut separating \(Y\) from \(q\), where the cut may contain
vertices of \(Y\).
Let \(r\) be the rank function of the gammoid on \(V_0\) defined by
\(G-q\) with target set \(N^{-}(q)\) \citep{ingleton1973gammoids},
and let \(r^*\) be its dual rank function.
Then \(g=r\), and with \(K:=\min\{N^*,|V_0|\}\),
\[
  \mathrm{OPT}
  \;=\;
  r(V_0)-K+
  \min_{\substack{X\subseteq V_0\\|X|=K}} r^*(X).
\]
Thus, whenever \(N^*<|V_0|\), unit-cost PPI parameterized
by \(N^*\) reduces to finding a minimum-rank \(N^*\)-element subset
of the dual of a gammoid; the identity does not require the terminals
to be sources.
\end{proposition}

\begin{proof}
By the vertex form of Menger's theorem
\citep[Cor.~9.1a]{schrijver2003combinatorial}, \(g(Y)\) equals the
maximum number of \(Y\)--\(q\) paths that are pairwise vertex-disjoint
except at \(q\) (in particular, with distinct initial vertices).
Deleting the final vertex \(q\) from such paths gives pairwise
vertex-disjoint paths in \(G-q\) from \(Y\) to \(N^{-}(q)\), and
conversely appending \(q\) to such a linkage restores the former.
Length-zero linkage paths are allowed.
Hence \(g(Y)=r(Y)\) \citep{ingleton1973gammoids,oxley2011matroid}.
Proposition~\ref{prop:xp-threshold} gives
\(\mathrm{OPT}=\min_{|X|\le K}r(V_0\setminus X)\):
at most \(N^*\) leaves may remain exposed, and no exposed set larger
than \(V_0\) arises.
By monotonicity of rank, \(X\) may be extended to size \(K\) without
increasing the objective, so the minimum is attained with \(|X|=K\).
The dual rank identity
\(r^*(X)=|X|-r(V_0)+r(V_0\setminus X)\)
then gives \(g(V_0\setminus X)=r(V_0)-K+r^*(X)\).
\end{proof}

\begin{remark}[Two special cases]
\label{rem:matroid-checks}
Direct leaf-to-\(q\) arcs give the free matroid and
\(\mathrm{OPT}=|V_0|-K\).
If all leaves feed one intermediate vertex, the matroid is
\(U_{1,|V_0|}\), and \(\mathrm{OPT}=1\) for \(K<|V_0|\).
\end{remark}

\begin{remark}[Matroid classes at depth two]
\label{rem:depth-two-matroid-classes}
On a canonical depth-two instance, after subdividing any direct
leaf-to-\(q\) arcs, the gammoid in
Proposition~\ref{prop:ppi-matroid} is the transversal matroid
represented by the bipartite graph between the leaves and the intermediate
vertices.  Thus the residual-threshold formulation minimizes rank over
\(K\)-element subsets of its dual.
The bicircular specialization requires an additional restriction:
every leaf must have exactly two distinct intermediate neighbors.
Identifying each leaf with an edge between these neighbors yields a
graph, possibly with parallel edges.
A set of these edges can be matched to distinct incident vertices
exactly when each connected component contains at most one cycle.
Hence the corresponding transversal matroid is the bicircular
matroid of this graph.  Accordingly, co-bicircular
matroids describe this two-neighbor subclass, not arbitrary depth-two
instances.
\end{remark}

\begin{proposition}[Depth-two unit-cost PPI]
\label{prop:depth-two-unit}
Consider unit-cost PPI on a canonical DAG of depth at
most two, with every vertex other than \(q\) admissible,
leaf set \(V_0\), and \(k:=N^*<|V_0|\).

First subdivide every direct leaf-to-\(q\) arc by a
distinct new unit-cost intermediate vertex.
This preserves the optimum: an original cache remains
feasible, and a selected subdivision vertex can be
replaced by its associated leaf at no greater cost.
In the resulting graph, write \(V_1\) for the full
intermediate set.
All notation below refers to this normalized instance.

For \(W\subseteq V_0\), put
\[
  e(W):=\bigl|\{x\in V_1:N^-(x)\subseteq W\}\bigr|,
  \qquad f(W):=|W|-e(W).
\]
Then
\[
  \mathrm{OPT}=|V_1|-k+\min\{f(W):W\subseteq V_0,\ |W|\ge k\}.
\]
Moreover \(f\) is submodular.  Its inclusion-maximal minimizer
\(W_{\max}\) is computable by one maximum-flow computation, and some
optimal \(W\) contains \(W_{\max}\).
\end{proposition}

\begin{proof}
A set \(Y\subseteq V_0\) is separated from \(q\) exactly when each
\(b\in Y\) is cached or all children of \(b\) are cached.  Hence
\(g(Y)=\min_{Z\subseteq Y}\bigl(|Y\setminus Z|+|N^+(Z)|\bigr)\).
By Proposition~\ref{prop:xp-threshold},
\(\mathrm{OPT}=\min_{|X|\le k}g(V_0\setminus X)\).
Writing \(W:=V_0\setminus Z\supseteq X\), the pair \((X,W)\)
contributes \(|W|-|X|+|V_1|-e(W)\), since
\(|N^+(V_0\setminus W)|=|V_1|-e(W)\).
For fixed \(W\) the best choice is \(|X|=\min\{k,|W|\}\).
If \(|W|<k\), the value is \(|V_1|-e(W)\); enlarging \(W\) to any
\(k\)-set does not increase \(|V_1|-e(W)\), since \(e\) is monotone.
Hence we may assume \(|W|\ge k\) and \(|X|=k\), giving
\[
  \mathrm{OPT}=\min_{|W|\ge k}\bigl(|W|-k+|V_1|-e(W)\bigr).
\]
For each intermediate vertex \(x\), the indicator
\(\mathbf 1\{N^-(x)\subseteq W\}\) is supermodular.
Hence \(e\) is supermodular and \(f(W)=|W|-e(W)\)
is submodular.

To compute an inclusion-maximal minimizer of \(f\), form a
maximum-closure network.
Add a source \(s_0\), a sink \(t_0\), an arc
\(s_0\to x\) of capacity \(1\) for every \(x\in V_1\),
an arc \(b\to t_0\) of capacity \(1\) for every \(b\in V_0\),
and an arc \(x\to b\) of capacity
\[
  C_\infty:=|V_0|+|V_1|+1
\]
whenever \(b\in N^-(x)\).
For a fixed source-side leaf set \(W\), an optimal cut places
on the source side exactly the intermediate vertices whose
parents all lie in \(W\).
Its capacity is
\[
  |V_1|-e(W)+|W|=|V_1|+f(W).
\]
Thus minimum cuts minimize \(f\).

After one maximum-flow computation, the largest source side
among minimum cuts consists of the vertices that cannot reach
\(t_0\) in the residual network.
Its leaf set is the inclusion-maximal minimizer \(W_{\max}\).

If \(W^*\) minimizes \(f\) subject to \(|W^*|\ge k\), then
submodularity and the global minimality of \(W_{\max}\) give
\[
\begin{aligned}
  f(W^*\cup W_{\max})
  &\le
  f(W^*)+f(W_{\max})-f(W^*\cap W_{\max})\\
  &\le f(W^*).
\end{aligned}
\]
Since \(|W^*\cup W_{\max}|\ge k\), some constrained optimum
contains \(W_{\max}\).
\end{proof}

Figure~\ref{fig:indegree-two} shows the forest structure exploited
below.

\begin{corollary}[Intermediate in-degree two is linear-time for every \(N^*\)]
\label{cor:indegree-two-poly}
If every intermediate vertex has at most two parents
and every vertex other than \(q\) is admissible, unit-cost
PPI on depth-two canonical DAGs is solvable for every
\(N^*\) given in the input.
An optimal cache can be computed using
\(O(|V(G)|+|E(G)|)\) graph operations on the
normalized instance.  In particular, the graph
class of Theorem~\ref{thm:ppi-wh} becomes linear-time
solvable when its leaf costs are replaced by unit costs.
\end{corollary}

\begin{proof}
Put \(k:=N^*\).
If \(k\ge|V_0|\), the empty cache is optimal.
Otherwise subdivide every direct leaf-to-\(q\) arc as in
Proposition~\ref{prop:depth-two-unit}; the new intermediate
vertices have in-degree one.

Construct an undirected multigraph \(M\) on \(V_0\).
An intermediate vertex with two parents contributes
an edge between them, and one with a single parent
contributes a loop.
With \(f\) from Proposition~\ref{prop:depth-two-unit},
\[
  f(W)=|W|-e_M(W),
\]
where \(e_M(W)\) counts induced edges, including loops
and multiplicities.

We first identify the inclusion-maximal global
minimizer of \(f\).
Let \(C\) be a connected component of \(M\) that is
not a tree.
Then \(e_M(C)\ge|C|\), so \(f(C)\le0\).
For a nonempty proper subset \(W\subset C\), put
\(U:=C\setminus W\).
If \(M[U]\) has \(t\) connected components, it has at
least \(|U|-t\) internal edges.
Each of these components has an edge to \(W\), because
\(M[C]\) is connected.
Consequently,
\[
  e_M(C)-e_M(W)\ge|U|,
  \qquad
  f(W)-f(C)\ge0.
\]
The same inequality holds for \(W=\varnothing\), since
\(f(C)\le0=f(\varnothing)\).
Hence \(C\) itself is the inclusion-maximal minimizer on
that component.
If \(C\) is a tree, every nonempty \(W\subseteq C\)
induces a forest with at least one component, so
\(f(W)\ge1\); the empty set is its unique minimizer.
Since \(f\) is additive over the components of \(M\),
its inclusion-maximal global minimizer \(W_{\max}\) is
precisely the union of all non-tree components.

By Proposition~\ref{prop:depth-two-unit}, some
cardinality-constrained optimum contains \(W_{\max}\).
If \(|W_{\max}|\ge k\), no further vertices are needed.
Otherwise put \(k':=k-|W_{\max}|\).
The remaining components are trees.
For \(A\subseteq V_0\setminus W_{\max}\), the increment
\[
  f(W_{\max}\cup A)-f(W_{\max})
\]
equals the number of components of the induced forest
\(M[A]\).
Let \(j^*\) be the smallest number of largest remaining
trees whose total size is at least \(k'\).
Any induced forest with \(j\) components meets at most
\(j\) original trees and therefore has at most as many
vertices as the \(j\) largest trees.
Thus the minimum increment is \(j^*\), attained by
their union \(A^*\).

Let \(W\) be the resulting optimal set with \(|W|\ge k\).
Choose any \(X\subseteq W\) of size \(k\), and cache
\[
  S:=
  \{x\in V_1:N^-(x)\nsubseteq W\}
  \cup(W\setminus X).
\]
Every leaf outside \(W\) has all its children cached,
and every leaf in \(W\setminus X\) is cached.
Only leaves in \(X\) can remain exposed, and
\[
  c(S)=|V_1|-e(W)+|W|-k
      =|V_1|-k+f(W).
\]
Proposition~\ref{prop:depth-two-unit} proves optimality.
Replace selected subdivision vertices by their
associated leaves to obtain a cache in the original
graph.

The multigraph and its connected components are
computed in linear time.
Bucket the tree components by their sizes in
\(\{1,\ldots,|V_0|\}\), and scan the buckets in
decreasing order until the quota is met.
Selection, reconstruction, and reversal of the
subdivisions are also linear-time operations.
\end{proof}

\begin{figure}[tbp]
\centering
\begin{tikzpicture}[ppi,
  lv/.style={leaf,minimum size=15pt},
  ie/.style={thick,black!65},
  il/.style={font=\scriptsize,fill=white,inner sep=0.5pt}]
\node[lv] (b1) at (0,0)     {$b_1$};
\node[lv] (b2) at (1.2,0.9) {$b_2$};
\node[lv] (b3) at (1.2,-0.9){$b_3$};
\node[lv] (b7) at (2.5,-0.9){$b_7$};
\node[lv,fill=green!25] (b4) at (4.2,0.6) {$b_4$};
\node[lv,fill=green!25] (b5) at (5.4,0.6) {$b_5$};
\node[lv,fill=green!25] (b6) at (6.6,0.6) {$b_6$};
\node[lv] (b8) at (4.8,-0.9) {$b_8$};
\node[lv] (b9) at (6.0,-0.9) {$b_9$};
\draw[ie] (b1)--node[il]{$x_1$}(b2);
\draw[ie] (b2)--node[il]{$x_2$}(b3);
\draw[ie] (b1)--node[il]{$x_3$}(b3);
\draw[ie] (b3)--node[il]{$x_6$}(b7);
\draw[ie] (b4)--node[il]{$x_4$}(b5);
\draw[ie] (b5)--node[il]{$x_5$}(b6);
\draw[ie] (b8)--node[il]{$x_7$}(b9);
\begin{scope}[on background layer]
\node[region,fill=black!6,fit=(b1)(b2)(b3)(b7),inner sep=5pt,
  label={[note]above:$W_{\max}$}] {};
\node[region,fill=green!8,fit=(b4)(b6),inner sep=4pt,
  label={[note]above:tree $T_1$ (size 3)}] {};
\node[region,fit=(b8)(b9),inner sep=4pt,
  label={[note]below:tree $T_2$ (size 2)}] {};
\end{scope}
\end{tikzpicture}
\caption{The leaf graph behind Corollary~\ref{cor:indegree-two-poly}:
leaves are drawn as vertices, and an intermediate vertex with parents
\(b,b'\) as an edge \(bb'\); every \(x_i\) also has an arc to \(q\).
The unique non-tree component is
\(W_{\max}=\{b_1,b_2,b_3,b_7\}\), with
\(f(W_{\max})=4-4=0\); the remaining components are trees.  For
\(N^*=7\), the residual quota is \(k'=3\) and the largest tree \(T_1\)
meets it, so \(W=W_{\max}\cup T_1\).  Caching only \(x_7\) leaves the
seven leaves of \(W\) exposed, at cost \(|V_1|-N^*+f(W)=7-7+1=1\).}
\label{fig:indegree-two}
\end{figure}

\begin{remark}[The unresolved depth-two unit-cost cases]
\label{rem:nstar-open}
Theorem~\ref{thm:unit-wh} establishes residual-threshold
hardness at depth three.
At depth two, Proposition~\ref{prop:ppi-matroid} expresses
the problem as minimum-rank \(K\)-subsets of duals of
transversal matroids.
The subclass with exactly two intermediate neighbors
per leaf corresponds to co-bicircular matroids.
The results here do not determine fixed-parameter
tractability in the residual threshold for either
this subclass or general depth-two unit-cost instances.
\end{remark}

\subsection{Approximation}
\label{subsec:approximation-ppi}

We first obtain approximation guarantees from individual
cuts and the submodular separation-cost function.
We then give approximation-preserving embeddings of
smallest \(T\)-edge subgraph and minimum union.
The lower bounds are conditional transfers through these
embeddings.
All vertices other than \(q\) are admissible in this subsection
unless stated otherwise.
Figure~\ref{fig:reduction-map} gives an overview.

\begin{figure}[tbp]
\centering
\begin{tikzpicture}[ppi,
  src/.style={draw,thick,rounded corners=3pt,fill=black!5,text width=2.6cm,
    align=center,minimum height=0.9cm,font=\footnotesize},
  tgt/.style={draw,thick,rounded corners=3pt,fill=blue!6,text width=3.4cm,
    align=center,minimum height=0.9cm,font=\footnotesize},
  rd/.style={->,thick,black!65},
  lab/.style={font=\scriptsize,fill=white,inner sep=1pt,align=center}]
\node[src] (clq)  at (0,2.7)  {\textsc{Clique}\\[-1pt]{\scriptsize ETH}};
\node[src] (dks)  at (0,0.9)  {densest \(\kappa\)-subgraph\\[-1pt]
  {\scriptsize ETH, Gap-ETH}};
\node[src] (smes) at (0,-0.6) {smallest \(T\)-edge subgraph\\[-1pt]
  {\scriptsize SPCH, ETH}};
\node[src] (mu)   at (0,-2.6) {min \(t\)-union\\[-1pt]{\scriptsize HDvR}};
\node[tgt] (t2) at (7.6,3.9)  {weighted, depth 2:\\ vertices as leaves};
\node[tgt] (t3) at (7.6,2.6)  {unit, depth 3:\\ copy gadget};
\node[tgt] (t1) at (7.6,1.2)  {unit, depth 2:\\ edges as leaves};
\node[tgt] (t4) at (7.6,-0.6) {unit, depth 2:\\ bundled edge leaves};
\node[tgt] (t5) at (7.6,-2.6) {weighted, depth 2:\\ sets as leaves};
\draw[rd] (clq)--node[lab,sloped,pos=0.55]{W[1] in \(N^*\)
  (Thm.~\ref{thm:ppi-wh})}(t2);
\draw[rd] (clq)--node[lab,sloped,pos=0.55]{W[1] in \(N^*\)
  (Thm.~\ref{thm:unit-wh})}(t3);
\draw[rd] (clq)--node[lab,sloped,pos=0.45]{NP-c.; W[1] in \(\theta,t\)
\\  (Prop.~\ref{prop:module-selection-np-hard})}(t1);
\draw[rd] (dks)--node[lab,sloped,pos=0.5]{\(\alpha\)-approx.\
  \(\Rightarrow\) \(8\alpha^2\)-approx.\
\\  (Lem.~\ref{lem:ppi-dks-transfer})}(t1);
\draw[rd] (smes)--node[lab,sloped,pos=0.4]{factor 2
  (Prop.~\ref{prop:leafonly-tight})}(t1);
\draw[rd] (smes)--node[lab,pos=0.5]{exact
  (Prop.~\ref{prop:skes-containment})}(t4);
\draw[rd] (mu)--node[lab,pos=0.5]{exact
  (Prop.~\ref{prop:min-union-embedding})}(t5);
\end{tikzpicture}
\caption{Reductions used in Section~\ref{sec:complexity-classification}.
An arrow from a problem to a class of PPI instances means that hardness
of the problem, under the hypothesis named in its box, transfers to the
class with the stated loss.  The resulting lower bounds appear in
Theorem~\ref{thm:ppi-dks}, Corollary~\ref{cor:gapeth}, and
Remark~\ref{rem:approx-gap-open}.}
\label{fig:reduction-map}
\end{figure}

\begin{remark}[Individual coverage and greedy methods]
\label{rem:greedy-scope}
Under individual coverage, admissible vertices define a
partial set-cover instance on the initially exposed
terminals.
The truncated-density greedy rule gives the classical
\(H(t)\) guarantee
\citep{johnson1974approximation,lovasz1975ratio,
chvatal1979greedy}.
For the single-sink source-terminal regime considered
here, however, the blocking sets are postdominator
subtrees, and
Theorem~\ref{thm:dominator-dp} gives an exact algorithm.
Set-cover hardness therefore does not transfer merely
from the partial-cover formulation.
\end{remark}

\begin{proposition}[Single-leaf-cut approximation on general DAGs]
\label{prop:t-approx}
Let \((G,V_0,q,c,N^*)\) be a PPI instance, let
\(U:=\mathrm{Exp}(\varnothing)\), and suppose the coverage demand
\(t:=(|U|-N^*)^+\) is at least one and the instance is feasible.
For \(\ell\in U\), let \(g(\{\ell\})\) be the minimum cost of a
vertex cut separating \(\ell\) from \(q\), where the cut may contain
\(\ell\); each value is computable in polynomial time by maximum
flow.
Separating the \(t\) leaves of smallest single-leaf cut value, each
by its own minimum cut, is feasible with total cost at most
\(t\cdot\mathrm{OPT}\).
\end{proposition}

\begin{proof}
Feasibility is immediate: the union of the individual cuts separates
\(t\) initially exposed leaves from \(q\), so at most
\(|U|-t\le N^*\) leaves remain exposed.
For the guarantee, let \(S^*\) be optimal and let
\(Y^*\subseteq U\) be the leaves it intercepts, so \(|Y^*|\ge t\).
Every \(t\)-subset \(Y'\subseteq Y^*\) satisfies
\(\max_{y\in Y'}g(\{y\})\le g(Y')\le g(Y^*)\le c(S^*)=\mathrm{OPT}\)
by monotonicity of the minimum-cut function, since \(S^*\) itself
separates \(Y^*\) from \(q\).
The \(t\)-th smallest single-leaf cut value over \(U\) therefore does
not exceed \(\mathrm{OPT}\), and the total cost is at most
\(t\cdot\mathrm{OPT}\); if \(\mathrm{OPT}=0\), the cut values are
zero and the bound holds trivially.
\end{proof}

\begin{proposition}[Submodular separation-cost approximation]
\label{prop:submodular-approx}
Let
\[
  U:=\mathrm{Exp}(\varnothing),
  \qquad n:=|U|,
  \qquad t:=(n-N^*)^+\ge1,
\]
and suppose every vertex other than \(q\) is admissible.
For \(Y\subseteq U\), let \(g(Y)\) be the minimum cost of a
vertex set separating \(Y\) from \(q\), allowing deletion
of vertices of \(Y\).
Then \(g\) is normalized, monotone, and submodular, and one
maximum-flow computation evaluates \(g(Y)\).

There is a randomized polynomial-time algorithm that always
returns a feasible cache and, with probability at least
\(1-(n+1)^{-2}\), has cost at most
\[
  \rho(n)\bigl(\lfloor\log_2 t\rfloor+1\bigr)\mathrm{OPT},
  \qquad
  \rho(n)=
  O\!\left(\sqrt{\frac{n}{\log(2+n)}}\right).
\]
Taking the cheaper of this cache and the cache from
Proposition~\ref{prop:t-approx} gives a randomized
\(\min\{t,\widetilde O(\sqrt n)\}\)-approximation.
\end{proposition}

\begin{proof}
A feasible cache separates at least \(t\) initially exposed
terminals, and a separator of any \(t\) such terminals is
feasible.
Hence
\[
  \mathrm{OPT}=\min_{\substack{Y\subseteq U\\|Y|\ge t}}g(Y).
\]

For submodularity, form a fixed node-split network
without a super-source.
Replace each vertex \(v\) by
\(v^{\mathrm{in}},v^{\mathrm{out}}\).
The splitting arc has capacity \(c(v)\) for \(v\neq q\);
the splitting arc of \(q\), and every arc
\(u^{\mathrm{out}}\to v^{\mathrm{in}}\) corresponding
to an original arc \(u\to v\), have capacity
\[
  C_\infty:=1+\sum_{v\neq q}c(v).
\]

Let \(h(R)\) be the capacity of the arcs leaving a set
\(R\) of split vertices, and define
\[
  \mathcal R(Y):=
  \left\{
    R:
    q^{\mathrm{out}}\notin R,\ 
    y^{\mathrm{in}}\in R\text{ for every }y\in Y
  \right\}.
\]
Deleting the terminals in \(Y\) gives a separator
of cost less than \(C_\infty\), so a minimizing cut
crosses only deletion arcs.
The node-splitting correspondence therefore gives
\[
  g(Y)=\min_{R\in\mathcal R(Y)}h(R).
\]
If \(R_A\) and \(R_B\) attain the minima for \(A\) and \(B\),
their union is feasible for \(A\cup B\), and their
intersection is feasible for \(A\cap B\).
Submodularity of directed cut capacity therefore gives
\[
\begin{aligned}
  g(A)+g(B)
  &=h(R_A)+h(R_B)\\
  &\ge h(R_A\cup R_B)+h(R_A\cap R_B)\\
  &\ge g(A\cup B)+g(A\cap B).
\end{aligned}
\]
Normalization and monotonicity are immediate.

We use the bicriteria decision procedure of
\citet[Section~4.1]{svitkina2011submodular}.
For a normalized monotone submodular function \(f\),
an integer demand \(\tau\ge1\), and a bound
\[
  B>\mathrm{OPT}_f,
  \qquad
  \mathrm{OPT}_f:=\min_{|Y|\ge\tau}f(Y),
\]
it returns, with constant success probability, a set
\(Z\) satisfying
\[
  |Z|\ge\lceil\tau/2\rceil,
  \qquad
  f(Z)\le\rho(n)B,
\]
where
\[
  \rho(n)=
  O\!\left(\sqrt{\frac{n}{\log(2+n)}}\right).
\]
Constant-size ground sets are handled directly.

For the separation-cost functions used here, the
unknown bound can be eliminated in polynomial time.
Scale all rational vertex costs to integers using a
common denominator of polynomial bit length.
The values of \(g\) and its residual functions are
then nonnegative integers bounded by the total scaled
vertex cost \(C\).

For a residual function \(f\), compute
\[
  Z_0:=\{u:f(\{u\})=0\}.
\]
Normalization and submodularity imply \(f(Z_0)=0\);
monotonicity implies that every zero-cost set is
contained in \(Z_0\).
If
\[
  |Z_0|\ge\lceil\tau/2\rceil,
\]
return \(Z_0\).
Otherwise \(\mathrm{OPT}_f>0\).

Try power-of-two bounds up to the first power strictly
greater than \(C\).
For each bound, repeat the bicriteria procedure and
retain only outputs satisfying
\[
  |Z|\ge\lceil\tau/2\rceil.
\]
Return an output of minimum evaluated cost.
One tested bound satisfies
\[
  \mathrm{OPT}_f<B\le2\mathrm{OPT}_f.
\]
Thus, after absorbing the factor two into \(\rho(n)\),
the returned set satisfies
\[
  f(Z)\le\rho(n)\mathrm{OPT}_f
\]
with any prescribed inverse-polynomial failure
probability.
The number of bounds and repetitions is polynomial
in the input length.

Start with \(Y_0:=\varnothing\).
While
\[
  t_i:=t-|Y_i|>0,
\]
apply this optimization step to
\[
  f_i(Z):=g(Y_i\cup Z)-g(Y_i),
  \qquad Z\subseteq U\setminus Y_i,
\]
with demand \(t_i\) and failure probability at most
\((n+1)^{-3}\).
The function \(f_i\) is normalized, monotone, and submodular.

Let \(Y^*\) attain \(\mathrm{OPT}\).
Then
\[
  |Y^*\setminus Y_i|\ge t_i
\]
and
\[
\begin{aligned}
  f_i(Y^*\setminus Y_i)
  &=g(Y_i\cup Y^*)-g(Y_i)\\
  &\le g(Y^*)=\mathrm{OPT}.
\end{aligned}
\]
Thus a successful round returns \(Z_i\) with
\[
  |Z_i|\ge t_i/2,
  \qquad
  f_i(Z_i)\le\rho(n)\mathrm{OPT}.
\]
Set \(Y_{i+1}:=Y_i\cup Z_i\).
The residual demand satisfies
\[
  t_{i+1}\le\lfloor t_i/2\rfloor,
\]
so at most \(\lfloor\log_2 t\rfloor+1\) rounds are needed.
On successful completion, telescoping gives
\[
  g(Y_{\mathrm{final}})
  =
  \sum_i f_i(Z_i)
  \le
  \rho(n)\bigl(\lfloor\log_2 t\rfloor+1\bigr)\mathrm{OPT}.
\]
A minimum separator for \(Y_{\mathrm{final}}\) is the output
cache.

Use fresh randomness in each call.
Conditional on the preceding rounds, each round fails with
probability at most \((n+1)^{-3}\).
There are at most \(n\) rounds, so a union bound gives total
failure probability at most \((n+1)^{-2}\).
If a round produces no set of size at least
\(\lceil t_i/2\rceil\), return the
feasible cache of Proposition~\ref{prop:t-approx}.
Finally, compare costs with that same cache.
The output is always feasible, and the stated approximation
guarantee holds with the claimed probability.
\end{proof}

\begin{proposition}[Smallest \(T\)-edge subgraph
embeds into PPI]
\label{prop:skes-containment}
Given a simple graph \(H=(V_H,E_H)\) and
\(1\le T\le|E_H|\), one can construct in polynomial
time a unit-cost canonical depth-two PPI instance
with leaf out-degree two whose optimum equals the
minimum number of vertices inducing at least \(T\)
edges in \(H\).
Polynomial-time solution maps preserve approximation
factors in both directions.
\end{proposition}

\begin{proof}
Put \(m:=|E_H|\),
let \(V_{+}:=\{v\in V_H:\deg_H(v)>0\}\), and set the bundling
factor \(\beta:=|V_H|+1\).  Construct the canonical DAG as in
Proposition~\ref{prop:module-selection-np-hard}, but replace the
leaf for each edge \(e\in E_H\) by \(\beta\) distinct leaves
\[
  p_e^1,\ldots,p_e^\beta
\]
with identical out-neighborhoods.
The intermediate vertices are
\(X=\{x_v:v\in V_{+}\}\), with designated parents
\(\operatorname{par}(x_v)=(p_e^{i}:e\ni v,\,1\le i\le\beta)\) and
\(\operatorname{par}(q)=(x_v:v\in V_{+})\); unit costs; threshold
\(N^{*}:=\beta(m-T)\ge0\).  The instance has depth two, every leaf has
out-degree two, and its size is polynomial in \(|V_H|+m\).

For \(C\subseteq X\) write \(e(C):=|\{\{u,v\}\in E_H:\{x_u,x_v\}\subseteq
C\}|\) for the edges induced by the corresponding vertex set.  A leaf
\(p_e^{i}\), \(e=\{u,v\}\), is non-exposed under \(S=C\cup R\) if and only
if \(p_e^{i}\in R\) or \(\{x_u,x_v\}\subseteq C\).  Deleting from \(R\)
any leaf already blocked by \(C\) preserves feasibility and does not
increase the cost, so we may assume that \(R\) contains only leaves
exposed under \(C\); then a cache
\(S=C\cup R\) is feasible if and only if
\[
  |R|\;\ge\;\beta\bigl(m-e(C)\bigr)-N^{*}
  \;=\;\beta\bigl(T-e(C)\bigr),
  \qquad\text{with cost}\quad
  \ell_c(S)=|C|+|R|.
\]
If \(e(C)\ge T\), the cheapest completion is \(R=\varnothing\), of cost
\(|C|\); minimizing over such \(C\) gives exactly the smallest vertex set
inducing at least \(T\) edges, i.e.\ the optimum of \((H,T)\).  If
\(e(C)<T\), feasibility forces \(|R|\ge\beta=|V_H|+1\), so
\(\ell_c(S)>|V_H|\ge|V_{+}|\), whereas caching \(C=V_{+}\) with
\(R=\varnothing\) is feasible (it induces all \(m\ge T\) edges) at cost
\(|V_{+}|\le|V_H|\); such an \(S\) is therefore never optimal.  The two
optima coincide, and from any feasible \(S=C\cup R\) of cost \(a\) one
recovers in polynomial time a vertex set inducing at least \(T\) edges of
size at most \(a\) (namely \(C\) if \(e(C)\ge T\), else \(V_{+}\)), so
the reduction preserves approximation factors exactly.
\end{proof}

\begin{proposition}[Approximation bounds on the
unbundled edge-incidence class]
\label{prop:leafonly-tight}
Consider the unit-cost instance of
Proposition~\ref{prop:module-selection-np-hard} built from a graph
\(H=(V_H,E_H)\) with \(m:=|E_H|\), but with an arbitrary coverage
demand \(T=m-N^*\in\{1,\ldots,m\}\).  Let
\(\mathrm{OPT}_{\mathrm{SmES}}(H,T)\) be the least number of vertices
of \(H\) inducing at least \(T\) edges.
\begin{enumerate}[label=\textup{(\roman*)}]
\item Caching any \(T\) leaves is feasible, optimal if \(T\le3\),
and an \(O(\sqrt T)\)-approximation in general, because
\(\mathrm{OPT}\ge(1+\sqrt{1+8T})/2\ge\sqrt{2T}\) for \(T\ge4\).
\item There are polynomial-time solution maps in both directions, and
\(\mathrm{OPT}\le\mathrm{OPT}_{\mathrm{SmES}}(H,T)\le2\,
\mathrm{OPT}\).
\end{enumerate}
\end{proposition}

\begin{proof}
(i) A cache with \(j\) intermediate vertices and \(a-j\) leaves
blocks at most
\[
  \binom{j}{2}+a-j
  \le\max\left\{a,\binom{a}{2}\right\},
\]
because the expression is convex in \(j\in[0,a]\).
For \(T\le3\), every \(a<T\) satisfies
\(\max\{a,\binom a2\}=a<T\), so the leaf-only cost \(T\) is optimal.
For \(T\ge4\), a feasible cache of cost \(a\) satisfies \(a\ge T\) or
\(\binom a2\ge T\).
Since \(T\ge(1+\sqrt{1+8T})/2\) for \(T\ge3\), both cases give
\(a\ge(1+\sqrt{1+8T})/2\), which is the stated lower bound on
\(\mathrm{OPT}\); dividing the feasible leaf-only cost \(T\)
by \(\mathrm{OPT}\) gives its approximation guarantee.

(ii) If a vertex set \(C'\subseteq V_H\) induces at least \(T\)
edges, the cache
\[
  \{x_v:v\in C'\cap V_{+}\},
\]
with \(V_{+}\) the non-isolated vertices of \(H\), is feasible and
has cost at most \(|C'|\).
Thus \(\mathrm{OPT}\le\mathrm{OPT}_{\mathrm{SmES}}(H,T)\).

Conversely, consider any feasible cache \(S=C\cup R\).
Discard from \(R\) the leaves already blocked by \(C\), without
increasing cost or changing feasibility.
The remaining cached leaves correspond to distinct edges not
induced by \(C\), and \(e(C)+|R|\ge T\).
Take the graph vertices corresponding to \(C\), and add both
endpoints of every edge represented by a member of \(R\).
The resulting vertex set induces at least \(T\) edges and has size
at most \(|C|+2|R|\le2c(S)\).
This is a polynomial-time solution map for every feasible cache.
Applying it to an optimum gives
\(\mathrm{OPT}_{\mathrm{SmES}}(H,T)\le2\,\mathrm{OPT}\).
\end{proof}

We use the Strongish Planted Clique Hypothesis (SPCH) of
\citet[Hypothesis~1]{manurangsi2021strongish}.
It asserts that, for some constant \(\xi\in(0,1/2)\), no
\(n^{o(\log n)}\)-time algorithm distinguishes
\(G(n,1/2)\) from the same distribution with a uniformly
planted clique of size \(\lceil n^\xi\rceil\), with error
at most \(1/3\) under either distribution.

\begin{theorem}[Conditional lower bounds]
\label{thm:ppi-dks}
Unit-cost PPI on canonical depth-two DAGs with leaf out-degree two
satisfies the following.
\begin{enumerate}[label=\textup{(\alph*)}]
\item Under SPCH, no algorithm running in time
\(f(\mathrm{OPT})|I|^{O(1)}\), for a computable \(f\), and in particular
no polynomial-time algorithm, achieves ratio \(o(\sqrt{\mathrm{OPT}})\).
\item Under SPCH, on the unbundled class no algorithm running in time
\(f(T)|I|^{O(1)}\) achieves ratio \(o(\sqrt T)\), where \(T=|V_0|-N^*\)
is the coverage demand.
\item Under ETH, there is a constant \(c>0\) such that, already on the
unbundled class, no polynomial-time algorithm achieves ratio
\(|I|^{1/(\log\log|I|)^{c}}\).
\end{enumerate}
\end{theorem}

\begin{proof}
Under SPCH, smallest \(T\)-edge subgraph admits no approximation
with ratio \(o(\sqrt T)\) in time \(f(T)|I|^{O(1)}\)
\citep[Corollary~11]{manurangsi2021strongish}.

\smallskip\noindent
\emph{(a)}
Use the bundled construction of
Proposition~\ref{prop:skes-containment}.
Its optimum \(a\) equals the source optimum and satisfies
\(a\le2T\), since the endpoints of any \(T\) edges form
a feasible source solution.
If a PPI algorithm has approximation ratio
\(h(a)=o(\sqrt a)\), then
\[
  \max_{1\le a\le2T}h(a)=o(\sqrt T).
\]
Its running time is bounded by
\[
  \max_{1\le a\le2T}f(a)\cdot |I|^{O(1)}.
\]
The construction has polynomial size, and its solution
map does not increase cost.
This contradicts the source hardness.

\smallskip\noindent
\emph{(b)}
The unbundled construction and solution maps of
Proposition~\ref{prop:leafonly-tight}(ii) preserve the
demand \(T\) and increase the approximation factor by
at most two.
They therefore transfer the same SPCH exclusion.

\smallskip\noindent
\emph{(c)}
By \citet{manurangsi2017dks}, ETH implies that, for some
constant \(c_1>0\), no polynomial-time algorithm
distinguishes the following cases for an \(n\)-vertex
graph \(H\) and an integer \(2\le\kappa\le n\):
\[
\begin{array}{ll}
\text{YES:} &
  H\text{ contains a }\kappa\text{-clique};\\
\text{NO:} &
  \text{every }\kappa\text{-set induces at most }
  \binom{\kappa}{2}/g\text{ edges},
\end{array}
\]
where
\[
  g:=n^{1/(\log\log n)^{c_1}}.
\]

Put \(T:=\binom{\kappa}{2}\).
If \(|E_H|<T\), the instance cannot be a YES instance,
so return NO directly.
Otherwise construct the unbundled PPI instance and
run the hypothetical approximation algorithm.
Here \(n\) continues to denote the number of vertices
in the original graph, including any isolated vertices.

On a YES instance, the PPI optimum is at most
\(\kappa\).
An \(\alpha\)-approximate cache maps to a vertex set
\(Y\) inducing at least \(T\) edges and satisfying
\[
  |Y|\le2\alpha\kappa.
\]
If \(|Y|\le\kappa\), pad it to \(\kappa\) vertices.
Otherwise, a uniformly random \(\kappa\)-subset of \(Y\)
retains each induced edge with probability
\[
  \frac{\kappa(\kappa-1)}{|Y|(|Y|-1)}
  \ge\frac{1}{8\alpha^2}.
\]
Conditional expectations find a \(\kappa\)-set
inducing at least \(T/(8\alpha^2)\) edges.
Thus the two promise cases are distinguished whenever
\[
  8\alpha^2<g.
\]

The constructed encoding has size at most \(n^d\)
for some fixed constant \(d\).
For a sufficiently large constant \(c>c_1\), an
approximation ratio
\[
  \alpha=|I|^{1/(\log\log|I|)^c}
\]
satisfies \(8\alpha^2<g\) for all sufficiently large
instances.
The finitely many smaller instances can be handled
directly.
This contradicts ETH.
\end{proof}

\begin{lemma}[Transfer to densest subgraph]
\label{lem:ppi-dks-transfer}
Let \(\alpha\ge1\), and let \(\mathcal B\) be an
algorithm that returns a feasible cache of cost at most
\(\alpha\,\mathrm{OPT}\) on every unbundled
edge-incidence instance.

Given a graph \(H=(V_H,E_H)\) and an integer
\(2\le\kappa\le|V_H|\), one can use \(\mathcal B\)
on all demands
\(T\in\{1,\ldots,|E_H|\}\) to obtain a
\(\kappa\)-vertex set inducing at least
\[
  \frac{\mu^\star}{8\alpha^2}
\]
edges, where \(\mu^\star\) is the maximum number
of edges induced by \(\kappa\) vertices of \(H\).
The construction uses at most \(|E_H|\) calls to
\(\mathcal B\) and polynomial additional time.
\end{lemma}

\begin{proof}
If \(E_H=\varnothing\), return any \(\kappa\)-vertex
set; the guarantee is immediate.
Otherwise \(\mu^\star\ge1\).
For each \(T\in\{1,\ldots,|E_H|\}\), run
\(\mathcal B\) on the unbundled edge-incidence
instance with coverage demand \(T\), and map its
output as in
Proposition~\ref{prop:leafonly-tight}(ii) to a
vertex set \(Y_T\) inducing at least \(T\) edges with
\(|Y_T|\le2\alpha\,\mathrm{OPT}_{\mathrm{SmES}}(H,T)\).
For \(T=\mu^\star\) we have \(\mathrm{OPT}_{\mathrm{SmES}}\le\kappa\),
so \(|Y_T|\le2\alpha\kappa\).
If \(|Y_T|\le\kappa\), pad it to \(\kappa\) vertices; otherwise a
uniformly random \(\kappa\)-subset retains each induced edge with
probability
\(\kappa(\kappa-1)/(|Y_T|(|Y_T|-1))\ge1/(8\alpha^2)\), and the method
of conditional expectations finds a \(\kappa\)-set with at least
\(\mu^\star/(8\alpha^2)\) edges.
Returning the best \(\kappa\)-set over all \(T\) gives the claim.
\end{proof}

\begin{corollary}[A Gap-ETH bound]
\label{cor:gapeth}
Under Gap-ETH, for computable \(f\) and \(h\ge1\) with
\(h(a)=a^{o(1)}\), no deterministic algorithm running in time
\(f(\mathrm{OPT})|I|^{O(1)}\) achieves ratio \(h(\mathrm{OPT})\) on the
unbundled edge-incidence class.
\end{corollary}

\begin{proof}
Suppose such an algorithm exists.
Consider an input \((H,\kappa)\) to densest
\(\kappa\)-subgraph with perfect completeness, where
\(2\le\kappa\le|V_H|\), and put
\[
  T:=\binom{\kappa}{2}.
\]
If \(|E_H|<T\), return any \(\kappa\)-vertex set; such an
input cannot satisfy the completeness promise.

Otherwise construct the unbundled edge-incidence instance
with demand \(T\) and run the hypothetical algorithm with
time cap
\[
  \Bigl(\max_{1\le a\le\kappa}f(a)\Bigr)\,|I|^d,
\]
where \(d\) is a fixed exponent bounding its polynomial
running-time factor.
On time-out or invalid output, return any
\(\kappa\)-vertex set.

On a completeness instance, \(H\) contains a
\(\kappa\)-clique, so the PPI optimum is at most
\(\kappa\).
The run therefore terminates within the cap and returns a
feasible cache of cost at most
\[
  \bar h(\kappa)\,\kappa,
  \qquad
  \bar h(\kappa):=\max_{1\le a\le\kappa}h(a)=\kappa^{o(1)}.
\]
Map the cache to a vertex set as in
Proposition~\ref{prop:leafonly-tight}(ii), then pad or
shrink it to \(\kappa\) vertices using conditional
expectations, as in Lemma~\ref{lem:ppi-dks-transfer}.
The resulting set induces at least
\[
  \frac{T}{8\bar h(\kappa)^2}
\]
edges.

This gives a deterministic algorithm with running time
\(F(\kappa)\,|V_H|^{O(1)}\) and approximation ratio
\(\kappa^{o(1)}\) under perfect completeness, contradicting
the Gap-ETH hardness of densest \(\kappa\)-subgraph, which
excludes ratio \(\kappa^{o(1)}\) in time
\(f(\kappa)\,|V_H|^{O(1)}\) under that promise
\citep{chalermsook2017gapeth}.
\end{proof}

The SPCH bound of Theorem~\ref{thm:ppi-dks}(a) is stronger but rests on
a different hypothesis.  Through these reductions, neither ETH nor
Gap-ETH yields a square-root exclusion in terms of the demand.

\begin{proposition}[Min \(t\)-union as a depth-two subclass]
\label{prop:min-union-embedding}
Let \(A_1,\ldots,A_n\) be nonempty subsets of a finite ground set
\(\mathcal M\), and let \(1\le t\le n\).  Assume that every element
of \(\mathcal M\)
occurs in at least one set.
There is a polynomial-time construction of a weighted canonical
depth-two PPI instance whose optimum equals the minimum
union size of \(t\) input sets.
The construction admits solution maps that do not increase
approximation factors.
\end{proposition}

\begin{proof}
Create a source leaf \(\ell_i\) for every set \(A_i\), an intermediate
vertex \(x_a\) for every \(a\in\mathcal M\), and a target \(q\).  Add
\(\ell_i\to x_a\) if \(a\in A_i\), and \(x_a\to q\) for every
\(a\in\mathcal M\).
Give each intermediate vertex cost \(1\), each leaf cost \(|\mathcal M|+1\),
and set
\(N^*:=n-t\).
The assumptions ensure that the leaves are exactly the sources and
that every vertex reaches \(q\), so the instance is canonical with
polynomially bounded positive costs.

For a cache containing only intermediate vertices, a leaf \(\ell_i\) is
intercepted exactly when all vertices \(\{x_a:a\in A_i\}\) are
cached.
Consequently, the union of any \(t\) input sets yields a feasible
cache of the same cost.
Conversely, any feasible intermediate-only cache intercepts at least \(t\)
leaves and contains the union of their corresponding sets.
Caching all intermediate vertices is feasible at cost \(|\mathcal M|\),
whereas any
cache containing a leaf costs at least \(|\mathcal M|+1\).
Hence an optimum contains no leaf, proving equality of the optimum
values.

For an arbitrary feasible cache returned by an approximation
algorithm, if it contains a leaf, replace the whole cache by all
intermediate vertices; this strictly decreases its cost.
From the resulting intermediate-only cache, select any \(t\) intercepted
leaves and return the union of their sets.
The union has size at most the cache cost, so the solution map does
not increase approximation factors.
\end{proof}

\begin{remark}[Minimum union and the remaining gap]
\label{rem:approx-gap-open}
On the weighted subclass of Proposition~\ref{prop:min-union-embedding},
where \(n\) is the number of input sets and hence of initially exposed
terminals, the algorithm of \citet{chlamtac2017minimizing} gives an
\(O(n^{1/4+\eta})\)-approximation for every fixed \(\eta>0\), and their
Hypergraph Dense-vs-Random conjecture (HDvR) implies a nearly matching
\(\widetilde\Omega(n^{1/4-\eta})\) lower bound for every fixed
\(\eta\in(0,1/4)\).  In terms of the demand alone, the \(o(\sqrt t)\)
exclusion leaves a gap up to the factor \(t\) of
Proposition~\ref{prop:t-approx}: the estimate
\(\mathrm{OPT}=\Omega(\sqrt t)\) behind
Proposition~\ref{prop:leafonly-tight} fails on general DAGs, where many
terminals may share one unit-cost bottleneck.  Likewise, oracle lower
bounds for arbitrary monotone submodular functions do not directly
establish lower bounds for separation-cost functions.
\end{remark}

\subsection{Summary and Open Problems}
\label{subsec:summary-open}

Figure~\ref{fig:landscape} and
Table~\ref{tab:approximation-landscape} collect the classification
proved in this section.

\begin{table}[ht]
\centering
\footnotesize
\caption{Approximation landscape.
Here \(n=|\mathrm{Exp}(\varnothing)|\),
\(t=(n-N^*)^+\), and \(|I|\) is the input size.
For the unbundled edge-incidence class, \(T=t\).
ETH-based exclusions concern deterministic algorithms,
as specified at the start of
Section~\ref{sec:complexity-classification}.
HDvR denotes the Hypergraph Dense-vs-Random conjecture.
Its lower bounds are stated for each fixed
\(\eta\in(0,1/4)\).}
\label{tab:approximation-landscape}
\begin{tabularx}{\textwidth}{@{}L L L@{}}
\toprule
Class & Upper bound & Conditional lower bound\\
\midrule

General DAGs &
\(\min\{t,\widetilde O(\sqrt n)\}\), randomized
(Propositions~\ref{prop:t-approx}
and~\ref{prop:submodular-approx}) &
No \(|I|^{1/(\log\log|I|)^c}\) under ETH
(Theorem~\ref{thm:ppi-dks});
no \(o(\sqrt{\mathrm{OPT}})\) or \(o(\sqrt t)\)
under SPCH
(Theorem~\ref{thm:ppi-dks});
no \(\mathrm{OPT}^{o(1)}\) approximation in
\(f(\mathrm{OPT})|I|^{O(1)}\) time under Gap-ETH
(Corollary~\ref{cor:gapeth});
\(\widetilde\Omega(n^{1/4-\eta})\) under HDvR
(Remark~\ref{rem:approx-gap-open}).\\

Minimum-union depth-two subclass &
\(O(n^{1/4+\eta})\)
\citep{chlamtac2017minimizing} &
\(\widetilde\Omega(n^{1/4-\eta})\) under HDvR
(Remark~\ref{rem:approx-gap-open}).\\

Unit-cost unbundled edge-incidence &
\(O(\sqrt T)\)
(Proposition~\ref{prop:leafonly-tight}) &
No \(o(\sqrt T)\) under SPCH
(Theorem~\ref{thm:ppi-dks}).\\

Individual coverage &
Exact polynomial time
(Theorem~\ref{thm:dominator-dp}) &
---\\

\bottomrule
\end{tabularx}
\end{table}

The results leave the following questions.

\begin{enumerate}
\item
\emph{Bounded reachable regions.}
With leaves forbidden, is PPI FPT in
\(N^*+r\), where \(r\) is the maximum number of vertices reachable from
a single leaf?

\item
\emph{Unit costs at depth two.}
Is the residual-threshold parameterization FPT on
the two-intermediate-neighbor subclass, or on all
depth-two unit-cost instances?
We conjecture FPT for the former.

\item
\emph{The unit-cost ETH exponent.}
Can the gap between the
\(|I|^{O(N^*)}\) enumeration and the
\(|I|^{o(\sqrt{N^*})}\) exclusion be closed?

\item
\emph{Demand-only approximation.}
Does a polynomial-time approximation with ratio
\(\alpha(t)=o(t)\), independent of \(n\), exist
on general DAGs?
\end{enumerate}

\section{Shared Workloads and Semantic Modules}
\label{sec:semantic-scenario}

Let
\[
  W=\{q_1,\ldots,q_L\}\subseteq\mathcal Q\setminus B
\]
share one cache and one erasure realization, with the workload
notation of Subsection~\ref{subsec:standing-notation}.
The canonical assumptions remain in force.
We derive the joint and maximal-error constraints,
then optimize joint storage under exact module routing.

\subsection{Workload Projection and Exact Reliability Laws}
\label{subsec:workload-exact-laws}

Let
\[
  V_W
  :=
  \bigcup_{\ell=1}^{L}V(G_\ell)
\]
be the workload-relevant vertex set.

\begin{definition}[Workload projection and exposed-leaf union]
\label{def:workload-projection}
For a finite common transparent cache \(S_W\subseteq\Cn(B)\), define its
projection onto the workload region by
\begin{equation}
  U(W,S_W):=S_W\cap V_W.
\label{eq:workload-projection}
\end{equation}
For each query \(q_\ell\), retain the query-local exposed-leaf set
\[
  D_{\mathrm{exp}}(q_\ell,S_W)
\]
from Definition~\ref{def:query-local-residual}.  The workload exposed-leaf
union is
\begin{equation}
  D_{\cup}(W,S_W)
  :=
  \bigcup_{\ell=1}^{L}D_{\mathrm{exp}}(q_\ell,S_W).
\label{eq:workload-exposed-union}
\end{equation}
\end{definition}

\begin{theorem}[Workload-level rigidity and exact reliability laws]
\label{thm:workload-rigidity}
Under Assumptions~\ref{assump:det-local} and
\ref{assump:hereditary-dag}, for every finite
\(S_W\subseteq\Cn(B)\) and every \(\widetilde B\subseteq B\),
\begin{align}
  q_\ell\in\Cn(\widetilde B\cup S_W)
  &\Longleftrightarrow
  q_\ell\in\Cn\bigl(\widetilde B\cup U(W,S_W)\bigr)
  \nonumber\\
  &\Longleftrightarrow
  D_{\mathrm{exp}}(q_\ell,S_W)\subseteq\widetilde B,
  \quad \ell=1,\ldots,L.
\label{eq:workload-querywise-rigidity}
\end{align}
Moreover,
\begin{equation}
  \bigcap_{\ell=1}^{L}
  \left\{
    q_\ell\in\Cn(\widetilde B\cup S_W)
  \right\}
  \Longleftrightarrow
  D_{\cup}(W,S_W)\subseteq\widetilde B.
\label{eq:workload-joint-rigidity}
\end{equation}
Consequently,
\begin{equation}
  \Pr_\varepsilon\!\left[
    q_\ell\in\Cn(\widetilde B\cup S_W)
  \right]
  =
  (1-\varepsilon)^{|D_{\mathrm{exp}}(q_\ell,S_W)|},
  \; \ell=1,\ldots,L,
\label{eq:workload-success}
\end{equation}
and
\begin{equation}
  \Pr_\varepsilon\!\left[
    \bigcap_{\ell=1}^{L}
    \left\{
      q_\ell\in\Cn(\widetilde B\cup S_W)
    \right\}
  \right]
  =
  (1-\varepsilon)^{|D_{\cup}(W,S_W)|}.
\label{eq:workload-joint-success}
\end{equation}
\end{theorem}

\begin{proof}
For each \(\ell\),
\[
  U(q_\ell,S_W)
  =
  U(W,S_W)\cap V(G_\ell).
\]
Theorem~\ref{thm:exact-residual-law} gives the
querywise identities.
All queries recover exactly when every querywise
exposed set survives, equivalently when their union
survives.
Independence then gives the probability formulas.
\end{proof}

The query events may be dependent; their shared
premises are counted once in the exposed-leaf union.

\begin{corollary}[Exact workload reliability thresholds]
\label{cor:workload-thresholds}
Let \(N^*(\varepsilon,\delta)\) be defined by \eqref{eq:model-N-star}.

A finite common transparent cache \(S_W\) satisfies the maximal-error criterion
if and only if
\begin{equation}
  |D_{\mathrm{exp}}(q_\ell,S_W)|
  \le
  N^*(\varepsilon,\delta),
  \qquad \ell=1,\ldots,L.
\label{eq:workload-reliability-constraint}
\end{equation}
It satisfies the joint criterion if and only if
\begin{equation}
  |D_{\cup}(W,S_W)|
  \le
  N^*(\varepsilon,\delta).
\label{eq:workload-joint-reliability-constraint}
\end{equation}
\end{corollary}

\begin{proof}
Apply the definition of \(N^*\) to
Theorem~\ref{thm:workload-rigidity}.
\end{proof}

\subsection{Exact Workload-Level Storage Formulations}
\label{subsec:workload-exact-optimization}

We now restrict to the answer-excluding common-cache class
\[
  S_W\cap W=\varnothing.
\]

\begin{theorem}[Exact workload storage formulations]
\label{thm:exact-workload-opt}
Under Assumptions~\ref{assump:det-local} and
\ref{assump:hereditary-dag},
\begin{align}
  \sigma_{\mathrm{sem}}^{*,\max}
  (W,\varepsilon,\delta)
  =
  \min\Bigl\{
    \ell_c(S):\,
    &S\subseteq V_W\setminus W,\notag\\
    &\;|D_{\mathrm{exp}}(q_\ell,S)|
      \le N^*(\varepsilon,\delta),
      \ \ell=1,\ldots,L
  \Bigr\},
\label{eq:exact-workload-opt-max}
\end{align}
and
\begin{align}
  \sigma_{\mathrm{sem}}^{*,\mathrm{joint}}
  (W,\varepsilon,\delta)
  =
  \min\Bigl\{
    \ell_c(S):
    S\subseteq V_W\setminus W,
    \;|D_{\cup}(W,S)|
      \le N^*(\varepsilon,\delta)
  \Bigr\}.
\label{eq:exact-workload-opt-joint}
\end{align}
Thus, under either workload reliability convention, an optimal semantic-aware
cache may be chosen entirely inside the union of the canonical query DAGs.
\end{theorem}

\begin{proof}
Workload projection preserves every recovery event
and cannot increase cost.
Corollary~\ref{cor:workload-thresholds} converts the
two reliability criteria into the stated constraints.
The projected candidate family is finite and feasible:
caching \(\bigcup_\ell D_\ell\) guarantees recovery.
Thus both minima are attained.
\end{proof}

\begin{corollary}[Exact leaf-only workload baselines]
\label{cor:workload-leaf-baselines}
Write \(N^*:=N^*(\varepsilon,\delta)\).

For the maximal-error criterion,
\begin{align}
  \sigma_{\mathrm{leaf}}^{*,\max}
  &(W,\varepsilon,\delta)
  =
  \min\Bigl\{
    \ell_c(R):
    R\subseteq D_{\cup},
    \nonumber\\[-1mm]
    &|R\cap D_\ell|
      \ge
      \bigl(\kappa_\ell-N^*(\varepsilon,\delta)\bigr)^+,
      \ \ell=1,\ldots,L
  \Bigr\}.
\label{eq:workload-leaf-opt-max}
\end{align}
For maximal error, this is binary weighted set multicover.
With
\[
  d_\ell:=(\kappa_\ell-N^*)^+,
  \qquad
  D_{\mathrm{req}}:=\sum_\ell d_\ell,
\]
marginal-gain greedy gives an
\(H(D_{\mathrm{req}})\) approximation when
\(D_{\mathrm{req}}>0\)
\citep{rajagopalan1998primal,vazirani2001approximation};
otherwise the empty cache is optimal.

For the joint criterion, arrange the costs of the leaves in \(D_{\cup}\) in
nondecreasing order:
\[
  c_{(1)}^{\cup}\le c_{(2)}^{\cup}\le\cdots\le c_{(|D_{\cup}|)}^{\cup}.
\]
Then
\begin{equation}
  \sigma_{\mathrm{leaf}}^{*,\mathrm{joint}}
  (W,\varepsilon,\delta)
  =
  \sum_{i=1}^{(|D_{\cup}|-N^*)^+}
  c_{(i)}^{\cup}.
\label{eq:workload-leaf-opt-joint}
\end{equation}
An empty sum equals zero.
If all base leaves have the same cost
\(c_B\), then
\begin{equation}
  \sigma_{\mathrm{leaf}}^{*,\mathrm{joint}}
  (W,\varepsilon,\delta)
  =
  \bigl(|D_{\cup}|-N^*(\varepsilon,\delta)\bigr)^+c_B.
\label{eq:workload-leaf-opt-joint-hom}
\end{equation}

If, in addition, \(D_1,\ldots,D_L\) are pairwise disjoint and all leaf costs
equal \(c_B\), then
\begin{equation}
  \sigma_{\mathrm{leaf}}^{*,\max}
  (W,\varepsilon,\delta)
  =
  \sum_{\ell=1}^{L}
  \bigl(\kappa_\ell-N^*(\varepsilon,\delta)\bigr)^+c_B.
\label{eq:workload-leaf-opt-max-disjoint}
\end{equation}
\end{corollary}

\begin{proof}
For a leaf-only cache \(R\subseteq B\), source leaves cannot
lie on paths from distinct source leaves.
Hence
\[
  D_{\mathrm{exp}}(q_\ell,R)=D_\ell\setminus R,
  \qquad
  D_{\cup}(W,R)=D_{\cup}\setminus R.
\]
The maximal-error formulation follows immediately from
Corollary~\ref{cor:workload-thresholds}.
For joint reliability, one must select at least
\((|D_{\cup}|-N^*)^+\) leaves, so sorting their costs gives
\eqref{eq:workload-leaf-opt-joint}.
The homogeneous formula follows.
For pairwise disjoint \(D_\ell\), the maximal-error
formulation separates over the queries, giving
\eqref{eq:workload-leaf-opt-max-disjoint}.

For the greedy guarantee, define
\[
  f(R):=\sum_{\ell=1}^{L}
  \min\{d_\ell,|R\cap D_\ell|\}.
\]
This function is nondecreasing, integer-valued, and
submodular; feasibility is equivalent to
\(f(R)=D_{\mathrm{req}}\).
At a partial selection \(R\), put
\(r:=D_{\mathrm{req}}-f(R)>0\).
For an optimal selection \(R^*\),
\[
  \sum_{b\in R^*\setminus R}
  \bigl(f(R\cup\{b\})-f(R)\bigr)
  \ge f(R\cup R^*)-f(R)=r.
\]
Thus an available leaf has cost per positive marginal gain
at most \(\mathrm{OPT}/r\), and the greedy choice has no
larger density.  If a step reduces the residual demand from
\(r\) to \(r'<r\), its cost is at most
\(\mathrm{OPT}\,(r-r')/r\le\mathrm{OPT}\sum_{j=r'+1}^{r}1/j\);
summing over the strictly decreasing integer residuals yields
the \(H(D_{\mathrm{req}})\) bound.
\end{proof}

\begin{remark}[Fixed thresholds and workload criteria]
\label{rem:max-workload-hardness}
Maximal-error workloads remain NP-complete at
\(N^*=1\).
Given a simple graph \(F=(V,E)\) without isolated
vertices, create unit-cost leaves \(b_v\) and queries
\[
  q_{\{u,v\}}\leftarrow b_u,b_v.
\]
Use the finite Horn universe consisting only of
these objects.
After projection and answer exclusion, feasible
caches satisfy
\[
  |S\cap\{b_u,b_v\}|\ge1
  \qquad(\{u,v\}\in E),
\]
so they are exactly vertex covers of \(F\).
Feasibility is polynomial-time checkable.

For joint reliability, add a protected super-sink
with an arc from every query.
Enumerating the at most \(N^*\) allowed exposed
leaves and separating all others gives a
polynomial-time algorithm for fixed \(N^*\).
The query vertices and super-sink are forbidden
in the node-splitting network.

The vertex-cover instances have exact module routing
with no modules.
Hence the polynomial module result below is specific
to joint reliability.
\end{remark}

\subsection{Shared Semantic Modules and Residual Raw-Premise Completion}
\label{subsec:workload-shared-modules}

Fix a set of cached internal consequences
\[
  C\subseteq V_W\setminus(B\cup W).
\]
For each query \(q_\ell\), define the residual exposed-leaf set after caching
\(C\) by
\begin{equation}
  E_\ell(C):=D_{\mathrm{exp}}(q_\ell,C),
  \qquad \ell=1,\ldots,L,
\label{eq:def-workload-residual-family}
\end{equation}
and define the residual exposed-leaf union by
\begin{equation}
  E_{\cup}(C)
  :=
  \bigcup_{\ell=1}^{L}E_\ell(C).
\label{eq:def-workload-residual-union}
\end{equation}
Thus \(E_\ell(C)\) and \(E_{\cup}(C)\) are exact residual objects, not merely
unions of single-vertex postdominator sets.

\begin{corollary}[Exact raw-premise completion and its cost]
\label{cor:workload-completion-costs}
For a fixed internal module set \(C\subseteq V_W\setminus(B\cup W)\)
and every \(R\subseteq B\),
\begin{equation}
  D_{\mathrm{exp}}(q_\ell,C\cup R)=E_\ell(C)\setminus R
  \quad (\ell=1,\ldots,L),
  \qquad
  D_{\cup}(W,C\cup R)=E_{\cup}(C)\setminus R.
\label{eq:workload-residual-completion-joint}
\end{equation}
Hence the minimum additional raw-premise cost needed to satisfy the
maximal-error criterion is
\begin{align}
  &\tau_{\max}(C;\varepsilon,\delta)
  :=
  \min\Bigl\{
    \ell_c(R):
    R\subseteq E_{\cup}(C),
    \nonumber\\[-1mm]
    &|R\cap E_\ell(C)|
      \ge
      \bigl(|E_\ell(C)|-N^*(\varepsilon,\delta)\bigr)^+,
      \ \ell=1,\ldots,L
  \Bigr\},
\label{eq:workload-completion-cost-max}
\end{align}
a binary weighted set-multicover program; for the joint criterion,
sorting the costs of the leaves in \(E_{\cup}(C)\) as
\(c_{(1)}^{\cup}(C)\le\cdots\le c_{(|E_{\cup}(C)|)}^{\cup}(C)\) gives
\begin{equation}
  \tau_{\mathrm{joint}}(C;\varepsilon,\delta)
  =
  \sum_{i=1}^{(|E_{\cup}(C)|-N^*)^+}
  c_{(i)}^{\cup}(C),
\label{eq:workload-completion-cost-joint}
\end{equation}
where \(N^*=N^*(\varepsilon,\delta)\); under a homogeneous leaf cost
\(c_B\) this equals \((|E_{\cup}(C)|-N^*)^+c_B\).
\end{corollary}

\begin{proof}
Adding a cached source leaf removes only that leaf from the exposed
set, as in Proposition~\ref{prop:module-root-achievability}; taking
unions over \(\ell\) gives
\eqref{eq:workload-residual-completion-joint}.
The maximal-error constraints then reduce to the displayed multicover
program, and joint reliability requires caching at least
\((|E_{\cup}(C)|-N^*)^+\) residual leaves, attained by the cheapest
ones.
\end{proof}

\begin{definition}[Workload-globally protected leaf set]
\label{def:workload-global-protection}
For an internal consequence
\[
  r\in V_W\setminus(B\cup W),
\]
define its workload support set by
\begin{equation}
  \mathcal I_W(r)
  :=
  \{\ell\in\{1,\ldots,L\}: r\in V(G_\ell)\}.
\label{eq:def-workload-support}
\end{equation}
A set \(H\subseteq B\) is said to be \emph{\(W\)-globally protected by \(r\)}
if
\begin{align}
  H
  &\subseteq
  \bigcap_{\ell\in\mathcal I_W(r)}
  \operatorname{Pdom}_{q_\ell}(r),
\label{eq:global-protected-local}\\
  b\in H,\ b\in D_\ell
  &\Longrightarrow
  \ell\in\mathcal I_W(r),
  \qquad \ell=1,\ldots,L.
\label{eq:global-protected-workload}
\end{align}
Thus every workload query that depends on a leaf in \(H\) routes through the
same root \(r\).
By the hereditary parent map, any leaf postdominated by \(r\)
is a predecessor leaf of \(r\), and therefore belongs to every
query DAG containing \(r\).
The second condition excludes workload queries that use such a
leaf without containing the module.
\end{definition}

\begin{theorem}[Explicit shared-module gain under the joint criterion]
\label{thm:workload-shared-module-gain}
Let
\[
  C=\{r_1,\ldots,r_M\}\subseteq V_W\setminus(B\cup W)
\]
and, for each \(j\), let \(H_j\subseteq B\) be \(W\)-globally protected
by \(r_j\), with \(H_1,\ldots,H_M\) pairwise disjoint.  Then
\begin{equation}
  E_{\cup}(C)
  \subseteq
  D_{\cup}\setminus\bigcup_{j=1}^{M}H_j,
  \qquad
  |E_{\cup}(C)|
  \le
  |D_{\cup}|-\sum_{j=1}^{M}|H_j|,
\label{eq:workload-protected-union}
\end{equation}
and, under a homogeneous leaf cost \(c_B\),
\begin{equation}
  \sigma_{\mathrm{sem}}^{*,\mathrm{joint}}
  (W, \varepsilon,  \delta)
  \le
  \sum_{j=1}^{M}c(r_j)
   +
  \biggl(
    |D_{\cup}|-\sum_{j=1}^{M}|H_j|-N^*(\varepsilon,\delta)
  \biggr)^{+}c_B.
\label{eq:workload-protected-achievable}
\end{equation}
Consequently, writing \(N^*=N^*(\varepsilon,\delta)\) and assuming
\(|D_{\cup}|-\sum_{j=1}^{M}|H_j|\ge N^*\), the construction strictly
improves on the exact joint leaf-only baseline
\eqref{eq:workload-leaf-opt-joint-hom} whenever
\(\sum_{j=1}^{M}c(r_j)<\bigl(\sum_{j=1}^{M}|H_j|\bigr)c_B\);
in the uniform case \(c(r_j)=c_I\), \(|H_j|\ge s\ge1\), \(M\ge1\), this
sufficient condition reads
\begin{equation}
  c_I<s\,c_B.
\label{eq:workload-protected-gain-uniform}
\end{equation}
\end{theorem}

\begin{proof}
Every workload path from a leaf in \(H_j\) passes
through \(r_j\).
Thus caching \(C\) removes
\(\bigcup_j H_j\) from the exposed union.
Disjointness gives the count, and
Corollary~\ref{cor:workload-completion-costs}
gives the achievable cost.
Under the residual-slack condition, subtracting
this bound from the leaf-only optimum gives
\[
  \left(\sum_j|H_j|\right)c_B-\sum_jc(r_j).
\]
Its positivity proves the gain conditions.
\end{proof}

\subsection{Exact Optimality of the Shared-Module Construction}
\label{subsec:shared-module-optimality}

We now impose exact module routing.
This makes the shared-module construction optimal
over all answer-excluding transparent caches under
joint reliability.

\begin{definition}[Exact module routing]
\label{def:no-spurious-interception}
Let \(W=\{q_1,\ldots,q_L\}\), and let
\[
  R_{\mathrm{mod}}=\{r_1,\ldots,r_M\}
  \subseteq V_W\setminus(B\cup W)
\]
be a set of module roots.
Let \(H_1,\ldots,H_M\subseteq D_\cup\) be pairwise
disjoint nonempty leaf groups; \(M=0\) is allowed.
Put
\[
  P_{\mathrm{raw}}
  :=
  D_\cup\setminus\bigcup_{j=1}^{M}H_j.
\]

The workload satisfies \emph{exact module routing} if
\[
  V_W=D_\cup\cup R_{\mathrm{mod}}\cup W,
  \qquad
  \operatorname{pa}(r_j)=H_j
  \quad (j=1,\ldots,M),
\]
and, for each \(\ell\), there are sets
\[
  J_\ell\subseteq\{1,\ldots,M\},
  \qquad
  P_\ell\subseteq P_{\mathrm{raw}}
\]
such that
\[
  \operatorname{pa}(q_\ell)
  =
  \{r_j:j\in J_\ell\}\cup P_\ell.
\]
Consequently,
\[
  D_\ell
  =
  P_\ell\cup\bigcup_{j\in J_\ell}H_j.
\]

Each query DAG has depth at most two.
Every dependency path from a leaf in \(H_j\) to a
workload query passes through its assigned root \(r_j\);
raw leaves are direct query dependencies.
Given the explicit workload union, roots, and groups,
these conditions can be checked in linear time in
their combined encoding size.
\end{definition}

In particular, each \(H_j\) is \(W\)-globally protected
by \(r_j\) in the sense of
Definition~\ref{def:workload-global-protection}; see
Figure~\ref{fig:modules}(a).

\begin{theorem}[Exact optimality of shared modules]
\label{thm:shared-module-exact-optimality}
Assume exact module routing
(Definition~\ref{def:no-spurious-interception}),
a common module cost \(c(r_j)=c_I\), and a homogeneous
leaf cost \(c_B\).
Write \(N^*:=N^*(\varepsilon,\delta)\), and sort the
group sizes \(s_j:=|H_j|\) as
\(s_{(1)}\ge\cdots\ge s_{(M)}\).
The following statements hold.

\begin{enumerate}[label=\textup{(\Roman*)}]
\item \emph{Exact residual characterization.}  For any module set
\(J\subseteq\{1,\ldots,M\}\),
\begin{equation}
  E_{\cup}(\{r_j:j\in J\})
  =
  D_{\cup}\setminus\bigcup_{j\in J}H_j.
\label{eq:exact-residual-union}
\end{equation}

\item \emph{Exact joint optimum.}  Under the joint criterion,
\begin{equation}
  \sigma^{*,\mathrm{joint}}_{\mathrm{sem}}(W,\varepsilon,\delta)
  =\min_{0\le k\le M}\Bigl\{kc_I+\Bigl(|D_\cup|-N^*
   -\textstyle\sum_{j=1}^{k}s_{(j)}\Bigr)^+c_B\Bigr\},
  \label{eq:exact-joint-optimum}
\end{equation}
which is computable in \(O(1+M\log(M+1))\) arithmetic
operations after the workload representation has been
read and checked.

\item \emph{Uniform group sizes.}  If every group has size \(s\) and
\(t:=(|D_\cup|-N^*)^+\), the minimum is attained at \(k=0\) when
\(c_I\ge sc_B\), and otherwise at \(k=\min\{M,(\lceil t/s\rceil-1)^+\}\)
or at \(k=\min\{M,\lceil t/s\rceil\}\); two candidates survive because
the final module may replace fewer than \(s\) leaves.
\end{enumerate}
\end{theorem}

\begin{proof}
\textbf{(I)}
If \(b\in H_j\), every workload path from \(b\)
passes through \(r_j\), and no other module root
lies on such a path.
Thus selecting module roots blocks exactly their
assigned groups.
Raw leaves retain their direct dependency paths.
This proves \eqref{eq:exact-residual-union}.

\textbf{(II)}
By Theorem~\ref{thm:exact-workload-opt} and
Definition~\ref{def:no-spurious-interception},
the only candidates are module roots and leaves.
For selected roots \(\{r_j:j\in J\}\), part~(I) gives residual
union size
\[
  |D_\cup|-\sum_{j\in J}s_j.
\]
The minimum leaf completion
(Corollary~\ref{cor:workload-completion-costs})
is therefore
\[
  \left(
    |D_\cup|-\sum_{j\in J}s_j-N^*
  \right)^+c_B.
\]
For \(|J|=k\), the \(k\) largest groups minimize this completion, so the
optimum over subsets reduces to the prefix minimum
\eqref{eq:exact-joint-optimum}; sorting and prefix evaluation take
\(O(1+M\log(M+1))\) arithmetic operations.

\textbf{(III)}
For uniform group size, write
\[
  f(k):=kc_I+(t-ks)^+c_B.
\]
Its increments satisfy
\[
  f(k+1)-f(k)
  =
  c_I-c_B\min\{s,(t-ks)^+\}.
\]
If \(c_I\ge sc_B\), all increments are nonnegative.
Otherwise they are negative while \(t-ks\ge s\), positive once
\(ks\ge t\), and of either sign at the single remaining step, where
\(0<t-ks<s\); clipping to \(\{0,\ldots,M\}\) yields the two candidates.
For \(t=0\), both candidates are zero.
\end{proof}

\begin{figure}[tbp]
\centering
\begin{adjustbox}{max width=\linewidth}
\begin{tikzpicture}[ppi]
\node[leaf] (a1) at (0,0)   {$a_1$};
\node[leaf] (a2) at (0.8,0) {$a_2$};
\node[leaf] (a3) at (1.6,0) {$a_3$};
\node[leaf] (p1) at (2.6,0) {$p_1$};
\node[leaf] (a4) at (3.6,0) {$a_4$};
\node[leaf] (a5) at (4.4,0) {$a_5$};
\node[leaf] (a6) at (5.2,0) {$a_6$};
\node[leaf] (p2) at (6.2,0) {$p_2$};
\node[inter] (r1) at (0.8,1.5) {$r_1$};
\node[inter] (r2) at (4.4,1.5) {$r_2$};
\node[query] (q1) at (1.5,3) {$q_1$};
\node[query] (q2) at (3.1,3) {$q_2$};
\node[query] (q3) at (5.3,3) {$q_3$};
\foreach \x in {a1,a2,a3} \draw[arc] (\x)--(r1);
\foreach \x in {a4,a5,a6} \draw[arc] (\x)--(r2);
\draw[arc] (r1)--(q1); \draw[arc] (r1)--(q2);
\draw[arc] (r2)--(q2); \draw[arc] (r2)--(q3);
\draw[arc] (p1)--(q1); \draw[arc] (p2)--(q3);
\begin{scope}[on background layer]
\node[region,fit=(a1)(a3),inner sep=3pt,label={[note]below:$H_1$}] {};
\node[region,fit=(a4)(a6),inner sep=3pt,label={[note]below:$H_2$}] {};
\end{scope}
\node[note] at (3.1,-1.05) {(a) exact module routing; raw leaves
  \(p_1,p_2\)};
\begin{scope}[xshift=7.6cm,x=0.8cm,y=0.28cm]
\draw[->] (0,0)--(4.6,0) node[right]{$k$};
\draw[->] (0,0)--(0,11) node[above]{$f(k)$};
\foreach \k in {0,...,4} \draw (\k,0)--(\k,-0.35)
  node[below]{\scriptsize$\k$};
\foreach \yy in {5,10} \draw (0,\yy)--(-0.12,\yy)
  node[left]{\scriptsize$\yy$};
\draw[thick,blue!60!black]
  (0,10)--(1,7.5)--(2,5)--(3,4.5)--(4,6);
\foreach \k/\yy in {0/10,1/7.5,2/5,3/4.5,4/6}
  \fill[blue!60!black] (\k,\yy) circle (1.6pt);
\draw[red!70!black,thick] (3,4.5) circle (4pt);
\node[note,anchor=north] at (2.2,-2.2)
  {(b) \(f(k)=kc_I+(t-ks)^+c_B\) with\\
   \(t=10,\ s=4,\ c_I=1.5,\ c_B=1,\ M=4\)};
\end{scope}
\end{tikzpicture}
\end{adjustbox}
\caption{Shared modules under the joint criterion.  (a)~Each group
\(H_j\) reaches the workload only through its root \(r_j\); raw leaves
are direct dependencies.  (b)~The cost \(f(k)\) of
Theorem~\ref{thm:shared-module-exact-optimality}(III) on a separate
instance with uniform groups.  The candidates are
\(\lceil t/s\rceil-1=2\) and \(\lceil t/s\rceil=3\); the third module
replaces only two leaves but is still worth its cost.}
\label{fig:modules}
\end{figure}

\section{Discussion and Conclusion}
\label{sec:discussion}

\paragraph*{Scope}
\label{para:discussion-strong-assumption}
The exact reduction to PPI requires availability to obey
the designated conjunctive parent map.
A recorded preferred proof in a richer deductive system
does not by itself establish this semantics.
All algorithmic bounds use explicit dependency graphs;
compilation from compact rule descriptions is a separate
problem.

\paragraph*{Alternative derivations and coded storage}
\label{para:discussion-nondeterminism}
With alternative derivations, recovery is a union of
witness-survival events rather than the survival of one
exposed set.
\ref{app:envelope} characterizes this event and
bounds the resulting converse slack.
Its certificates do not automatically extend the
canonical optimization algorithms to compact AND--OR
representations.

The coded benchmark in \ref{sec:coded-outer}
requires full recovery of the relevant premise payloads.
Under the stated comparison convention, this suffices
for query recovery but can require more information
than the workload outputs.
A transparent query cache may therefore be smaller
than this benchmark without contradicting the
expressiveness of coded storage.

\paragraph*{Algorithmic consequences}
For fixed reliability parameters, the residual
threshold is constant and enumeration gives a
polynomial-time algorithm.
When the threshold is part of the input, the cost
model and graph structure become decisive.
In particular, the weighted depth-two hardness
construction becomes polynomial-time solvable after
its leaf costs are replaced by unit costs, whereas
depth three supports a unit-cost hardness construction.
The fan-out algorithm additionally depends on terminal
inadmissibility.

\paragraph*{Conclusion}
\label{para:discussion-limitations}
\label{para:discussion-partial}
PPI provides a common formulation for partial directed
separation and proof-valid caching under conjunctive
dependencies.
The results distinguish exact structural algorithms,
residual-threshold hardness, and approximation through
minimum cuts and submodularity.
For shared workloads, joint reliability depends on
the exposed-leaf union, while maximal-error reliability
retains separate query constraints.
Exact module routing makes the joint optimization
solvable by sorting.

The questions in
Subsection~\ref{subsec:summary-open} concern the
remaining depth-two parameterizations and
demand-dependent approximation.
Further extensions include correlated premise
erasures, compact AND--OR optimization, and coded
converses tailored to workload outputs rather than
full payload recovery.

\appendix
\renewcommand{\theassumption}{\Alph{section}.\arabic{assumption}}
\renewcommand{\thedefinition}{\Alph{section}.\arabic{definition}}
\renewcommand{\thetheorem}{\Alph{section}.\arabic{theorem}}
\renewcommand{\theremark}{\Alph{section}.\arabic{remark}}
\renewcommand{\thecorollary}{\Alph{section}.\arabic{corollary}}
\renewcommand{\thelemma}{\Alph{section}.\arabic{lemma}}
\renewcommand{\thefigure}{\arabic{figure}}
\renewcommand{\thetable}{\arabic{table}}

\section{The AND--OR Envelope Law and Its Consequences}
\label{app:envelope}

Fix \(q\in\mathcal Q\setminus B\).
We replace the canonical parent tuple by a finite family
of admissible parent sets.
Recovery is then a union of witness-survival events.
Figure~\ref{fig:and-or} gives an example.
The following assumption replaces
Assumptions~\ref{assump:det-local}
and~\ref{assump:hereditary-dag}.

\begin{assumption}[AND--OR local semantics with well-founded expansion]
\label{assump:and-or-local}
Retain the closure axioms, decidability, finite relevant
region, and base irreducibility of
Assumption~\ref{assump:det-local}.
Replace the designated-parent map by a computable
\emph{admissible parent family}: for every
\(v\in\mathcal V_{\mathrm{rel}}(B)\setminus B\),
\[
  \operatorname{Par}(v)
  \subseteq
  2^{\,\mathcal V_{\mathrm{rel}}(B)\setminus\{v\}}
\]
is a nonempty finite family of finite nonempty sets, and for every finite
\(\Gamma\subseteq\Cn(B)\) and every
\(v\in\mathcal V_{\mathrm{rel}}(B)\setminus B\),
\begin{equation}
  v\in\Cn(\Gamma)
  \quad\Longleftrightarrow\quad
  v\in\Gamma
  \ \lor\
  \exists P\in\operatorname{Par}(v):
  P\subseteq\Cn(\Gamma).
  \label{eq:and-or-local-law}
\end{equation}
Thus a derived object is available when it is stored or when \emph{at least
one} of its admissible parent sets is fully derivable---an AND over each
set and an OR over the family.  Write \(r_v:=|\operatorname{Par}(v)|\),
and call \(v\) \emph{ambiguous} if \(r_v\ge 2\).
The canonical semantics is the special case
\(r_v=1\): \(\operatorname{pa}(v)\) is the
unique member of \(\operatorname{Par}(v)\).

Finally, backward expansion is well-founded: there exists a rank function
\[
  \operatorname{rk}:\mathcal V_{\mathrm{rel}}(B)\longrightarrow\mathbb Z_{\ge0},
  \qquad
  \operatorname{rk}^{-1}(0)=B,
\]
such that
\begin{equation}
  u\in P\in\operatorname{Par}(v)
  \quad\Longrightarrow\quad
  \operatorname{rk}(u)<\operatorname{rk}(v).
\label{eq:well-founded-rank}
\end{equation}
\end{assumption}

\begin{definition}[Candidate witness DAGs and best-case residual]
\label{def:candidate-witness}
Let \(S\subseteq\Cn(B)\) be a finite transparent cache.  An \emph{\(S\)-relative
candidate witness DAG} for \(q\) is a finite acyclic digraph \(G\) with
\(V(G)\subseteq\mathcal V_{\mathrm{rel}}(B)\) and \(q\in V(G)\) such that:
\begin{enumerate}[label=\textup{(\roman*)}]
\item every vertex of \(G\) lies on a directed path terminating at \(q\);
\item every non-source vertex belongs to
\(\mathcal V_{\mathrm{rel}}(B)\setminus B\) and has
in-neighborhood
\(N_G^-(v)\in\operatorname{Par}(v)\);
\item every source of \(G\) lies in \(B\cup S\).
\end{enumerate}
A base vertex is always a source, since no admissible parent set is defined
on \(B\).  Denote the family of all such DAGs by \(\mathcal G_S(q,B)\), and
abbreviate the \emph{\(B\)-grounded} subfamily by
\[
  \mathcal G(q,B):=\mathcal G_{\varnothing}(q,B).
\]

For \(G\in\mathcal G_S(q,B)\), define
\[
  U_G(q,S):=S\cap V(G),
  \qquad
  V_0(G):=V(G)\cap B,
\]
and the exposed-leaf set
\begin{equation}
  \begin{split}
    D_{\mathrm{exp}}^{G}(q,S)
    := \bigl\{ b\in V_0(G) : \;
    &\exists\text{ a directed path }P:b\leadsto q
      \text{ in }G\\
    &\text{with }V(P)\cap U_G(q,S)=\varnothing \bigr\},
  \end{split}
\label{eq:def-exposed-leaves-witness}
\end{equation}
exactly as in Definition~\ref{def:query-local-residual} but with directed
paths taken in \(G\).  Write
\[
  e_G(q,S):=|D_{\mathrm{exp}}^{G}(q,S)|,
  \quad
  e_*(q,S):=\min_{G\in\mathcal G_S(q,B)}e_G(q,S).
\]
\end{definition}

The witness family is nonempty, since consistent parent choices
terminate in \(B\) by rank decrease.  A witness \(G\) is determined by
the map sending each \(v\in\mathcal V_{\mathrm{rel}}(B)\setminus B\) to
\(N_G^-(v)\) if \(v\) is a non-source of \(G\), and to a null symbol
otherwise: backward expansion from \(q\) along this map reconstructs
\(G\), the null symbol marking both cached sources and absent vertices.
Hence
\begin{equation}
  |\mathcal G_S(q,B)|
  \le
  \prod_{v\in\mathcal V_{\mathrm{rel}}(B)\setminus B}
  (1+r_v).
  \label{eq:candidate-count-bound}
\end{equation}

\begin{lemma}[Completion of cache-relative witnesses]
\label{lem:witness-completion}
For every finite transparent cache \(S\) and every
\(G\in\mathcal G_S(q,B)\), there is a \(B\)-grounded witness
\(\widehat G\in\mathcal G(q,B)\) such that
\[
  G\subseteq\widehat G,
  \qquad
  D_{\mathrm{exp}}^{\widehat G}(q,S)
  =
  D_{\mathrm{exp}}^{G}(q,S).
\]
Consequently,
\[
  e_*(q,S)
  =
  \min_{G\in\mathcal G(q,B)}e_G(q,S),
  \qquad
  |\mathcal G(q,B)|
  \le
  \prod_{v\in\mathcal V_{\mathrm{rel}}(B)\setminus B}r_v.
\]
\end{lemma}

\begin{proof}
Preserve the parent choice of every non-source vertex of \(G\).
Expand each non-base source of \(G\), which necessarily belongs to
\(S\), using one admissible parent set, and continue recursively,
recording one choice per vertex and reusing any choice already
recorded.
Rank decrease guarantees termination and acyclicity.
The resulting graph \(\widehat G\) is \(B\)-grounded, contains
\(G\), and every vertex reaches \(q\).

Every cache-free path of \(G\) remains cache-free in
\(\widehat G\).
Conversely, trace a cache-free path in \(\widehat G\) backward
from \(q\).
At a non-source vertex of \(G\) the in-neighborhood has not
changed, and at a base vertex there is nothing to expand.
Leaving \(G\) backward is therefore possible only through a
non-base source of \(G\), which belongs to \(S\) and hence to
\(U_{\widehat G}(q,S)\), contradicting cache-freeness.
The trace thus remains in \(G\), and the two exposed sets
coincide; leaves added by the expansion reach \(q\) only through
a vertex of \(S\) and are not exposed.

Since \(\mathcal G(q,B)\subseteq\mathcal G_S(q,B)\), the
minimum-residual identity follows.
Finally, a complete assignment of one admissible parent set to
every relevant derived vertex determines one \(B\)-grounded
witness by expansion from \(q\), and every such witness arises
from at least one assignment, giving the counting bound.
\end{proof}

\begin{figure}[tbp]
\centering
\begin{adjustbox}{max width=\linewidth}
\begin{tikzpicture}[ppi]
\node[leaf] (b1) at (0,0)   {$b_1$};
\node[leaf] (b2) at (1,0)   {$b_2$};
\node[leaf] (b3) at (2.2,0) {$b_3$};
\node[leaf] (b4) at (3.4,0) {$b_4$};
\node[leaf] (b5) at (4.6,0) {$b_5$};
\node[andnode] (au)  at (0.5,0.8) {};
\node[andnode] (av1) at (2.9,0.8) {};
\node[andnode] (av2) at (4.6,0.8) {};
\node[inter] (u) at (0.5,1.6) {$u$};
\node[inter] (v) at (3.8,1.6) {$v$};
\node[andnode] (aq1) at (1.3,2.4) {};
\node[andnode] (aq2) at (3.8,2.4) {};
\node[query] (q) at (2.5,3.3) {$q$};
\draw[thick,black!50] (b1)--(au) (b2)--(au) (b3)--(av1) (b4)--(av1)
  (b5)--(av2) (u)--(aq1) (b3)--(aq1) (v)--(aq2);
\draw[arc] (au)--(u); \draw[arc] (av1)--(v); \draw[arc] (av2)--(v);
\draw[arc] (aq1)--(q); \draw[arc] (aq2)--(q);
\node[draw=black!40,rounded corners=3pt,fill=black!3,anchor=west,
  align=left,font=\footnotesize,inner sep=5pt] at (5.4,1.6)
  {\(\operatorname{Par}(q)=\{\{u,b_3\},\{v\}\}\)\\
   \(\operatorname{Par}(u)=\{\{b_1,b_2\}\}\),
   \(\operatorname{Par}(v)=\{\{b_3,b_4\},\{b_5\}\}\)\\[2pt]
   witness exposed sets: \(\{b_1,b_2,b_3\},\ \{b_3,b_4\},\ \{b_5\}\)\\[2pt]
   recovery \(=X_1X_2X_3\lor X_3X_4\lor X_5\)\\[2pt]
   \(e_*=1\):\ \((1-\varepsilon)\le\Pr\le\sum_G(1-\varepsilon)^{e_G}\)};
\end{tikzpicture}
\end{adjustbox}
\caption{Alternative derivations.  Each black dot is one admissible
parent set, whose inputs are joined conjunctively; distinct dots
entering a vertex are alternatives.  With \(S=\varnothing\), choosing
one dot at each visited vertex yields a witness DAG, and recovery is
the monotone DNF over the witnesses' exposed sets
(Theorem~\ref{thm:envelope-rigidity}).}
\label{fig:and-or}
\end{figure}

\begin{theorem}[Envelope rigidity under non-unique derivations]
\label{thm:envelope-rigidity}
Under Assumption~\ref{assump:and-or-local}, suppose \(q\in\Cn(B)\) and let
\(S\subseteq\Cn(B)\) be finite.  Then:
\begin{enumerate}[label=\textup{(\alph*)}]
\item for every \(\widetilde B\subseteq B\),
\begin{equation}
  q\in\Cn(\widetilde B\cup S)
  \Longleftrightarrow
  \exists\,G\in\mathcal G_S(q,B):
  \ D_{\mathrm{exp}}^{G}(q,S)\subseteq\widetilde B;
\label{eq:envelope-event}
\end{equation}
by Lemma~\ref{lem:witness-completion}, the same event can be
expressed using only \(B\)-grounded witnesses:
\[
  q\in\Cn(\widetilde B\cup S)
  \Longleftrightarrow
  \exists\,G\in\mathcal G(q,B):
  \ D_{\mathrm{exp}}^{G}(q,S)\subseteq\widetilde B;
\]
equivalently, with survival indicators
\(X_b:=\mathbf 1\{b\in\widetilde B\}\), the recovery indicator is the
monotone DNF
\begin{equation}
  \mathbf 1\{q\in\Cn(\widetilde B\cup S)\}
  =
  \bigvee_{G\in\mathcal G_S(q,B)}
  \ \bigwedge_{b\in D_{\mathrm{exp}}^{G}(q,S)} X_b;
\label{eq:envelope-dnf}
\end{equation}
with the convention that an empty conjunction is interpreted as
true;
\item under independent premise erasures with erasure probability
\(\varepsilon\),
\begin{align}
  (1-\varepsilon)^{e_*(q,S)}
  &\le
  \Pr_{\varepsilon}\!\left[q\in\Cn(\widetilde B\cup S)\right]
  \nonumber\\
  &\le
  \min\!\left\{
    1,\;
    \sum_{G\in\mathcal G(q,B)}(1-\varepsilon)^{e_G(q,S)}
  \right\}
  \nonumber\\
  &\le
  \min\!\left\{
    1,\;
    |\mathcal G(q,B)|\,(1-\varepsilon)^{e_*(q,S)}
  \right\};
\label{eq:envelope-probability}
\end{align}
\item if
\[
  D_{\mathrm{exp}}^{G}(q,S)
  \cap D_{\mathrm{exp}}^{H}(q,S)=\varnothing
\]
for every pair of distinct witnesses
\(G,H\in\mathcal G(q,B)\), then exactly
\begin{equation}
  \Pr_{\varepsilon}\!\left[q\in\Cn(\widetilde B\cup S)\right]
  =
  1-\prod_{G\in\mathcal G(q,B)}
  \Bigl(1-(1-\varepsilon)^{e_G(q,S)}\Bigr);
\label{eq:envelope-independent}
\end{equation}
\item if \(r_v=1\) for every \(v\), then \(e_*(q,S)=|D_{\mathrm{exp}}(q,S)|\)
and the lower bound in \eqref{eq:envelope-probability} coincides with the
exact law \eqref{eq:exact-residual-probability}.
\end{enumerate}
\end{theorem}

\begin{proof}
(a)\ \emph{Sufficiency within a fixed witness.}
Fix \(G\in\mathcal G_S(q,B)\) and suppose
\(D_{\mathrm{exp}}^{G}(q,S)\subseteq\widetilde B\), yet
\(q\notin\Cn(\widetilde B\cup S)\).  Then \(q\notin S\), so \(q\) is not a
source of \(G\), and by \eqref{eq:and-or-local-law} applied with
\(P=N_G^{-}(q)\) some \(u_1\in N_G^{-}(q)\) satisfies
\(u_1\notin\Cn(\widetilde B\cup S)\).  Iterating, any unavailable non-source
vertex has an unavailable in-neighbor in \(G\); unavailable vertices cannot
belong to \(S\); following predecessor edges backward strictly
decreases the rank \eqref{eq:well-founded-rank}, so the trace
terminates at an unavailable source \(b\).  A source in \(S\) would be available, and an unavailable base source
satisfies \(b\notin\widetilde B\cup S\) by \eqref{eq:base-irreducibility}.
Hence \(b\in B\setminus\widetilde B\), and the traced path
\(b\leadsto q\) meets no vertex of \(U_G(q,S)\), so
\(b\in D_{\mathrm{exp}}^{G}(q,S)\), contradicting
\(D_{\mathrm{exp}}^{G}(q,S)\subseteq\widetilde B\).

\emph{Extraction of a surviving witness.}
Suppose
\(q\in\Cn(\widetilde B\cup S)\), and write
\[
  A:=\Cn(\widetilde B\cup S)
      \cap\mathcal V_{\mathrm{rel}}(B).
\]
Expand backward from \(q\), processing each visited vertex only
once.
At a vertex in \(B\cup S\), stop.
At every other visited vertex \(v\), the local law provides a
set
\[
  P_v\in\operatorname{Par}(v),
  \qquad P_v\subseteq A.
\]
Choose one such set, record it, and add the arcs \(u\to v\) for
\(u\in P_v\).
If another branch later reaches \(v\), reuse this recorded
choice rather than selecting an additional parent set.

Backward expansion strictly decreases rank, so the construction
terminates.
The resulting graph is acyclic, every visited vertex reaches
\(q\), and each non-source has exactly its recorded admissible
parent set as its in-neighborhood.
Every source belongs to \(B\cup S\).
Hence the graph belongs to \(\mathcal G_S(q,B)\).

All visited vertices lie in \(A\).
If
\(b\in D_{\mathrm{exp}}^G(q,S)\), then \(b\in A\cap B\) and
\(b\notin S\).
Base irreducibility therefore gives
\(b\in\widetilde B\).
Thus
\[
  D_{\mathrm{exp}}^G(q,S)\subseteq\widetilde B,
\]
proving the reverse implication in
\eqref{eq:envelope-event}.
Equation~\eqref{eq:envelope-dnf} is an equivalent Boolean
representation of that event identity.

(b)\ By part~(a) and Lemma~\ref{lem:witness-completion},
recovery is the union of the witness-survival events
\(\{D_{\mathrm{exp}}^{G}(q,S)\subseteq\widetilde B\}\) indexed by
\(\mathcal G(q,B)\).  Each event has probability
\((1-\varepsilon)^{e_G(q,S)}\) exactly as in
Theorem~\ref{thm:exact-residual-law}.  The lower bound is the
largest single-event probability, and the upper bounds follow
from the union bound and the minimum-residual identity of
Lemma~\ref{lem:witness-completion}.

(c)\ Pairwise disjoint exposed sets make the survival events independent,
and the union probability evaluates to \eqref{eq:envelope-independent}.

(d)\ If every derived vertex has a unique admissible parent set,
let \(G_q\) be the canonical \(B\)-grounded witness.
Every \(S\)-relative witness is a subgraph of \(G_q\), so
\[
  D_{\mathrm{exp}}^{G}(q,S)
  \subseteq D_{\mathrm{exp}}(q,S).
\]

Conversely, let \(P:b\leadsto q\) be a cache-free path in \(G_q\).
Starting from \(q\) and following \(P\) backward, each visited
non-base vertex lies outside \(S\).
It therefore cannot be a source of an \(S\)-relative witness and
must expand along its unique parent set.
Induction along \(P\) shows that every vertex and edge of \(P\)
belongs to every \(S\)-relative witness.
Hence
\[
  D_{\mathrm{exp}}(q,S)
  \subseteq D_{\mathrm{exp}}^{G}(q,S).
\]
Thus all witnesses have the same exposed set, and
\[
  e_*(q,S)=|D_{\mathrm{exp}}(q,S)|.
\]
The lower bound in \eqref{eq:envelope-probability} is therefore the
exact probability in \eqref{eq:exact-residual-probability}.
\end{proof}

\begin{remark}[Reliability evaluation]
\label{rem:envelope-sharp-p}
Exact AND--OR reliability is \(\#\mathrm P\)-hard,
even without caching and with parent sets of size two.
For a graph \(F=(V,E)\), create leaves \(b_v\) and
\[
  \operatorname{Par}(q)
  =
  \{\{b_u,b_v\}:\{u,v\}\in E\}.
\]
At \(\varepsilon=1/2\),
\[
  P_{\mathrm{rec}}
  =
  1-\frac{\#\mathrm{IS}(F)}{2^{|V|}},
\]
giving a polynomial-time Turing reduction from
independent-set counting
\citep{valiant1979complexity,provan1983complexity}.
Edgeless graphs are handled directly.

For an explicitly listed witness-DNF, the
Karp--Luby estimator gives an FPRAS under rational
survival probabilities \citep{karp1983monte}.
This is polynomial in the DNF size, which may be
exponential in a compact AND--OR representation.
Relative-error estimation also does not decide a
reliability threshold without a separation margin.

In contrast, the canonical exposed set is computed
by one reverse reachability search; numerical
evaluation additionally depends on the encoding
of \(\varepsilon\).
\end{remark}

\paragraph*{Consequences for Certification and Converse Slack}

\begin{corollary}[Robust certification: single-witness designs remain valid]
\label{cor:envelope-certified-threshold}
Under Assumption~\ref{assump:and-or-local}, if there exists
\(G_0\in\mathcal G_S(q,B)\) with \(e_{G_0}(q,S)\le N^*(\varepsilon,\delta)\),
then \(S\) is \((\varepsilon,\delta)\)-reliable for \(q\).  It thus suffices
to verify the bound on a single \(B\)-grounded witness, for instance the
canonical DAG of the unique-tuple sub-model.
\end{corollary}

\begin{proof}
The envelope lower bound gives
\[
  P_{\mathrm{rec}}
  \ge
  (1-\varepsilon)^{e_{G_0}(q,S)}
  \ge1-\delta.
\]
\end{proof}

\begin{corollary}[Quantified relaxation cost: the converse slack]
\label{cor:envelope-relaxation-cost}
Under Assumption~\ref{assump:and-or-local}, define a uniform witness-count
budget, its induced slack, and the inflated threshold by
\begin{align}
  \Lambda_q
  &:=
  \sum_{v\in\mathcal V_{\mathrm{rel}}(B)\setminus B}\ln r_v,\nonumber\\
  \Delta_{\mathcal G}
  &:=
  \frac{\Lambda_q}{\ln\!\bigl(1/(1-\varepsilon)\bigr)},
  \nonumber\\
  N^+_{\mathcal G}(\varepsilon,\delta)
  &:=
  \left\lfloor
    \frac{
      \ln\!\bigl(1/(1-\delta)\bigr)+\Lambda_q
    }{
      \ln\!\bigl(1/(1-\varepsilon)\bigr)
    }
  \right\rfloor .
\label{eq:envelope-N-plus}
\end{align}
If \(S\) is \((\varepsilon,\delta)\)-reliable for \(q\), then
\begin{equation}
  e_*(q,S)
  \le
  N^+_{\mathcal G}(\varepsilon,\delta)
  \le
  N^*(\varepsilon,\delta)+\Delta_{\mathcal G}+1.
\label{eq:envelope-converse}
\end{equation}
The converse concerns the best witness residual.
For homogeneous leaf-only caches, put
\[
  \kappa_q^{\min}:=
  \min_{G\in\mathcal G(q,B)}|V_0(G)|.
\]
Then
\begin{equation}
  \begin{aligned}
    & \bigl(\kappa_q^{\min}-N^+_{\mathcal G}(\varepsilon,\delta)\bigr)^+ c_B \\
    \le &\sigma_{\mathrm{leaf,AO}}^{*}(q,\varepsilon,\delta) \\
    \le &\bigl(\kappa_q^{\min}-N^*(\varepsilon,\delta)\bigr)^+ c_B,
  \end{aligned}
\label{eq:envelope-leaf-sandwich}
\end{equation}
where \(\sigma_{\mathrm{leaf,AO}}^{*}\) is the leaf-only optimum under the
AND--OR semantics; the two sides differ by at most
\((\Delta_{\mathcal G}+1)c_B\).  Finally,
\begin{equation}
  \Delta_{\mathcal G}
  \le
  \frac{\Lambda_q}{\varepsilon}.
\label{eq:envelope-slack-bound}
\end{equation}
Thus the slack is controlled by the logarithmic parent-choice
counts and charges only ambiguous vertices.
The budget applies to arbitrary transparent caches, because
cache-relative witnesses can be completed without changing their
exposed sets (Lemma~\ref{lem:witness-completion}).
Restricting the sum to the expansion region of \(q\) can further
tighten the bound.
In the canonical regime, \(\Lambda_q=0\), and the converse
threshold equals \(N^*(\varepsilon,\delta)\).
\end{corollary}

\begin{proof}
By Lemma~\ref{lem:witness-completion},
\(|\mathcal G(q,B)|\le e^{\Lambda_q}\).  Reliability and the upper bound in
\eqref{eq:envelope-probability} give
\(1-\delta\le\Pr_{\varepsilon}[q\in\Cn(\widetilde B\cup S)]
\le e^{\Lambda_q}(1-\varepsilon)^{e_*(q,S)}\),
and solving for the integer \(e_*(q,S)\) yields
\eqref{eq:envelope-converse}; the comparison with \(N^*(\varepsilon,\delta)\)
is immediate from \eqref{eq:model-N-star}.  For
\eqref{eq:envelope-leaf-sandwich}, let \(S\subseteq B\).  Then every
\(S\)-relative candidate is \(B\)-grounded, and a path from a base leaf is
cache-free if and only if its first vertex is not cached, so
\(e_G(q,S)=|V_0(G)\setminus S|=\kappa_G-|S\cap V_0(G)|\) with
\(\kappa_G:=|V_0(G)|\).  The upper side of
\eqref{eq:envelope-leaf-sandwich} follows by caching any
\((\kappa_q^{\min}-N^*)^+\) leaves of a \(B\)-grounded witness attaining
\(\kappa_q^{\min}\) and applying
Corollary~\ref{cor:envelope-certified-threshold}.  For the lower side,
reliability forces \(e_*(q,S)\le N^+_{\mathcal G}\), i.e.\
\(|S\cap V_0(G)|\ge\kappa_G-N^+_{\mathcal G}\) for some \(G\), hence
\(|S|\ge(\kappa_q^{\min}-N^+_{\mathcal G})^+\).
Finally \eqref{eq:envelope-slack-bound} uses
\(\ln(1/(1-\varepsilon))\ge\varepsilon\).
\end{proof}

\begin{remark}[Sublinear ambiguity]
\label{rem:leaf-only-witness-budget}
Fix \(\varepsilon,\delta\in(0,1)\).
For a sequence of AND--OR instances, suppose that
\[
  \kappa_q^{\min}\to\infty,
  \qquad
  \Lambda_q=o(\kappa_q^{\min}).
\]
Then
\[
  N^+_{\mathcal G}(\varepsilon,\delta)=o(\kappa_q^{\min}),
\]
and the leaf-only sandwich gives
\[
  \sigma_{\mathrm{leaf,AO}}^*(q,\varepsilon,\delta)
  =
  \kappa_q^{\min}c_B
  +o(\kappa_q^{\min}c_B).
\]
Thus sublinear logarithmic ambiguity does not change the
first-order leaf-only storage requirement.
\end{remark}

\section{Comparison with Coded Storage}
\label{sec:coded-outer}

We study a benchmark that reconstructs all premise
payloads indexed by the workload dependency union
\(D_\cup\).
Under the comparison convention of
Subsection~\ref{subsec:erasure-cache-model},
recovering these payloads suffices for workload
recovery.
The benchmark therefore upper-bounds the coded
query-recovery optimum, but need not equal it.

Each premise carries a payload of
\(c_B\in\mathbb Z_{\ge1}\) bits.
The payload alphabet \(\mathcal X\) has size \(2^{c_B}\),
and a stored raw payload costs \(c_B\) bits.
A benchmark cache has a fixed length
\(\sigma\in\mathbb Z_{\ge0}\).

\begin{definition}[Coded outer-benchmark task]
\label{def:coded-benchmark-task}
Let \(A\subseteq B\) be a finite family of \(n:=|A|\) premises with
payloads \(X^n=(X_b)_{b\in A}\in\mathcal X^n\).  The observation is
\(Y^n=(Y_b)_{b\in A}\) with \(Y_b=X_b\) if premise \(b\) survives and
\(Y_b=\bot\) otherwise, so that, conditionally on the payloads, each
coordinate is independently revealed with probability
\(1-\varepsilon\) and erased with probability \(\varepsilon\).
A \emph{coded benchmark scheme} \((\Enc,\Dec,\sigma)\) consists of
computable maps
\[
  \Enc:\mathcal X^{n}\longrightarrow\{0,1\}^{\sigma},
  \;
  \Dec:\bigl(\mathcal X\cup\{\bot\}\bigr)^{n}\times\{0,1\}^{\sigma}
  \longrightarrow\mathcal X^{n}.
\]
The cache string \(T=\Enc(X^{n})\) is produced before the erasure
realization.  The scheme is \(\delta\)-reliable for \(A\) if, for
every payload realization,
\(\Pr[\Dec(Y^{n},T)=X^{n}]\ge1-\delta\),
with probability taken over the erasure process.  The minimum
admissible \(\sigma\) is denoted
\(\bar\sigma_{\mathrm{code}}(A,\varepsilon,\delta)\).
\end{definition}

\begin{definition}[MDS parity threshold]
\label{def:mds-threshold}
For \(n\in\mathbb Z_{\ge 0}\), define
\begin{equation}
    r_{\mathrm{MDS}}^*(n,\varepsilon,\delta)
    := \min \Bigl\{ r\in\{0,1,\ldots,n\}:
    \sum_{j=0}^{r} \binom{n}{j} \varepsilon^{j}(1-\varepsilon)^{n-j}
    \ge 1-\delta \Bigr\},
\label{eq:def-r-mds-star}
\end{equation}
the smallest parity count for which ideal certified
MDS decoding meets the target under independent source-packet
erasures.
For \(\delta\ge1\), set \(r_{\mathrm{MDS}}^*=0\); this
convention is used with \(\delta+s_{c_B}\) in
Theorem~\ref{thm:mds-near-optimality}.
\end{definition}

For achievability, we use systematic generalized
Reed--Solomon codes, extended when needed, under
\[
  n+r\le2^{c_B}+1
  \qquad(r\ge1).
\]
The \(r\) cached parity packets certify recovery
whenever at most \(r\) source packets are erased.

\begin{lemma}[Counting bound for payload recovery]
\label{lem:list-bound}
For a coded benchmark scheme for \(A\), let
\(p_{\mathrm c}(x^n)\) be its success probability on payload vector
\(x^n\), with probability taken over erasures, and let
\(P_{\mathrm c}^{\mathrm{avg}}
:=|\mathcal X|^{-n}\sum_{x^n\in\mathcal X^n}p_{\mathrm c}(x^n)\).
Then
\[
  P_{\mathrm c}^{\mathrm{avg}}
  \;\le\;
  \E\!\left[\min\bigl\{1,\,2^{\sigma-c_B E_n}\bigr\}\right],
  \qquad
  E_n\sim\mathrm{Bin}(n,\varepsilon).
\]
Since an average is at least its minimum, the same bound applies to
the worst-case success probability
\(P_{\mathrm c}^{\mathrm{wc}}:=\min_{x^n}p_{\mathrm c}(x^n)\), and a
\(\delta\)-reliable scheme has both probabilities at least
\(1-\delta\).
\end{lemma}

\begin{proof}
Average over independent uniform payloads.
Condition on an erasure pattern of size \(e\) and
the revealed values.
There are \(2^{c_Be}\) equiprobable completions.
For fixed revealed values, each of the \(2^\sigma\)
cache strings produces one decoder output, so at
most
\[
  \min\{2^{c_Be},2^\sigma\}
\]
completions are decoded correctly.
The conditional success probability is therefore
at most
\[
  \min\{1,2^{\sigma-c_Be}\}.
\]
Average over the erasure pattern.
Worst-case success is no greater than average success.
\end{proof}

\begin{theorem}[MDS benchmark bounds within one packet
and a reliability adjustment]
\label{thm:mds-near-optimality}
Let \(A\subseteq B\), write \(n:=|A|\), and let
\(E_n\sim\mathrm{Bin}(n,\varepsilon)\).
For any coded benchmark cache of length \(\sigma\), put
\(r:=\lfloor\sigma/c_B\rfloor\).
Its success probability, either worst-case over payloads or
averaged over uniform payloads, satisfies
\begin{equation}
  P_{\mathrm c}
  \le
  \Pr[E_n\le r+1]
  +\frac{1}{2^{c_B}-1}.
  \label{eq:mds-converse-sandwich}
\end{equation}
This converse does not require an MDS code or an alphabet-size
condition.

Conversely, let \(r_0\in\{0,\ldots,n\}\).
If \(r_0=0\), the uncoded scheme achieves
\(P_{\mathrm c}\ge\Pr[E_n=0]\).
If \(r_0\ge1\) and \(n+r_0\le2^{c_B}+1\), a systematic
MDS scheme storing \(r_0\) parity packets achieves
\(P_{\mathrm c}^{\mathrm{MDS}}\ge\Pr[E_n\le r_0]\)
for every payload realization.

Write \(s_{c_B}:=(2^{c_B}-1)^{-1}\) and
\(r^*:=r_{\mathrm{MDS}}^*(n,\varepsilon,\delta)\).
If \(r^*=0\) or \(n+r^*\le2^{c_B}+1\), then the minimum
benchmark size of Definition~\ref{def:coded-benchmark-task} satisfies
\begin{equation}
  \bigl(
    r_{\mathrm{MDS}}^*
    (n,\varepsilon,\delta+s_{c_B})-1
  \bigr)^+c_B
  \le
  \bar\sigma_{\mathrm{code}}(A,\varepsilon,\delta)
  \le
  r_{\mathrm{MDS}}^*(n,\varepsilon,\delta)c_B.
  \label{eq:sigma-bar-one-packet}
\end{equation}
\end{theorem}

\begin{proof}
\emph{Converse.}  Let \((\Enc,\Dec,\sigma)\) be any coded benchmark
scheme for \(A\), and let \(r:=\lfloor\sigma/c_B\rfloor\), so that
\(\sigma<c_B(r+1)\).  By Lemma~\ref{lem:list-bound},
\[
  P_{\mathrm c}^{\mathrm{wc}}
  \;\le\;
  P_{\mathrm c}^{\mathrm{avg}}
  \;\le\;
  \E\!\left[\min\bigl\{1,\,2^{\sigma-c_B E_n}\bigr\}\right].
\]
Split the expectation according to whether \(E_n\le r+1\).
Then
\[
  P_{\mathrm c}^{\mathrm{avg}}
  \le
  \Pr[E_n\le r+1]
  +
  \sum_{e=r+2}^{n}
  \Pr[E_n=e]\,2^{\sigma-c_Be}.
\]
Using \(\Pr[E_n=e]\le1\) and
\(\sigma<c_B(r+1)\), the second term is at most
\[
  \sum_{j=1}^{\infty}2^{-c_Bj}
  =
  \frac{1}{2^{c_B}-1}.
\]
This proves \eqref{eq:mds-converse-sandwich}.

\emph{Achievability.}
For \(n=0\), the empty cache suffices.
For \(r_0=0\), output the observed vector when no
source is erased and an arbitrary completion otherwise.
For \(r_0\ge1\), the alphabet condition permits a
systematic generalized Reed--Solomon MDS code of
length \(n+r_0\) and dimension \(n\), using the
extended form when necessary
\citep{singleton1964maximum,macwilliams1977theory}.
Store its \(r_0\) parity symbols.
At most \(r_0\) erased source symbols leave at least
\(n\) coded symbols available, so decoding recovers
every payload vector.
The cache costs \(r_0c_B\), and success probability
is at least \(\Pr[E_n\le r_0]\).

\emph{The sandwich \eqref{eq:sigma-bar-one-packet}.}  For the upper side,
take \(r_0=r_{\mathrm{MDS}}^*(n,\varepsilon,\delta)\); by the theorem
hypothesis, either \(r_0=0\) or \(n+r_0\le2^{c_B}+1\), so the scheme
above is \(\delta\)-reliable with \(\sigma=r_0\,c_B\) and
\(\bar\sigma_{\mathrm{code}}(A,\varepsilon,\delta)
\le r_{\mathrm{MDS}}^*(n,\varepsilon,\delta)\,c_B\).  For the lower side,
let \((\Enc,\Dec,\sigma)\) be any \(\delta\)-reliable scheme for \(A\).
\(\delta\)-reliability gives
\(1-\delta\le P_{\mathrm c}^{\mathrm{avg}}\), and by
\eqref{eq:mds-converse-sandwich},
\[
  1-\delta
  \;\le\;
  \Pr[E_n\le r+1]+\frac{1}{2^{c_B}-1},
  \qquad\text{i.e.}\qquad
  \Pr[E_n\le r+1]\;\ge\;1-(\delta+s_{c_B}).
\]
By the minimality of \(r_{\mathrm{MDS}}^*\) in
\eqref{eq:def-r-mds-star},
\(r+1\ge r_{\mathrm{MDS}}^*(n,\varepsilon,\delta+s_{c_B})\), and since
\(\sigma\ge r\,c_B\),
\[
  \sigma
  \;\ge\;
  \bigl(r_{\mathrm{MDS}}^*(n,\varepsilon,\delta+s_{c_B})-1\bigr)c_B .
\]
Combining with \(\sigma\ge0\) and minimizing over schemes gives the lower
side of \eqref{eq:sigma-bar-one-packet}.
\end{proof}

\begin{lemma}[Uniform binomial quantile expansion]
\label{lem:binomial-quantile}
Fix \(\varepsilon\in(0,1)\) and a compact interval
\(I_0\subset(0,1)\).
Then, uniformly over \(u\in I_0\) as \(n\to\infty\),
\[
  \min\bigl\{r\in\mathbb Z_{\ge0}:\Pr[E_n\le r]\ge u\bigr\}
  =
  n\varepsilon
  +
  \Phi^{-1}(u)\sqrt{n\varepsilon(1-\varepsilon)}
  +
  O(1),
\]
where \(E_n\sim\mathrm{Bin}(n,\varepsilon)\), \(\Phi\) is the
standard normal cdf, and the \(O(1)\) term depends only on
\((\varepsilon,I_0)\).
\end{lemma}

\begin{proof}
Let \(r_n(u)\) be the stated quantile and
\(s_n:=\sqrt{n\varepsilon(1-\varepsilon)}\).
By minimality,
\[
  \Pr[E_n\le r_n(u)-1]
  <u
  \le
  \Pr[E_n\le r_n(u)].
\]
The Berry--Esseen bound
\citep{feller1991introduction,shevtsova2010berry}
and the bounded normal density imply
\[
  \Phi\left(
    \frac{r_n(u)-n\varepsilon}{s_n}
  \right)
  =
  u+O(n^{-1/2})
\]
uniformly for \(u\in I_0\).
The inverse normal cdf is Lipschitz on a compact
neighborhood of \(I_0\).
Applying it and multiplying by \(s_n\)
gives the expansion.
\end{proof}

\begin{theorem}[First-order payload-recovery law and
the leaf-only storage ratio]
\label{thm:first-order-overhead-law}
Fix \(\varepsilon,\delta\in(0,1)\).
Consider a sequence of payload-recovery instances with
\(n\to\infty\) and integer payload lengths
\(c_B=c_B(n)\ge1\).
Put
\[
  r_n:=r_{\mathrm{MDS}}^*(n,\varepsilon,\delta),
\]
and assume that
\[
  n+r_n\le2^{c_B(n)}+1
\]
for all sufficiently large \(n\).
\begin{enumerate}[label=\textup{(\alph*)}]
\item For every \(A\subseteq B\) with \(|A|=n\),
\(\bar\sigma_{\mathrm{code}}(A,\varepsilon,\delta)
=\varepsilon n c_B+o(n c_B)\).
The MDS parity quantile additionally satisfies
\(r_{\mathrm{MDS}}^*(n,\varepsilon,\delta)
=\varepsilon n
+\Phi^{-1}(1-\delta)\sqrt{n\varepsilon(1-\varepsilon)}+O(1)\),
where \(\Phi\) is the standard normal cdf.
\item For a sequence of canonical queries with
\(\kappa_q=n\) and homogeneous transparent leaf cost \(c_B\),
the leaf-only optimum of Corollary~\ref{cor:leaf-only-opt}
satisfies
\begin{equation}
  \frac{\sigma_{\mathrm{leaf}}^*(q,\varepsilon,\delta)}
       {\bar\sigma_{\mathrm{code}}\bigl(V_0(G_q),\varepsilon,\delta\bigr)}
  =
  \frac{\bigl(n-N^*(\varepsilon,\delta)\bigr)^+}{\varepsilon n}
  \bigl(1+o(1)\bigr)
  \longrightarrow
  \frac{1}{\varepsilon}.
\label{eq:universal-overhead-factor}
\end{equation}
\end{enumerate}
\end{theorem}

\begin{proof}
(a)  Applying Lemma~\ref{lem:binomial-quantile} with \(u=1-\delta\)
gives the quantile expansion.  For \(\bar\sigma_{\mathrm{code}}\), use
the sandwich \eqref{eq:sigma-bar-one-packet}.  Under the
alphabet-growth convention,
\(s_{c_B}=(2^{c_B}-1)^{-1}\to0\), so \(1-\delta-s_{c_B}\) remains in a
compact subinterval of \((0,1)\); applying
Lemma~\ref{lem:binomial-quantile} with \(u=1-\delta-s_{c_B}\) and the
continuity of \(\Phi^{-1}\) gives
\(r_{\mathrm{MDS}}^*(n,\varepsilon,\delta+s_{c_B})
=n\varepsilon+O(\sqrt n)\).  Both sides of
\eqref{eq:sigma-bar-one-packet} are therefore
\(n\varepsilon c_B+O(\sqrt n\,c_B)\), and
\(\bar\sigma_{\mathrm{code}}(A,\varepsilon,\delta)
=\varepsilon n c_B+o(n c_B)\).

(b)  By Corollary~\ref{cor:leaf-only-opt},
\(\sigma_{\mathrm{leaf}}^*(q,\varepsilon,\delta)
=(n-N^*(\varepsilon,\delta))^+c_B\), and
\(N^*(\varepsilon,\delta)\) is independent of \(n\).  Dividing by
part (a) applied to \(A=V_0(G_q)\),
\[
  \frac{\sigma_{\mathrm{leaf}}^*(q,\varepsilon,\delta)}
       {\bar\sigma_{\mathrm{code}}\bigl(V_0(G_q),\varepsilon,\delta\bigr)}
  \;=\;
  \frac{(n-N^*)^+\,c_B}
       {\varepsilon n c_B+O(\sqrt n\,c_B)}
  \;=\;
  \frac{(n-N^*)^+}{\varepsilon n}\,
  \frac{1}{1+O(n^{-1/2})},
\]
which tends to \(1/\varepsilon\), since
\((n-N^*)^+/(\varepsilon n)\to1/\varepsilon\).
\end{proof}

\begin{remark}[Module-to-benchmark ratio]
\label{rem:asymptotic-overhead-alpha}
Consider a sequence of workloads with
\[
  d:=|D_\cup|\to\infty
\]
and exact module routing.
Assume fixed uniform group size \(s\), fixed reliability
parameters \((\varepsilon,\delta)\), and a fixed cost ratio
\[
  \rho:=c_I/c_B\in(0,s).
\]
Write \(N^*:=N^*(\varepsilon,\delta)\).
Allow the integer payload length \(c_B=c_B(d)\) to grow so
that the MDS alphabet condition of
Theorem~\ref{thm:first-order-overhead-law} holds.

Use the optimal payload-recovery benchmark
\[
  \bar\sigma_{\mathrm{code}}
  (D_\cup,\varepsilon,\delta)
\]
of Definition~\ref{def:coded-benchmark-task}.
By Theorem~\ref{thm:first-order-overhead-law},
\[
  \bar\sigma_{\mathrm{code}}
  (D_\cup,\varepsilon,\delta)
  =
  \varepsilon d c_B(1+o(1)).
\]

If
\[
  d-N^*\le Ms\le d,
\]
then Theorem~\ref{thm:shared-module-exact-optimality}
gives
\[
\begin{aligned}
  \sigma_{\mathrm{sem}}^{*,\mathrm{joint}}
  (W,\varepsilon,\delta)
  &=
  \frac{d-N^*}{s}c_I+O(c_B)\\
  &=
  \frac{\rho}{s}d c_B(1+o(1)).
\end{aligned}
\]
Consequently,
\[
  \frac{
    \sigma_{\mathrm{sem}}^{*,\mathrm{joint}}
    (W,\varepsilon,\delta)
  }{
    \bar\sigma_{\mathrm{code}}
    (D_\cup,\varepsilon,\delta)
  }
  \longrightarrow
  \frac{\rho}{s\varepsilon}.
\]

If instead \(Ms=o(d)\), at most \(o(d)\) leaves belong
to modules.
The leaf-only upper bound and the required completion
of the remaining leaves give
\[
  \sigma_{\mathrm{sem}}^{*,\mathrm{joint}}
  (W,\varepsilon,\delta)
  =
  d c_B(1+o(1)),
\]
and the ratio tends to \(1/\varepsilon\).
Figure~\ref{fig:storage-ratio} illustrates these
first-order limits.

These ratios compare transparent query recovery with
full payload recovery.
They do not identify the optimum for coded query
recovery.
\end{remark}

\begin{figure}[ht]
\centering
\begin{adjustbox}{max width=\linewidth}
\begin{tikzpicture}[ppi,x=0.026cm,y=4.2cm]
\draw[->] (0,0)--(425,0);
\node[note] at (212,-0.14) {$n=|D_\cup|$};
\draw[->] (0,0)--(0,1.12) node[above]{storage$/(n\,c_B)$};
\foreach \xx in {100,200,300,400} \draw (\xx,0)--(\xx,-0.02)
  node[below]{\scriptsize\xx};
\foreach \yy in {0.1,0.25,0.5,1} \draw (0,\yy)--(-5,\yy)
  node[left]{\scriptsize\yy};
\draw[very thick,black!70,domain=10:400,samples=80,smooth,variable=\nn]
  plot (\nn,{(\nn-3)/\nn}) node[right,note]{leaf-only};
\draw[very thick,orange!80!black,domain=10:400,samples=80,smooth,
  variable=\nn] plot (\nn,{(\nn-3)/(4*\nn)})
  node[right,note]{modules, \(s=4\)};
\draw[very thick,blue!70!black,domain=10:400,samples=80,smooth,
  variable=\nn] plot (\nn,{0.1+0.1573/sqrt(\nn)})
  node[right,note,yshift=3pt]{MDS normal approximation};
\draw[very thick,green!50!black,domain=10:400,samples=80,smooth,
  variable=\nn] plot (\nn,{(\nn-3)/(16*\nn)})
  node[right,note,yshift=-3pt]{modules, \(s=16\)};
\node[note,anchor=west] at (140,0.85)
  {\(\varepsilon=0.1,\ \delta=0.3\ (N^*=3),\ \rho=c_I/c_B=1\)};
\end{tikzpicture}
\end{adjustbox}
\caption{Asymptotic storage scales, normalized by \(n\,c_B\), for
\(\varepsilon=0.1\), \(\delta=0.3\), and \(\rho=c_I/c_B=1\).
The leaf-only optimum tends to \(1\), the optimal payload-recovery
benchmark tends to \(\varepsilon\)
(Theorem~\ref{thm:first-order-overhead-law}), and module-rich instances
with fixed group size \(s\) have limiting normalized cost \(\rho/s\)
(Remark~\ref{rem:asymptotic-overhead-alpha}).  The blue curve is the
normal approximation
\(\varepsilon+\Phi^{-1}(1-\delta)\sqrt{\varepsilon(1-\varepsilon)/n}\)
of \(r^*_{\mathrm{MDS}}/n\).  The module curves are continuous
asymptotic proxies \(\rho(n-N^*)/(sn)\), not exact finite-instance
optima; they suppress integer module counts, residual leaf completion,
and group-size feasibility.  The limit for \(s=16\) lies below the
payload-recovery benchmark because transparent query recovery is a
weaker task than recovering all premise payloads.}
\label{fig:storage-ratio}
\end{figure}

\section{Transition Details for the Treewidth Dynamic Program}
\label{app:treewidth-transitions}

We use the states and invariant of
Theorem~\ref{thm:bounded-treewidth-fpt}.
The following equations supplement the join transition
given there.
At an introduce-vertex node with child \(t'\) and new vertex \(v\),
set
\[
  \alpha_v(\sigma)
  :=
  \mathbf 1\{v\in V_0,\ \sigma(v)=\mathrm R\},
  \qquad
  \beta_v(\sigma)
  :=
  \begin{cases}
    c(v), & \sigma(v)=\mathrm S,\\
    0, & \text{otherwise}
  \end{cases}
\]
(the state \(\mathrm S\) is forbidden for \(q\), so its cost is never
evaluated), and
\[
  \Phi_t(\sigma,k)
  =
  \Phi_{t'}\bigl(\sigma|_{X_{t'}},
             k-\alpha_v(\sigma)\bigr)
  +\beta_v(\sigma),
\]
with out-of-range counters interpreted as infeasible.
At an introduce-edge node for \(u\to v\), discard precisely the
entries with \(\sigma(u)=\mathrm O\) and \(\sigma(v)=\mathrm R\); all
other entries are copied from the child.
At a forget node removing \(v\),
\[
  \Phi_t(\sigma,k)
  =
  \min_{a\in\{\mathrm S,\mathrm R,\mathrm O\}}
  \Phi_{t'}(\sigma\cup\{v\mapsto a\},k),
\]
subject to the target and admissibility constraints; no contribution is added or
removed, since it was included when the vertex entered the partial
subgraph.
The join-node equation, including the bag subtractions \(a(\sigma)\)
and \(b(\sigma)\), is displayed in the main text.


\section*{Declaration of competing interest}
The author declares no competing interests relevant to the content of
this article.

\section*{Data availability}
No datasets were generated or analyzed in this study.

\bibliographystyle{elsarticle-harv}
\bibliography{refs}

\end{document}